\documentclass[11pt,a4paper]{article}
\pdfoutput=1
\usepackage{hyperref}
\usepackage{fullpage}
\usepackage[framemethod=tikz]{mdframed}
\usepackage{amssymb}
\usepackage{mathtools}
\usepackage{amsthm}
\usepackage{eucal}
\usepackage{stmaryrd}
\usepackage{dsfont}
\usepackage[T1]{fontenc}
\usepackage{lmodern}
\usepackage[stretch=10,shrink=10]{microtype}
\usepackage[capitalise]{cleveref}
\usepackage{color,soul}

\allowdisplaybreaks[1]

\def\A{\mathcal{A}}
\def\B{\mathcal{B}}
\def\C{\mathcal{C}}
\def\D{\mathcal{D}}

\def\G{\mathcal{G}}
\def\H{\mathcal{H}}

\def\L{\mathcal{L}}
\def\M{\mathcal{M}}

\def\P{\mathcal{P}}
\def\Q{\mathcal{Q}}
\def\R{\mathcal{R}}

\def\T{\mathcal{T}}

\def\X{\mathcal{X}}
\def\Y{\mathcal{Y}}

\def\g{\textbf{g}}

\theoremstyle{plain}
\newtheorem{theorem}{Theorem}[section]
\newtheorem{lemma}[theorem]{Lemma}
\newtheorem{prop}[theorem]{Proposition}

\theoremstyle{definition}
\newtheorem{definition}[theorem]{Definition}
\newtheorem{claim}[theorem]{Claim}

\newtheorem{remark}[theorem]{Remark}
\newtheorem{fact}[theorem]{Fact}

\newtheorem*{theorem*}{Theorem}

\newcommand {\minusspace} {\: \! \!}

\newcommand {\Fn} [2] {\ensuremath{ #1 \minusspace \Br{ #2 } }}
\newcommand{\reals}{{\mathbb R}}
\newcommand{\complex}{{\mathbb C}}

\newcommand {\set} [1] {\ensuremath{ \left\lbrace #1 \right\rbrace }}

\newcommand {\br} [1] {\ensuremath{ \left( #1 \right) }}
\newcommand {\Br} [1] {\ensuremath{ \left[ #1 \right] }}

\newcommand {\norm} [1] {\ensuremath{ \left\| #1 \right\| }}
\newcommand {\normsub} [2] {\ensuremath{ \norm{#1}_{#2} }}
\newcommand {\onenorm} [1] {\normsub{#1}{1}}
\newcommand {\twonorm} [1] {\normsub{#1}{2}}
\newcommand {\abs} [1] {\ensuremath{ \left| #1 \right| }}

\newcommand {\bra} [1] {\ensuremath{ \left\langle #1 \right| }}
\newcommand {\ket} [1] {\ensuremath{ \left| #1 \right\rangle }}
\newcommand {\ketbratwo} [2] {\ensuremath{ \left| #1 \middle\rangle \middle\langle #2 \right| }}
\newcommand {\ketbra} [1] {\ketbratwo{#1}{#1}}

\newcommand {\innerproduct} [2] {\ensuremath{\left \langle #1 , #2 \right \rangle}}
\newcommand{\nnorm}[1]{{\left\vert\kern-0.25ex\left\vert\kern-0.25ex\left\vert #1
    \right\vert\kern-0.25ex\right\vert\kern-0.25ex\right\vert}}

\newcommand {\defeq} {\ensuremath{ \stackrel{\mathrm{def}}{=} }}

\newcommand {\prob} [1] {\Fn{\Pr}{#1}}

\DeclareMathOperator*{\bigE}{\mathbb{E}}
\newcommand {\expec} [2] {\Fn{\bigE_{\substack{#1}}}{#2}}
\newcommand {\var} [1] {\Fn{\mathrm{Var}}{#1}}
\newcommand {\influence} {\ensuremath{ \mathrm{Inf} }}
\newcommand{\supp}[1]{\mathrm{supp}\br{#1}}

\newcommand {\Tr} {\ensuremath{ \mathrm{Tr} }}

\newcommand {\id} {\ensuremath{\mathds{1}}}

\newcommand{\conjugate}[1]{\overline{#1}}

\newcommand{\anticommutator}[2]{\set{#1,#2}}

\newcommand {\email} [1] {\href{mailto:#1}{\texttt{#1}}}

\newcommand {\mytitle} {A doubly exponential upper bound on noisy {EPR} states for binary games}

\newcommand {\Penghui} {Penghui Yao}
\newcommand{\SKL}{State Key Laboratory for Novel Software Technology}
\newcommand{\NJU}{ Nanjing University}

\newcommand {\authorblock} [3] {
	\begin{minipage}[t]{0.3\linewidth}
		\centering
		{#1}\\[0.8ex]
		{\footnotesize {#2}\\[-0.7ex]
			\email{#3}}
	\end{minipage}\vspace{1ex}
}

\hypersetup{
	pdfstartview={FitH},
	pdfdisplaydoctitle={true},
	breaklinks={true},
	bookmarksopen={true},
	bookmarksnumbered={false},
	pdftitle={\mytitle},
}

\begin{document}
	
	\begin{titlepage}
		\title{\textbf{\mytitle}\\[2ex]}
		
		\author{
			%    \authorblock{\Pnurag}{\CQT, \NUS}{a0109169@nus.edu.sg}
			%	\authorblock{\Rahul}{\CQTCS, \NUS}{rahul@comp.nus.edu.sg}
			%	\authorblock{\Priyanka}{\CQT, \NUS}{a0109168@nus.edu.sg}\\
			%    \authorblock{\Pla}{\IQC}{ashayeghi@uwaterloo.ca}
  \authorblock{\Penghui}{\SKL, \NJU}{pyao@nju.edu.cn}
		}

		\clearpage\maketitle

\thispagestyle{empty}

\abstract{
This paper initiates the study of a class of entangled games,  mono-state games, denoted by $(G,\psi)$, where $G$ is a two-player one-round game and $\psi$ is a bipartite state independent of the game $G$. In the mono-state game $\br{G,\psi}$, the players are only allowed to share arbitrary copies of $\psi$. This paper provides a doubly exponential upper bound on the copies of $\psi$ for the players to approximate the value of the game to an arbitrarily small constant precision for any mono-state binary game $(G,\psi)$, if $\psi$ is a noisy EPR state, which is a two-qubit state with completely mixed states as marginals and maximal correlation less than $1$. In particular, it includes $(1-\epsilon)\ketbra{\Psi}+\epsilon\frac{\id_2}{2}\otimes\frac{\id_2}{2}$, an EPR state with an arbitrary depolarizing noise $\epsilon>0$. The structure of the proofs is built on the recent framework about the decidability of the non-interactive simulations of joint distributions~\cite{7782969,doi:10.1137/1.9781611975031.174,Ghazi:2018:DRP:3235586.3235614} with significant extension, which is completely different from all previous optimization-based approaches~\cite{Cleve:2004:CLN:1009378.1009560,doi:10.1137/090772885,Navascu_s_2008} or "Tsirelson's problem"-based approaches~\cite{doi:10.1142/S0129055X12500122,slofra:2016,slofstra_2019}.
  The paper develops a series of new techniques about the Fourier analysis on matrix spaces and proves a quantum invariance principle and a  hypercontractive inequality for random operators. This novel approach provides a new angle to study the decidability of the complexity class MIP$^*$, a longstanding open problem in quantum complexity theory.}

\end{titlepage}

		\thispagestyle{empty}
		
\tableofcontents

\newpage

\setcounter{page}{1}
\section{Introduction}

The concept of {\em interactive proof systems} is nowadays fundamental to the theory of computing. It was first proposed by Babai~\cite{Babai:1985:TGT:22145.22192} and Goldwasser, Micali and Rackoff~\cite{Goldwasser:1985:KCI:22145.22178} and later extended to the multi-prover setting in~\cite{Ben-Or:1988:MIP:62212.62223}. The study of the interactive proof systems from different lens is at the heart of theory of computing, including the elegant characterizations $\mathrm{IP}=\mathrm{PSPACE}$~\cite{Shamir:1992:IP:146585.146609,Shen:1992:ISS:146585.146613} for single-prover interactive proof systems and  $\mathrm{MIP}=\mathrm{NEXP}$~\cite{Babai1991} for multiprover interactive proof systems. The later one has further led to the celebrated PCP theorem~\cite{Arora:1998:PVH:278298.278306, Arora:1998:PCP:273865.273901}.

The study on the power of interactive proof systems in the context of quantum computing also has a rich history. The model of single-prover quantum interactive proof systems was first studied by Watrous~\cite{Watrous:2003:PCQ:763677.763679}, followed by a series of work~\cite{Kitaev:2000:PAE:335305.335387,Marriott:2005:QAG:1391802.1391807,Gutoski:2007:TGT:1250790.1250873,Jain:2009:TQI:1747597.1748069}, finally led to the seminar result $\mathrm{QIP}=\mathrm{PSPACE}$~\cite{Jain:2011:QP:2049697.2049704}. Quantum multiprover interactive proof systems are more complicated. A key assumption on the classical multiprover interactive proof systems is that the provers are not allowed to communicate, which means that their only distributed resource is the shared randomness. In quantum multiprover interactive proof systems, this assumption is relaxed and allow the provers to share {\em entanglement} with the corresponding complexity class $\mathrm{MIP}^*$~\cite{Cleve:2004:CLN:1009378.1009560}. Surprisingly, understanding the power of $\mathrm{MIP}^*$ turns out to be extremely difficult. A trivial lower bound on $\mathrm{MIP}^*$ is $\mathrm{IP}$, or equivalently $\mathrm{PSPACE}$, which can be easily seen by ignoring all but one provers. By extending the techniques in~\cite{Babai1991} to the quantum setting, Ito and Vidick proved the containment of $\mathrm{NEXP}$ in $\mathrm{MIP}^*$~\cite{Ito:2012:MIP:2417500.2417883}. If the gap between the completeness and the soundness is exponentially small, then the lower bound can be improved to $\mathrm{QMA}_{\mathrm{EXP}}$, a quantum computational complexity class analog to $\mathrm{NEXP}$, and further to $\mathrm{NEEXP}$, the class of nondeterministic double-exponential time by Ji~\cite{Ji:2016:CVQ:2897518.2897634,Ji:2017:CQM:3055399.3055441}. Surprisingly, in a very recent work~\cite{NWright:2019}, Natarajan and Wright proved that $\mathrm{NEEXP}\subseteq\mathrm{MIP}^*$. Namely, the class still contains $\mathrm{NEEXP}$ even if the gap between the completeness and the soundness is constant. In contrast, little is known about the upper bound on $\mathrm{MIP}^*$. In his breakthrough results, Slofstra proved that it is {\em undecidable} to determine whether a multiprover interactive proof system has an entangled strategy that is accepted with probability $1$~\cite{slofra:2016,slofstra_2019}. His proof was later simplified by Dykema, Paulsen, and Prakash in~\cite{Dykema2019} and Fitzsimons, Ji, Vidick and Yuen in~\cite{FJVYuen:2019}.
	
This paper concerns {\em two-player one-round games}, a core model precisely capturing the power of multiprover interactive proof systems. A two-player one-round game $G=\br{\X,\Y,\A,\B,\mu,V}$, where $\X,\Y,\A,\B$ are finite sets, $\mu$ is a distribution over $\X\times\Y$ and $V:\X\times\Y\times\A\times\B\rightarrow\set{0,1}$ is a predicate and all of these are public, is run by three parties: a  "referee" and two non-communicating players. The referee samples a pair of questions $\br{x,y}$ according to $\mu$, and sends $x$ and $y$ to the two players, separately. The two players have to provide an answer each to the referee from $\A$ and $\B$, respectively, say $\br{a,b}$. The referee accepts the answers he receives if and only if $V\br{x,y,a,b}=1$. The only restriction to the players' strategies is that the players are not allowed to exchange any information once the game has started. In the classical setting, the value of the game $\omega\br{G}$, the highest probability that the referee accept the game, is
\[\omega\br{G}=\max_{h_A:\X\rightarrow\A\atop h_B:\Y\rightarrow\B}\sum_{xy}\mu\br{x,y}V(x,y,h_A(x),h_B(y)).\]
It is $\mathrm{NP}$-hard to approximate $\omega(G)$ within a multiplicative constant thanks to the PCP theorem~\cite{Arora:1998:PVH:278298.278306, Arora:1998:PCP:273865.273901}. The {\em entangled games}, which are same as the classical games except that the players are allowed to share arbitrary entangled states before they receive the questions, were first introduced by Cleve, Hoyer, Toner and Watrous~\cite{Cleve:2004:CLN:1009378.1009560} with the {\em entangled value} of a game, which is the highest probability that the referee accepts in a game when the players share entanglement, denoted by $\omega^*\br{G}$. It can be expressed as
\begin{equation}\label{eqn:omegastarG}
  \omega^*\br{G}=\lim_{n\rightarrow\infty}\max_{\psi_{AB}\in\D_n\atop \set{P^x_a}_{x,a},\set{Q^y_b}_{y,b}}\sum_{xy}\mu\br{x,y}\sum_{ab}V(x,y,a,b)\Tr\br{P^x_a\otimes Q^y_b}\psi_{AB},
\end{equation}
where $\D_n$ is the set of $n$-dimensional density operators, $\set{P^x_a}_a$ and $\set{Q^y_b}_b$ are POVM for any $x
\in\X,y\in\Y$, respectively. Namely, $\sum_{a\in\A}P^x_a=\id$, $\sum_{b\in\B}Q^y_b=\id$, $P^x_a\geq 0$ and  $Q^y_b\geq 0$.

In~\cite{Cleve:2004:CLN:1009378.1009560} Cleve et.al. discovered that the model of entangled games gave a re-interpretation of {Bell's inequalities}~\cite{PhysicsPhysiqueFizika.1.195} and an equivalent formulation CHSH games~\cite{PhysRevLett.23.880}, a central role in quantum mechanics from all aspects. The CHSH game is a simple two-player one-round game and the violation of Bell's inequalities by quantum mechanics implies that the classical value of CHSH game is strictly smaller than its entangled value. A large body of the subsequent work has been devoted to boost the gap between $\omega(G)$ and $\omega^*(G)$ and now we know of games of which $\omega^*(G)=1$ while $\omega^*(G)$ can be arbitrary small~\cite{doi:10.1137/S0097539795280895,aravind:2002}. However,the complexity of computing $\omega^*(G)$ is much more involved same as $\mathrm{MIP}^*$. It was shown in~\cite{Kempe:2008:EGH:1470582.1470594,Ito:2009:OTO:1602931.1603187} that approximating $\omega^*(G)$ to a inverse-polynomial precision is $\mathrm{NP}$-hard. Later, Vidick proved that $\omega^*(G)$ of three players is $\mathrm{NP}$-hard to approximate to a constant factor~\cite{doi:10.1137/140956622}. Recently, Ji proved that it is $\mathrm{QMA}_{\mathrm{EXP}}$-hard to approximate $\omega^*(G)$ of multiplayer games to a inverse-exponential precision~\cite{Ji:2016:CVQ:2897518.2897634}, which is further improved to be $\mathrm{NEXP}$-hard~\cite{Ji:2017:CQM:3055399.3055441,Ito:2012:MIP:2417500.2417883}. Very recently,  Natarajan and Vidick have proved that it is QMA-hard to approximate $\omega^*(G)$ to a constant precision under a randomized reduction~\cite{NVidick:2018}. Similar to the complexity class $\mathrm{MIP}^*$, the progress on the upper bound on the complexity of $\omega^*(G)$ is much less. For a few known classes of games, computing $\omega^*(G)$ is easier than computing $\omega(G)$. Cleve et.al. in~\cite{Cleve:2004:CLN:1009378.1009560} gave a polynomial-time algorithm to exactly compute $\omega^*(G)$ of XOR games $G$ building on the work of Tsirelson~\cite{Cirel'son1980}. Kempe, Regev and Toner later present a polynomial-time algorithm for unique games with a factor $6$ approximation to $1-\omega^*(G)$~\cite{doi:10.1137/090772885}. Interestingly, both of the two classes of games are believed to be $\mathrm{NP}$-hard under certain complexity assumptions~\cite{Hastad:2001:OIR:502090.502098,Khot:2002:PUG:509907.510017}. To the best of my knowledge, all the algorithms that compute $\omega^*(G)$ of certain class of games, including those mentioned above, are based on semidefinite programs. In particular, a hierarchy of semidefinite programs was proposed in~\cite{Navascu_s_2008}, whose optimal values converge to $\omega^*(G)$, while the speed of the convergence is unknown. On the other hand, Slofstra's results~\cite{slofra:2016,slofstra_2019} imply that determining whether $\omega^*(G)=1$ is undecidable. Whether approximating $\omega^*(G)$ is decidable is still widely open. The main difficulty in computing $\omega^*(G)$ is that there is no upper bound on the dimension of the preshared entangled states, because if we knew an upper bound, we could approximate the optimal value by using the $\epsilon$-net over all possible strategies and then brute force search. On the other hand, It is known that a positive answer to a so-called "Tsirelson's problem" (see e.g.~\cite{doi:10.1142/S0129055X12500122}) implies the existence of an algorithm approximating $\omega^*(G)$ of any entangled game, while Tsirelson's problem is equivalent to Conne's Embedding Conjecture~\cite{10.2307/1971057}, a longstanding open problem in functional analysis~\cite{doi:10.1063/1.3514538,Ozawa2013}.

\subsection*{Our contribution}

This paper initiates the study of {\em mono-state games}, a new class of entangled games denoted by $\br{G,\psi}$, where $G$ is a two-player one-round game and $\psi$ is a bipartite state (possibly mixed) independent of the game $G$. In the mono-state game $(G,\psi)$, the players are only allowed to share arbitrary copies the state $\psi$. The value of the game, denoted by $\omega^*\br{G,\psi}$, can be expressed as
\begin{equation}\label{eqn:omegastarGpsi}
  \omega^*\br{G,\psi}=\lim_{n\rightarrow\infty}\max_{ \set{P^x_a}_{x,a},\set{Q^y_b}_{y,b}}\sum_{xy}\mu\br{x,y}\sum_{ab}V(x,y,a,b)\Tr~\br{P^x_a\otimes Q^y_b}\psi^{\otimes n},
\end{equation}
where $\set{P^x_a}_a$ and $\set{Q^y_b}_b$ are POVM for any $x
\in\X$ and $y\in\Y$, respectively.

  To the best of my knowledge, the decidability of mono-state games has not been studied yet. It is easy to see that the highest probability that the referee accepts is equal to the classical value if $\psi$ is a separable state. However, the situation is more involved when $\psi$ is entangled as the amount entanglement increases and tends to infinity when having more copies of $\psi$, which is potentially helpful for the referee to accept the game with higher probability. Indeed, Man\v{c}inska and Vidick in~\cite{Mancinska:2015} constructed a mono-state game where the referee accepts with probability tends to $1$ when the copies of the shared states tend to infinity, while sharing any bounded dimensional entanglement, the probability that the referee accepts is bounded away from $1$.

  This paper takes a step towards understanding the decidability of the mono-state games. The following is an informal statement of the main result.

\begin{theorem*}[Main result, informal]
Given a mono-state binary game $\br{G,\psi}$, where $\psi$ is a noisy EPR state and a parameter $\epsilon$,  there exists an explicitly computable $D$ such that it suffices for the players to share $D$ copies of $\psi$ to achieve the probability of winning at least $\omega\br{G,\psi}-\epsilon$. Hence, the game $\br{G,\psi}$ is decidable.
\end{theorem*}

The class of noisy EPR states will be defined later, which includes $(1-\epsilon)\ketbra{\Psi}+\epsilon\frac{\id_2}{2}\otimes\frac{\id_2}{2}$, an EPR state with arbitrary small $\epsilon>0$ depolarizing noise. All the previous works studying the upper bound on the complexity of entangled games are either via convex optimiation~\cite{Cleve:2004:CLN:1009378.1009560,doi:10.1137/090772885,Navascu_s_2008} or based on Tsirelson's problem~\cite{doi:10.1142/S0129055X12500122,slofra:2016,slofstra_2019}. This paper generalizes the framework of Fourier analysis on the Boolean functions, a well studied and fruitful topic in theoretical computer science~\cite{Odonnell08}, to matrix spaces, and reduces the problem to {\em quantum non-interactive simulations of joint distributions}. It provides a new angle and novel tools to study the entangled two-prover one-round games and the complexity class $\mathrm{MIP}^*$. Moreover, a series of results about the Fourier analysis on matrix spaces have been developed in this paper, which may be useful for other topics such as quantum property testing, quantum machine learning, etc.

Non-interactive simulations of joint distributions is a fundamental problem in information theory and communication complexity. Consider two non-communicating players Alice and Bob. Suppose they are provided a sequence of independent samples $\br{x_1,y_1},\br{x_2,y_2},\ldots$ from a joint distribution $\mu$ on $\X\times\Y$, where Alice observes $x_1,x_2,\ldots$ and Bob observes $y_1,y_2,\ldots$. Without communicating with each other, what joint distribution $\nu$ can Alice and Bob jointly simulate? The research on this problem dates back to the classic works by G\'acs and K\"orner~\cite{Gacs:1973} and Wyner~\cite{Wyner:1975:CIT:2263311.2268812} and Witsenhausen~\cite{doi:10.1137/0128010}, followed by fruitful subsequent work (see, for example,~\cite{7452414} and the references therein). Recently, Ghazi, Kamath and Sudan in~\cite{7782969} studied the decidability of the non-interactive simulations of joint joint distributions by introducing a   made partial progress by introducing a framework built on the theory of Fourier analysis on discrete functions and Hermite analysis on Gaussian space~\cite{MosselOdonnell:2010,Mossel:2010,Odonnell08}. With such a framework, the decidability is settled in subsequent works~\cite{doi:10.1137/1.9781611975031.174,Ghazi:2018:DRP:3235586.3235614}.

In quantum universe, it is natural to consider the {\em non-interactive simulations of quantum states}, which is also named the {\em local state transformations}. Suppose the two non-communicating players Alice and Bob are provided arbitrary copies of bipartite quantum states $\psi_{AB}$. Without communicating with each other, what bipartite quantum state $\phi_{AB}$ can Alice and Bob jointly create? Delgosha and Beigi first studied this problem and gave a criterion for the impossibility of local state transformation of $\psi_{AB}$ to $\phi_{AB}$ exactly~\cite{Delgosha2014}. Other than this result, not much about this problem is known. The proofs of the decidability of the non-interactive simulations of joint distributions in~\cite{7782969,doi:10.1137/1.9781611975031.174,Ghazi:2018:DRP:3235586.3235614} heavily use Fourier analysis and Hermite analysis, which has been intensively studied and has fruitful applications in theoretical computer science~\cite{Odonnell08}. The hypercontractive inequality, a key component in Fourier analysis and Hermite analysis, has also been extended to the quantum setting from various aspects resulting several interesting applications~\cite{4690981,doi:10.1063/1.4769269,Temme_2014,King2014,Delgosha2014,doi:10.1063/1.4933219}.
However, the understanding of the Fourier analysis on matrix spaces, in particular, the Fourier analysis on quantum operations is much less compared with the one on Boolean functions or continuous functions. This paper essentially resolves the decidability of local state transformation when $\psi_{AB}$ is a noisy EPR state and $\phi_{AB}$ is a classical two-bit distribution, by systematically developing the Fourier analysis on random matrix spaces. In particular, this paper proves a {\em quantum invariance principle} and a {\em hypercontractive inequality for random operators}, successfully generalizing the framework established in~\cite{7782969,doi:10.1137/1.9781611975031.174,Ghazi:2018:DRP:3235586.3235614} to the quantum setting. The tools developed in this paper are interesting on their own right and are believed to have further  applications.

\subsection{Proof Overview}

To explain the ideas of the proof in high level, we start with a pair of measurements performed by Alice and Bob, denoted by  $\br{\set{P,\id-P},\set{Q,\id-Q}}$, respectively, where $P, Q\in\H_2^{\otimes n}$ and $0\leq P,Q\leq\id$ and $n$ copies of  noisy EPR states $\psi_{AB}^{\otimes n}$. The proof is to construct a universal bound $D$, which is independent of the measurements, and a transformation $f_n,g_n:\H_2^{\otimes n}\rightarrow\H_2^{n_0}$, such that the requirements in Figure~\ref{fig:requirements} are satisfied.

\begin{figure}
\begin{mdframed}
\textbf{Requirements.}
\begin{enumerate}
	\item $0\leq f_n\br{P}\leq\id~\mbox{and}~0\leq g_n\br{Q}\leq\id;$
	\item $\Tr~P\psi_A^{\otimes n}\approx\Tr~f_n\br{P}\psi_A^{\otimes D} ~\mbox{and}~ \Tr~Q_n\psi_B^{\otimes n}\approx\Tr~g_n\br{Q}\psi_B^{\otimes D};$
	\item $\Tr~\br{P\otimes Q}\psi_{AB}^{\otimes n}\approx\Tr~\br{f_n\br{P}\otimes g_n\br{Q}}\psi_{AB}^{\otimes D}.$
\end{enumerate}
\end{mdframed}
\caption{Requirements}\label{fig:requirements}
\end{figure}

The first item implies that $\set{f_n\br{P},\id-f_n\br{P}}$ and $\set{g_n\br{Q},\id-g_n\br{Q}}$ are both valid measurements. The second item imples that the probability that Alice outputs $1$ is almost unchanged under the transformation $f_n$.  Same for the probability that Bob outputs $1$. The last item implies that the probability that both Alice and Bob output $1$ is almost unchanged. As Alice's and Bob's outputs are both binary, it concludes that the distribution of the joint output is almost unchanged.

The construction of the transformations $f_n$ and $g_n$ are based on the framework introduced in~\cite{7782969}, which, in turn, is built on the results in Fourier analysis and Hermite analysis developed in~\cite{MosselOdonnell:2010,Mossel:2010}. Analogously, we choose an orthornormal basis in $\M_2$, which is of dimension $4$, and apply the theory of Fourier analysis to the expansions of $P$ and $Q$ on this basis. Let $\B=\set{\B_0,\B_1,\B_2,\B_3}$ be an orthonormal basis in $\M_2$ with all elements being Hermitian and $\B_0=\id$, whose existence is guaranteed by Lemma~\ref{lem:paulibasis}. It is easy to verify that the set $\set{\B_{\sigma}:\sigma\in\set{0,1,2,3}^n}$, where $\B_{\sigma}\defeq\B_{\sigma_1}\otimes\B_{\sigma_2}\otimes\ldots\otimes\B_{\sigma_n}$, forms an orthonormal basis in $\M_2^{\otimes n}$. Let $\A=\set{\A_0,\A_1,\A_2,\A_3}$ and $\B=\set{\B_0,\B_1,\B_2,\B_3}$  be two orthonormal basis in $\M_2$, which are speficifed later. The expansions of $P$ and $Q$ on the basis $\A$ and $\B$ can be expressed as
\[P=\sum_{\sigma\in\set{0,1,2,3}^n}\widehat{P}\br{\sigma}\A_{\sigma}~\mbox{and}~Q=\sum_{\sigma\in\set{0,1,2,3}^n}\widehat{Q}\br{\sigma}\B_{\sigma},\]
which are considered to be the Fourier expansions of $P$ and $Q$, respectively. The construction of the transformations $f_n$ and $g_n$ consists of the following several steps, which are summarized in the Figure~\ref{fig:construction}, where $L^2\br{\H_2^{\otimes h},\gamma_t}$ is a set of random operators defined in Subsection~\ref{subsec:randomoperators}.

\begin{center}
\end{center}
\begin{figure}
\tikzset{every picture/.style={line width=0.75pt}} %set default line width to 0.75pt

\begin{tikzpicture}[x=0.75pt,y=0.75pt,yscale=-1,xscale=1]
%uncomment if require: \path (0,945); %set diagram left start at 0, and has height of 945

%Shape: Rectangle [id:dp8296054240248445]
\draw   (221.21,50.22) -- (539.5,50.22) -- (539.5,100.67) -- (221.21,100.67) -- cycle ;
%Straight Lines [id:da6836064098073678]
\draw    (380.3,51.34) -- (379.6,101) ;

%Straight Lines [id:da9598113801081944]
\draw    (305.97,28.33) -- (305.92,49.33) ;
\draw [shift={(305.92,51.33)}, rotate = 270.13] [color={rgb, 255:red, 0; green, 0; blue, 0 }  ][line width=0.75]    (10.93,-3.29) .. controls (6.95,-1.4) and (3.31,-0.3) .. (0,0) .. controls (3.31,0.3) and (6.95,1.4) .. (10.93,3.29)   ;

%Straight Lines [id:da3951670667346401]
\draw    (460.1,27) -- (460.1,48.5) ;
\draw [shift={(460.1,50.5)}, rotate = 270] [color={rgb, 255:red, 0; green, 0; blue, 0 }  ][line width=0.75]    (10.93,-3.29) .. controls (6.95,-1.4) and (3.31,-0.3) .. (0,0) .. controls (3.31,0.3) and (6.95,1.4) .. (10.93,3.29)   ;

%Straight Lines [id:da34999940781789807]
\draw    (309.64,102.33) -- (309.36,128.33) ;
\draw [shift={(309.33,130.33)}, rotate = 270.62] [color={rgb, 255:red, 0; green, 0; blue, 0 }  ][line width=0.75]    (10.93,-3.29) .. controls (6.95,-1.4) and (3.31,-0.3) .. (0,0) .. controls (3.31,0.3) and (6.95,1.4) .. (10.93,3.29)   ;

%Straight Lines [id:da2494201958401736]
\draw    (459.97,100.33) -- (459.69,128.33) ;
\draw [shift={(459.67,130.33)}, rotate = 270.58] [color={rgb, 255:red, 0; green, 0; blue, 0 }  ][line width=0.75]    (10.93,-3.29) .. controls (6.95,-1.4) and (3.31,-0.3) .. (0,0) .. controls (3.31,0.3) and (6.95,1.4) .. (10.93,3.29)   ;

%Shape: Rectangle [id:dp45207587854231446]
\draw   (221.21,210.55) -- (538.5,210.55) -- (538.5,260.5) -- (221.21,260.5) -- cycle ;
%Shape: Rectangle [id:dp41412999776281834]
\draw   (220.71,131.22) -- (539,131.22) -- (539,181.67) -- (220.71,181.67) -- cycle ;
%Straight Lines [id:da05771465242130547]
\draw    (460.47,181.33) -- (460.19,209.33) ;
\draw [shift={(460.17,211.33)}, rotate = 270.58] [color={rgb, 255:red, 0; green, 0; blue, 0 }  ][line width=0.75]    (10.93,-3.29) .. controls (6.95,-1.4) and (3.31,-0.3) .. (0,0) .. controls (3.31,0.3) and (6.95,1.4) .. (10.93,3.29)   ;

%Straight Lines [id:da9181378098124038]
\draw    (309.97,181.33) -- (309.69,209.33) ;
\draw [shift={(309.67,211.33)}, rotate = 270.58] [color={rgb, 255:red, 0; green, 0; blue, 0 }  ][line width=0.75]    (10.93,-3.29) .. controls (6.95,-1.4) and (3.31,-0.3) .. (0,0) .. controls (3.31,0.3) and (6.95,1.4) .. (10.93,3.29)   ;

%Straight Lines [id:da8176411481273858]
\draw    (380.2,211.11) -- (380.6,261) ;

%Shape: Rectangle [id:dp9286152410496051]
\draw   (221.71,289.72) -- (540,289.72) -- (540,340.17) -- (221.71,340.17) -- cycle ;
%Straight Lines [id:da7169706443323178]
\draw    (310.47,259.83) -- (310.19,287.83) ;
\draw [shift={(310.17,289.83)}, rotate = 270.58] [color={rgb, 255:red, 0; green, 0; blue, 0 }  ][line width=0.75]    (10.93,-3.29) .. controls (6.95,-1.4) and (3.31,-0.3) .. (0,0) .. controls (3.31,0.3) and (6.95,1.4) .. (10.93,3.29)   ;

%Straight Lines [id:da042636824012835195]
\draw    (460.47,260.33) -- (460.19,288.33) ;
\draw [shift={(460.17,290.33)}, rotate = 270.58] [color={rgb, 255:red, 0; green, 0; blue, 0 }  ][line width=0.75]    (10.93,-3.29) .. controls (6.95,-1.4) and (3.31,-0.3) .. (0,0) .. controls (3.31,0.3) and (6.95,1.4) .. (10.93,3.29)   ;

%Straight Lines [id:da12397647595538008]
\draw    (309.97,339.83) -- (309.69,367.83) ;
\draw [shift={(309.67,369.83)}, rotate = 270.58] [color={rgb, 255:red, 0; green, 0; blue, 0 }  ][line width=0.75]    (10.93,-3.29) .. controls (6.95,-1.4) and (3.31,-0.3) .. (0,0) .. controls (3.31,0.3) and (6.95,1.4) .. (10.93,3.29)   ;

%Straight Lines [id:da06297553209728379]
\draw    (460.47,340.33) -- (460.19,368.33) ;
\draw [shift={(460.17,370.33)}, rotate = 270.58] [color={rgb, 255:red, 0; green, 0; blue, 0 }  ][line width=0.75]    (10.93,-3.29) .. controls (6.95,-1.4) and (3.31,-0.3) .. (0,0) .. controls (3.31,0.3) and (6.95,1.4) .. (10.93,3.29)   ;

%Shape: Rectangle [id:dp7223027897698828]
\draw   (221.21,371.22) -- (539.5,371.22) -- (539.5,421.67) -- (221.21,421.67) -- cycle ;
%Straight Lines [id:da08946548410719513]
\draw    (380.7,371.61) -- (380.01,421.27) ;

%Shape: Rectangle [id:dp017395472911364118]
\draw   (220.87,451.88) -- (539.17,451.88) -- (539.17,502.33) -- (220.87,502.33) -- cycle ;
%Straight Lines [id:da0347290527908799]
\draw    (309.3,421.5) -- (309.02,449.5) ;
\draw [shift={(309,451.5)}, rotate = 270.58] [color={rgb, 255:red, 0; green, 0; blue, 0 }  ][line width=0.75]    (10.93,-3.29) .. controls (6.95,-1.4) and (3.31,-0.3) .. (0,0) .. controls (3.31,0.3) and (6.95,1.4) .. (10.93,3.29)   ;

%Straight Lines [id:da19840516473627967]
\draw    (459.3,421.17) -- (459.02,449.17) ;
\draw [shift={(459,451.17)}, rotate = 270.58] [color={rgb, 255:red, 0; green, 0; blue, 0 }  ][line width=0.75]    (10.93,-3.29) .. controls (6.95,-1.4) and (3.31,-0.3) .. (0,0) .. controls (3.31,0.3) and (6.95,1.4) .. (10.93,3.29)   ;

%Straight Lines [id:da43696451132064085]
\draw    (380.37,452.28) -- (379.68,501.94) ;

%Shape: Rectangle [id:dp4189514413720581]
\draw   (221.54,531.22) -- (539.83,531.22) -- (539.83,581.67) -- (221.54,581.67) -- cycle ;
%Straight Lines [id:da8437784670606463]
\draw    (309.97,502.83) -- (309.69,530.83) ;
\draw [shift={(309.67,532.83)}, rotate = 270.58] [color={rgb, 255:red, 0; green, 0; blue, 0 }  ][line width=0.75]    (10.93,-3.29) .. controls (6.95,-1.4) and (3.31,-0.3) .. (0,0) .. controls (3.31,0.3) and (6.95,1.4) .. (10.93,3.29)   ;

%Straight Lines [id:da2658331721139231]
\draw    (459.97,503.17) -- (459.69,531.17) ;
\draw [shift={(459.67,533.17)}, rotate = 270.58] [color={rgb, 255:red, 0; green, 0; blue, 0 }  ][line width=0.75]    (10.93,-3.29) .. controls (6.95,-1.4) and (3.31,-0.3) .. (0,0) .. controls (3.31,0.3) and (6.95,1.4) .. (10.93,3.29)   ;

%Straight Lines [id:da5545972644705808]
\draw    (461.3,582.5) -- (461.02,610.5) ;
\draw [shift={(461,612.5)}, rotate = 270.58] [color={rgb, 255:red, 0; green, 0; blue, 0 }  ][line width=0.75]    (10.93,-3.29) .. controls (6.95,-1.4) and (3.31,-0.3) .. (0,0) .. controls (3.31,0.3) and (6.95,1.4) .. (10.93,3.29)   ;

%Straight Lines [id:da14144183505254793]
\draw    (309.3,582.83) -- (309.02,610.83) ;
\draw [shift={(309,612.83)}, rotate = 270.58] [color={rgb, 255:red, 0; green, 0; blue, 0 }  ][line width=0.75]    (10.93,-3.29) .. controls (6.95,-1.4) and (3.31,-0.3) .. (0,0) .. controls (3.31,0.3) and (6.95,1.4) .. (10.93,3.29)   ;

%Straight Lines [id:da9738094129930643]
\draw    (381.03,531.61) -- (380.34,581.27) ;

%Shape: Rectangle [id:dp5083828094107354]
\draw   (221.54,611.22) -- (539.83,611.22) -- (539.83,661.67) -- (221.54,661.67) -- cycle ;

%Straight Lines [id:da3801745932218987]
\draw    (310.8,660.33) -- (310.52,688.33) ;
\draw [shift={(310.5,690.33)}, rotate = 270.58] [color={rgb, 255:red, 0; green, 0; blue, 0 }  ][line width=0.75]    (10.93,-3.29) .. controls (6.95,-1.4) and (3.31,-0.3) .. (0,0) .. controls (3.31,0.3) and (6.95,1.4) .. (10.93,3.29)   ;

%Straight Lines [id:da12328662972457205]
\draw    (460.3,661.5) -- (460.02,689.5) ;
\draw [shift={(460,691.5)}, rotate = 270.58] [color={rgb, 255:red, 0; green, 0; blue, 0 }  ][line width=0.75]    (10.93,-3.29) .. controls (6.95,-1.4) and (3.31,-0.3) .. (0,0) .. controls (3.31,0.3) and (6.95,1.4) .. (10.93,3.29)   ;

% Text Node
\draw (150.67,80.5) node  [align=left] {Smooth};
% Text Node
\draw (600,20.67) node  [align=left] {$\H_2^{\otimes n}$\\ $ 0\leq P,Q\leq\id$};
% Text Node
\draw (307.33,20.67) node  [align=left] {$P$};
% Text Node
\draw (460.67,20.67) node  [align=left] {$Q$};
% Text Node
\draw (149.67,160) node  [align=left] {Regularization};
% Text Node
\draw (343.67,116.33) node  [align=left] {$P^{\br{1}}$};
% Text Node
\draw (500.33,116.33) node  [align=left] {$Q^{\br{1}}$};
% Text Node
\draw (600,116.33) node  [align=left] {$\H_2^{\otimes n}$ \\ $0\leq P^{(1)},Q^{(1)}\leq\id$};
% Text Node
\draw (140.17,240) node  [align=left] {Invariance principle};
% Text Node
\draw (139.67,320) node  [align=left] {Dimension reduction};
% Text Node
\draw (568.17,195.67) node  [align=left] {};
% Text Node
\draw (600.17,276) node  [align=left] {$L^2\br{\H_2^{\otimes h},\gamma_{3\br{n-h}}}$};
% Text Node
\draw (340.67,195.67) node  [align=left] {$P^{(1)}$};
% Text Node
\draw (489.17,276) node  [align=left] {$\mathbf{Q}^{(2)}$};
% Text Node
\draw (499.17,195.67) node  [align=left] {$Q^{(1)}$};
% Text Node
\draw (335.67,276) node  [align=left] {$\mathbf{P}^{(2)}$};
% Text Node
\draw (335.67,356.17) node  [align=left] {$\mathbf{P}^{(3)}$};
% Text Node
\draw (493.67,356.17) node  [align=left] {$\mathbf{Q}^{(3)}$};
% Text Node
\draw (149.67,401.5) node  [align=left] {Smooth};
% Text Node
\draw (310,397.33) node  [align=left] {Lemma~\ref{lem:smoothing of strategies}};
% Text Node
\draw (460.67,77.33) node  [align=left] {Lemma~\ref{lem:smoothing of strategies}};
% Text Node
\draw (380.85,158) node  [align=left] {Lemma~\ref{lem:regular}};
% Text Node
\draw (459.33,156.67) node  [align=left] {};
% Text Node
\draw (317.33,236.67) node  [align=left] {Lemma~\ref{lem:jointinvariance}};
% Text Node
\draw (460.67,236.67) node  [align=left] {Lemma~\ref{lem:jointinvariance}};
% Text Node
\draw (308.67,77.33) node  [align=left] {Lemma~\ref{lem:smoothing of strategies}};
% Text Node
\draw (460.67,397.33) node  [align=left] {Lemma~\ref{lem:smoothing of strategies}};
% Text Node
\draw (600.33,356.17) node  [align=left] {$L^2\br{\H_2^{\otimes h},\gamma_{3\br{n-h}}}$};
% Text Node
\draw (151,470.83) node  [align=left] {Multilinearlization};
% Text Node
\draw (494,438) node  [align=left] {$\mathbf{Q}^{(4)}$};
% Text Node
\draw (600,438) node  [align=left] {$L^2\br{\H_2^{\otimes h},\gamma_{n_0}}$};
% Text Node
\draw (340,438) node  [align=left] {$\mathbf{P}^{(4)}$};
% Text Node
\draw (305.33,480.33) node  [align=left] {Lemma~\ref{lem:multiliniearization}};
% Text Node
\draw (460,480.33) node  [align=left] {Lemma~\ref{lem:multiliniearization}};
% Text Node
\draw (132.17,561.33) node  [align=left] {Invariance principle};
% Text Node
\draw (339.67,516.67) node  [align=left] {$\mathbf{P}^{(5)}$};
% Text Node
\draw (499.33,516.67) node  [align=left] {$\mathbf{Q}^{(5)}$};
% Text Node
\draw (600,518.67) node  [align=left] {$L^2\br{\H_2^{\otimes h},\gamma_{n_0t}}$};
% Text Node
\draw (340.33,598) node  [align=left] {$\mathbf{P}^{(6)}$};
% Text Node
\draw (500.33,598) node  [align=left] {$\mathbf{Q}^{(6)}$};
% Text Node
\draw (131.5,639.33) node  [align=left] {Rounding};
% Text Node
\draw (399,640) node  [align=left] {Round to measurement operators};
% Text Node
\draw (309.67,710) node  [align=left] {$\widetilde{P}$};
% Text Node
\draw (461.67,710) node  [align=left] {$\widetilde{Q}$};
% Text Node
\draw (380.85,314.94) node  [align=left] {Lemma~\ref{lem:dimensionreduction}};
% Text Node
\draw (311.33,558.67) node  [align=left] {Lemma~\ref{lem:invariancejointgaussian}};
% Text Node
\draw (459.33,558.67) node  [align=left] {Lemma~\ref{lem:invariancejointgaussian}};
% Text Node
\draw (600,598) node  [align=left] {$\H_2^{\otimes h+n_0t}$};
% Text Node
\draw (600,688) node  [align=left] {$\H_2^{\otimes h+n_0t}$\\$0\leq\widetilde{P},\widetilde{Q}\leq\id$};
\end{tikzpicture}
\caption{Construction of the transformations}\label{fig:construction}
\end{figure}
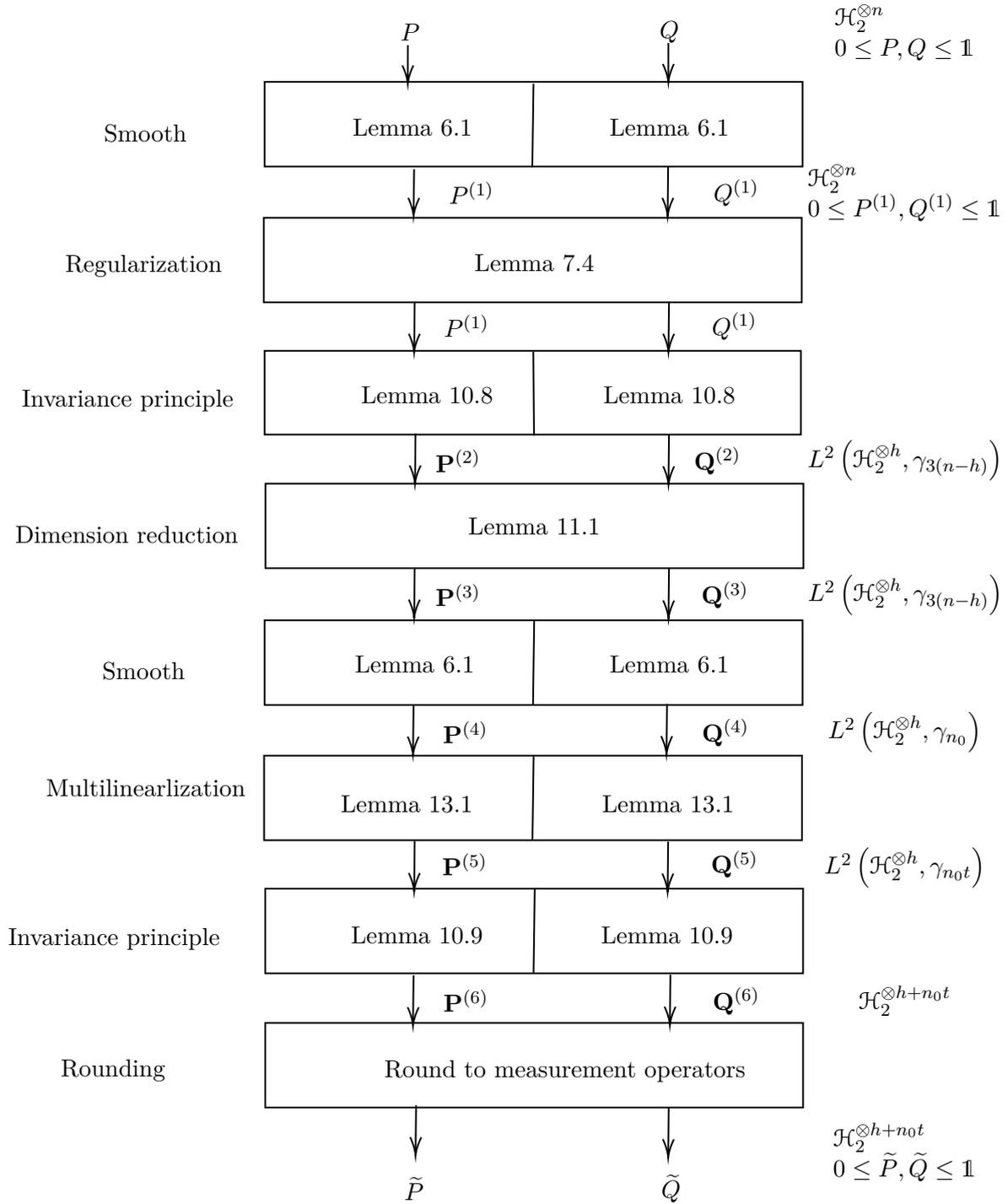
\begin{center}
\end{center}

\begin{itemize}
  \item \textbf{Smoothing operators}. We first convert the operators $\br{P,Q}$ to low-degree operators via smoothing operations. Here we borrow the standard techniques in the Fourier analysis on Boolean functions~\cite{Odonnell08} by applying a {\em noise operator} $\T_{\rho}$ for some constant $\rho\in(0,1)$, whose analog in Gaussian space is also called {\em Ornstein-Uhlenbeck operator}, to both $P$ and $Q$. We obtain

  \[\T_{\rho}P=\sum_{\sigma\in[4]_{\geq 0}^n}\widehat{P}\br{\sigma}\rho^{|\sigma|}\A_{\sigma}~\mbox{and}~\T_{\rho}Q=\sum_{\sigma\in[4]_{\geq 0}^n}\widehat{Q}\rho^{|\sigma|}\br{\sigma}\B_{\sigma},\]
  where $|\sigma|=\abs{\set{i:\sigma_i\neq 0}}$ and $[4]_{\geq 0}=\set{0,1,2,3}$.
  It is not hard to verify that the item 1 and item 2 of the requirements in Figure~\ref{fig:requirements} are satisfied. To meet the item 3, we need the notion of {\em quantum maximal correlation}, which was introduced by Beigi~~\cite{Beigi:2013}, extending the {\em maximal correlation coefficients}~\cite{hirschfeld:1935,Gebelein:1941,Renyi1959} in classical information theory. Note that after smoothing the operators, the weights of high-degree part of both $P$ and $Q$, namely, the part with high $\abs{\sigma}$, decrease exponentially. Thus, both of the operators can be approximated by low-degree operators.

  \item \textbf{Joint regularity}. In this step, we choose a bounded-sized subset of coordinates $H$ such that all the coordinates not in $H$ are low influential. The {\em influence} of a coordinate to a Hermitian operator is given in Definition~\ref{def:influencegaussian}, analogous to the notion of the influence in the analysis of boolean functions~\cite{Odonnell08}, first introduced by Montanaro in~\cite{doi:10.1063/1.4769269}. It, informally speaking, measures how much the coordinate can influence the operator. This step follows by quantizing a well known result in the classical boolean analysis. For any bounded function, the total influence, i.e., the summation of the influences of all coordinates, is upper bounded by the degree of the function. Note that the both operators can be approximated by low-degree operators after the first step.

  \item \textbf{Quantum invariance principle}. In this step, we replace all the low-influential coordinates by correlated Gaussian variables, resulting {\em random operators}. Prior to the replacement, we need to specify the choices of the basis $\A$ and $B$. It can be shown that there exist the basis $\A$ and $\B$ such that
  \[\Tr\br{\A_i\otimes\B_j}\psi_{AB}=\delta_{i,j}c_i,\]
  for $1=c_0\geq c_1\geq c_2\geq c_3\geq 0$, where $\A_0=\B_0=\id$.
  Hence
  \begin{equation}\label{eqn:pqreg}
  	\Tr\br{P\otimes Q}\psi_{AB}^{\otimes n}=\sum_{\sigma\in[4]_{\geq 0}^n}c_{\sigma}\widehat{P}\br{\sigma}\widehat{Q}\br{\sigma},
  \end{equation}
  where $c_{\sigma}=c_{\sigma_1}\cdot c_{\sigma_2}\cdots c_{\sigma_n}$ for any $\sigma\in[4]_{\geq 0}^n$.
  Now we introduce independent joint random variables $\set{\br{\mathbf{g}_{i,j},\mathbf{h}_{i,j}}}_{0\leq i\leq 3,j\notin H}$, where $\mathbf{g}_{0,j}=\mathbf{h}_{0,j}=1$ and $\br{\mathbf{g}_{i,j},\mathbf{h}_{i,j}}\sim\G_{c_i}$ for $i\geq 1,j\notin H$ and $\G_{\rho}$ is a $\rho$-correlated two-dimensional Gaussian distribution. Then we substitute all the basis outside $H$ in the expansion of $P$ and $Q$ by $\br{\mathbf{g}_{i,j},\mathbf{h}_{i,j}}$, accordingly, and obtain the random operators

  \begin{equation}\label{eqn:intro2}
  	\mathbf{P}=\sum_{\sigma}\widehat{P}\br{\sigma}\prod_{i\notin H}\mathbf{g}_{i,\sigma_i}\A_{\sigma_H}~\mbox{and}~\mathbf{Q}=\sum_{\sigma}\widehat{Q}\br{\sigma}\prod_{i\notin H}\mathbf{g}_{i,\sigma_i}\B_{\sigma_H}.
  \end{equation}
	  It is easy to verify that the item 2 and item 3 of the requirements in Figure~\ref{fig:requirements} are satisfied in expectation. However, $\mathbf{P}$ and $\mathbf{Q}$ are in general not measurement operator at all, which means that the item 1 of the requirements is violated. To meet item 1, it suffices to show that both of the random operators are close to measurement operators in $\ell_2$ distance in expectation. To prove it, we establish a hypercontractive inequality for random operators and quantum invariance principle, which will be explained with more details in the next subsection.

  \item \textbf{Dimension reduction}. The random operators in Eq.~\eqref{eqn:intro2} can be written as
  \begin{equation}\label{eqn:intro3}
  	\mathbf{P}=\sum_{\sigma_H}p_{\sigma_H}\br{\vec{\mathbf{g}}}\A_{\sigma_H}~\mbox{and}~\mathbf{Q}=\sum_{\sigma_H}q_{\sigma_H}\br{\vec{\mathbf{h}}}\B_{\sigma_H}
  \end{equation}
  And
  \[\expec{}{\Tr\br{\mathbf{P}\otimes\mathbf{Q}}\psi_{AB}^{\otimes h}}=\sum_{\sigma_H}c_{\sigma_H}\expec{}{p_{\sigma_H}\br{\vec{\mathbf{g}}}q_{\sigma_H}\br{\vec{\mathbf{h}}}}\]
  Applying the {\em dimension reduction} for polynomials of Gaussian variables, proved recently in~\cite{Ghazi:2018:DRP:3235586.3235614} to $\br{p_{\sigma_H},q_{\sigma_H}}$, the number of Gaussian random variables are reduced to a bounded number. Meanwhile, all the items in Figure~\ref{fig:requirements} are still satisfied.

  \item \textbf{Smoothing random operators}. To replace the Gaussian random variables by the operators in $\A$ and $B$, we follow the transformation similar to the ones in the previous steps. Thus, we apply the noise operators again to the both operators to reduce the weight of the high degree parts.

  \item \textbf{Multilinearlization}. Note that the degrees of the functions of the Gaussian variables are unbounded. They may even not be polynomials. Thus, we need the multilinearization lemma in~\cite{Ghazi:2018:DRP:3235586.3235614} to reduce the power of the Gaussian random operators to either $0$ or $1$, i.e, the polyonomials $p_{\sigma_H}$'s and $q_{\sigma_H}$'s in Eq.~\eqref{eqn:intro3} are all multilinear.
  \item \textbf{Quantum invariance principle and rounding}. In the final step, we substitute the Gaussian variables by the proper chosen basis operators and further round both of the operators to the measurement operators. Again, we need to apply a quantum invariance principle to ensure that all the requirements in Figure~\ref{fig:requirements} are satisfied.
\end{itemize}

Note that in mono-state games the players adopt different measurement operators for different inputs. Thus, we may have the issue of the consistency when applying the transformations above to mono-state games. Suppose the players' strategies are $\br{\set{P,\id-P},\set{Q,\id-Q}}$ and $\br{\set{P,\id-P},\set{Q',\id-Q'}}$ for the pairs of questions $(x,y)$ and $\br{x,y'}$, respectively. We need to ensure the resulting measurement operators on Alice's side must be same in the both cases. Namely, the transformation for each player should be independent of the other player. Fortunately, the correlation between the two players only occurs when choosing the set $H$ in the step of regularization, which can be resolved by applying an union bound on the all possible question pairs.

\subsection{Quantum invariance principle and quantum hypercontractive inequality}

The main difficulty in this construction is proving the quantum invariance principle, which is used for the interchanges between the basis operators and the Gaussian random variables. Let's recall the classical invariance principle in~\cite{MosselOdonnell:2010}. Let $f:\set{0,1}^n\rightarrow\reals$ be a low-degree multilinear polynomial with small influence for all coordinates. The invariance principle asserts that
\begin{equation}\label{eqn:inv}
\expec{\mathbf{x}\sim\set{0,1}^n}{\Psi\br{f\br{\mathbf{x}}}}\approx\expec{\mathbf{g}\sim\gamma_n}{\Psi\br{f\br{\mathbf{g}}}}
\end{equation}
for any $\Psi:\reals\rightarrow\reals$ with a constant Lipschitz coefficient.

 It was proved via approximating $\Psi$ by a $\C^{\infty}$ function, say $\Psi'$, and then applying the Taylor expansion to the both sides of Eq.~\eqref{eqn:inv} to reduce the difference between the both sides to the norm of third order terms in the Taylor expansion of $\Psi$. Further applying a hypercontractive inequality for random variables, it can be proved to be small for the functions with all coordinates being low influential.

 Generalizing such a mechanism to quantum operations is not an easy task due to the non-commutativity of the operators. In this paper, we adopt Fr\'echet derivatives, a notion of derivatives in Banach space, for which a similar form of Taylor expansion exists.  The differentiability of the real functions and the one of matrix-functions with respect to Fr\'echet derivatives~\cite{SENDOV2007240} share many properties. We follow the same mechanism by substituting the basis elements with Gaussian variables and obtain {\em random operators}, hybrids of operators and random variables. The Taylor expansions of matrix-valued functions are in general complicated again due to the nature of the non-commutativity. Fortunately, it suffices to prove the quantum invariance principle for $\C^3$ function for our purpose. To prove an operator $P$ is close to a measurement operators, it suffices to show that
 \[\zeta\br{P}\approx 0,\]
 where
 \[\zeta\br{x}=\begin{cases}
 \br{x-1}^2~&\mbox{if $x\geq 1$}\\
 0~&\mbox{if $0\leq x<1$}\\
 x^2~&\mbox{otherwise.}
 \end{cases}\]
 We introduce a $\C^3$ approximation of $\zeta$, denoted by $\zeta_{\lambda}$, and show that $\zeta\br{P}$ is upper bounded by the 3rd order term in the Taylor expansion of $\zeta_{\lambda}$, which is, in turn, upper bounded by the power 3 of its $4$-norm.

 A following difficulty is to prove a hypercontractive inequality for random random operators, which is expected to show that the $4$-th norm of a low-degree random operator can be upper bounded by its $2$-norm. As a random operator is a hybrid of operators and the Gaussian variables, the proof is a delicate combination of the hypercontractive inequality for unital qubit channels due to King~\cite{King2014} and the hypercontractive inequality for Gaussian variables due to Wolff~\cite{PawelWolff2007}.

\subsection{Organization of the paper}
Section~\ref{sec:pre} summarizes some useful definitions and basic facts on quantum mechanics, the analysis on Gaussian spaces, matrices spaces and random operators. The main results and the proofs are stated in Section~\ref{sec:mainresult}. Further work and some open problems are listed in Section~\ref{sec:openproblems}. Section~\ref{sec:markov} gives the definition of quantum maximal correlation, a key concept in this work, with several crucial properties. The smoothing operators lemma, joint regularity lemma, hypercontractive inequality for random operators, quantum invariance principle, dimension reduction for random operators, smoothing random operators and multilinearization for random operators are proved in Sections~\ref{sec:smoothoperators},\ref{sec:jointregular},\ref{sec:hypercontractive},\ref{sec:invariance}, \ref{sec:dimensionreduction},\ref{sec:smoothrandom} and \ref{sec:multilinear}, respectively.  Section~\ref{sec:derivative} introduces the notion of F\'echet derivatives. Some basic facts on F\'echet derivatives are summarized in Appendix~\ref{sec:frechet}. Appendix~\ref{sec:analysis} presents some useful facts on the analytical properties of the functions $f:\reals^n\rightarrow\M_d$.
\subsection*{Acknowledgments}

This work is supported by the National Key R\&D Program of China 2018YFB1003202 and a China Youth 1000-Talent grant and Anhui Initiative in Quantum Information Technologies Grant No. AHY150100. Part of the work was done when the author was a Hartree postdoctoral fellow at QuICS, University of Maryland. The author thanks Thomas Vidick pointing out that a union bound on the question sets was missing in the previous version. The author is grateful to Minglong Qin for his thorough reading and valuable feedback. The author also thanks Hong Zhang for helpful discussion and thanks Pritish Kamath and Ashley Montanaro for the correspondence. The author also thanks the anonymous reviewers' helpful feedback.

	\section{Preliminaries}\label{sec:pre}
	For an integer $n\geq 1$, let $[n]$ and $[n]_{\geq 0}$ represent the sets $\set{1,\ldots, n}$ and $\set{0,\ldots, n-1}$, respectively. Given a finite set $\X$ and a natural number  $k$, let $\X^k$ be the set $\X\times\cdots\times\X$, the Cartesian product of $\X$, $k$ times. Given $a=a_1,\ldots, a_k$ and a set $S\subseteq[k]$, we write $a_S$ to represent the projection of $a$ to the coordinates specified in $S$. For any $i\in[k]$, $a_{<i}$ represents $a_1,\ldots, a_{i-1}$. $a_{\leq i},a_{>i}, a_{\geq i}$ are defined similarly. Let $\mu$ be a probability distribution on $\X$, and $\mu\br{x}$ represent the probability of $x\in\X$ according to $\mu$. Let $X$ be a random variable distributed according to $\mu$. We use the same symbol to represent a random variable and its distribution whenever it is clear from the context. The expectation of a function $f$ on $\X$ is defined as $\expec{}{f(X)}\defeq\expec{\mathbf{x}\sim X}{f(\mathbf{x})}=\sum_{x\in\X}\prob{X=x}\cdot f\br{x}=\sum_x\mu\br{x}\cdot f\br{x}$, where $\mathbf{x}\sim X$ represents that $\mathbf{x}$ is drawn according to $X$. For two distributions $p$ and $q$, the $\ell_1$-distance between $p$ and $q$ is defined to be $\onenorm{p-q}\defeq\sum_x\abs{p\br{x}-q\br{x}}$.

	In this paper, the lower-cased letters in bold $\mathbf{x},\mathbf{y},\cdots$ are reserved for random variables. The capital letters in bold, $\mathbf{P},\mathbf{Q},\ldots$ are reserved for random operators.

	\subsection{Gaussian spaces}
	
	For any integer $n>0$, let $\gamma_n$ represent the standard $n$-dimensional normal distribution. All the functions considered in this paper are in $L^2\br{\complex,\gamma_n}$ unless explicitly mentioned. We say $f\in L^2\br{\reals,\gamma_n}$ if $f\br{x}\in\reals$ for all $x$. We equipped $L^2\br{\complex,\gamma_n}$ with an inner product $\innerproduct{f}{g}_{\gamma_n}\defeq\expec{\mathbf{x}\sim\gamma_n}{\conjugate{f\br{\mathbf{x}}}g\br{\mathbf{x}}}$. Given $p\geq 1$ and $f\in L^2\br{\complex,\gamma_n}$, the {\em $p$-norm} of $f$ is defined to be $\norm{f}_p\defeq\br{\int_{\reals^n}\abs{f(x)}^p\gamma_n\br{dx}}^{\frac{1}{p}}$. Then $\innerproduct{f}{f}=\twonorm{f}^2$. The set of {\em Hermite polynomials} forms an orthonormal basis in $L^2\br{\complex,\gamma_1}$ with respect to the inner product $\innerproduct{\cdot}{\cdot}_{\gamma_1}$. The Hermite polynomial $H_r:\reals\rightarrow\reals$ for $r\in\mathbb{Z}_{\geq 0}$ is defined as
	\begin{equation}\label{eqn:hermitebasis}
		H_0\br{x}=1; H_1\br{x}=x; H_r\br{x}=\frac{(-1)^r}{\sqrt{r!}}e^{x^2/2}\frac{d^r}{dx^r}e^{-x^2/2}.
	\end{equation}

For any $\sigma\in\br{\sigma_1,\ldots,\sigma_n}\in\mathbb{Z}_{\geq 0}^n$, define
$H_{\sigma}:\reals^n\rightarrow\reals$ as
\begin{equation}\label{eqn:hermite}
	H_{\sigma}\br{x}\defeq\prod_{i=1}^nH_{\sigma_i}\br{x_i}.
\end{equation}
And $\abs{\sigma}\defeq\sum_i\sigma_i$. It is easy to verify that the set $\set{H_{\sigma}:\sigma\in\mathbb{Z}_{\geq 0}^n}$ forms an orthonormal basis in $L^2\br{\complex,\gamma_n}$. Every function $f\in L^2\br{\complex,\gamma_n}$ has an {\em Hermite expansion}  as
\[f\br{x}=\sum_{\sigma\in\mathbb{Z}_{\geq 0}^n}\widehat{f}\br{\sigma}\cdot H_{\sigma}\br{x},\]
 where $\widehat{f}\br{\sigma}$'s are the {\em Hermite coefficients} of $f$, which can be obtained by $\widehat{f}\br{\sigma}=\innerproduct{H_{\sigma}}{f}_{\gamma_n}$. The degree of $f$ is defined to be $\deg\br{f}\defeq\max\set{\abs{\sigma}:~\widehat{f}\br{\sigma}\neq 0}$.  Analogous to Fourier analysis, we have {\em Parseval's identity}, that is , $\twonorm{f}^2=\sum_{\sigma\in\mathbb{Z}_{\geq 0}^n}\abs{\widehat{f}\br{\sigma}}^2$. We say $f\in L^2\br{\complex,\gamma_n}$ is {\em multilinear} if $\widehat{f}\br{\sigma}$ is non-zero only if $\sigma\in\set{0,1}^n$. We need the following basic notions of {\em variance} and {\em influence of a coordinate} on a function.

 \begin{definition}\label{def:influencegaussian}
 	Given a function $f\in L^2\br{\complex,\gamma_n}$,
 	the {\em variance} of $f$ is defined to be
 	\begin{equation}\label{eqn:variance}
 	\var{f}\defeq\expec{\mathbf{x}\sim \gamma_n}{\abs{f\br{\mathbf{x}}-\expec{}{f}}^2}.
 	\end{equation}
 	For any set $S\subseteq[n]$, the {\em conditional variance} $\var{f\br{\mathbf{x}}|\mathbf{x}_S}$ is defined via
 	\begin{equation}\label{eqn:conditionalvariance}
 	\var{f\br{\mathbf{x}}|\mathbf{x}_S}\defeq\expec{\mathbf{x}\sim \gamma_n}{\abs{f\br{\mathbf{x}}-\expec{}{f\br{\mathbf{x}}|\mathbf{x}_S}}^2\text{\Large$\mid$}\mathbf{x}_S}.
 	\end{equation}
 	The {\em influence of $i$-th coordinate} on $f$, denoted by $\influence_i\br{f}$, is defined by
 	\begin{equation}\label{eqn:influencegaussian}
 		\influence_i\br{f}\defeq\expec{\mathbf{x}\sim \gamma_n}{\var{f\br{\mathbf{x}}|\mathbf{x}_{-i}}}.
 	\end{equation}
 	The {\em total influence} of $f$ is defined by
 	\[\influence\br{f}=\sum_i\influence_i\br{f}.\]
  	\end{definition}
 The following fact summarizes the basic properties of variance and influence. Readers may refer to~\cite{Odonnell08} for a thorough treatment.
 \begin{fact}\label{fac:influencegaussian}~\cite{Odonnell08,MosselOdonnell:2010}
 	Given $f\in L^2\br{\complex,\gamma_n}$, it holds that
 	\begin{enumerate}
 \item $\widehat{f}\br{\sigma}\in \reals$ if $f\in L^2\br{\reals,\gamma_n}$;
 		\item $\var{f}=\sum_{\sigma\neq 0^n}\abs{\widehat{f}\br{\sigma}}^2$;
 		\item $\influence_i\br{f}=\sum_{\sigma:\sigma_i\neq 0}\abs{\widehat{f}\br{\sigma}}^2$, and hence for all $i$, $\influence_i\br{f}\leq\var{f}$;
 		\item $\influence\br{f}=\sum_{\sigma}\abs{\set{i:\sigma_i>0}}\abs{\widehat{f}\br{\sigma}}^2$.
 		\item $\influence\br{f}\leq\deg\br{f}\var{f}$.
 	\end{enumerate}
 \end{fact}

\begin{definition}\label{def:ornstein}
  Given $0\leq\rho\leq 1$ and $f\in L^2\br{\complex,\gamma_n}$, we define the \em Ornstein-Uhlenbeck operator $U_{\rho}$ to be
  \[U_{\rho}f\br{z}\defeq\expec{\mathbf{x}\sim \gamma_n}{f\br{\rho z+\sqrt{1-\rho^2}\mathbf{x}}}.\]
\end{definition}

\begin{fact}\label{fac:gaussiannoisy}~\cite[Page 338, Proposition 11.37]{Odonnell08}
	For any $0\leq\rho\leq 1$ and $f\in L^2\br{\complex,\gamma_n}$, it holds that
	\[U_{\rho}f=\sum_{\sigma\in\mathbb{Z}_{\geq0}^n}\widehat{f}\br{\sigma}\rho^{\abs{\sigma}}H_{\sigma}.\]
\end{fact}

%
% We will work on vector-valued functions as well. For any vector $v\in\complex^k$, the norm of $v$ is defined to be $\norm{v}\defeq\sum\br
%{\abs{v_i}^2}^{1/2}$. For any function $f:\reals^n\rightarrow\reals^k$ and $\sigma\in\mathbb{Z}_{\geq 0}^n$, we will write $f=\br{f_1,\ldots, f_k}$, where $f_i:\reals\rightarrow\reals$  is the $i$-th coordinate of the output of $f$. That is, $f_i\br{x}\defeq\br{f(x)}_i$ for any $x\in\reals$. we use $\widehat{f}\br{\sigma}$ to denote $\br{\widehat{f_1}\br{\sigma},\ldots,\widehat{f_k}\br{\sigma}}$. The norm $\twonorm{f}^2\defeq\expec{\mathbf{x}\sim\gamma_n}{\norm{f\br{\mathbf{x}}}^2}=\sum_{i=1}^k\twonorm{f_i}^2=\sum_{\sigma\in\mathbb{Z}_{\geq 0}^n}\norm{\widehat{f}\br{\sigma}}^2$.  Also, the degree of $f$ is defined to be $\deg\br{f}\defeq\max_{i\in[k]}\deg\br{f_i}$. Again, all vector-valued functions with domain $\reals^n$ will be such that each coordinate is in $L^2\br{\gamma_n}$ unless explicitly mentioned.
	
	We will also be working on {\em vector-valued functions} in this paper.  A vector-valued function $f=\br{f_1,\ldots,f_k}:\reals^n\rightarrow\complex^k$ is in $L^2\br{\complex^k,\gamma_n}$ if $f_i\in L^2\br{\complex,\gamma_n}$ for all $i$. It is in $L^2\br{\reals^k,\gamma_n}$ if $f_i\in L^2\br{\reals,\gamma_n}$ for all $i$.  For any $f,g\in L^2\br{\complex^k,\gamma_n}$, the inner product $\innerproduct{f}{g}_{\gamma_n}\defeq\sum_{t=1}^k\innerproduct{f_t}{g_t}_{\gamma_n}$. The $p$-norm of $f$ is $\norm{f}_p=\br{\sum_{t=1}^k\norm{f_t}_p^p}^{1/p}$.
	For any $f\in L^2\br{\complex^k,\gamma_n}$, the Hermite coefficients of $f$ are the vectors  $\widehat{f}\br{\sigma}\defeq\br{\widehat{f_1}\br{\sigma},\ldots,\widehat{f_k}\br{\sigma}}$. The degree of $f$ is $\deg\br{f}\defeq\max_t\deg\br{f_t}$. We say $f$ is multilinear if each $f_t$ is multilinear. The variance of $f$ is $\var{f}\defeq\sum_t\var{f_t}$. The influence of $i$-th coordinate on $f$ is $\influence_i\br{f}\defeq\sum_i\influence_i\br{f_t}$. The total influence of $f$ is $\influence\br{f}=\sum_i\influence_i\br{f}$.
	The action of Ornstein-Uhlenbeck opeartor on $f$ is defined to be $U_{\rho}f\defeq\br{U_{\rho}f_1,\ldots, U_{\rho}f_k}$.
	
	For any vector $v\in\complex^k$, the norm of $v$ is defined to be $\norm{v}_2\defeq\sqrt{\sum_{i=1}^k\abs{v_i}^2}$.
	
	 It is easy to verify that Fact~\ref{fac:influencegaussian} and Fact~\ref{fac:gaussiannoisy} can be generalized to vector-valued functions by definitions.
	
	\begin{fact}\label{fac:vecfun}
		Given $f\in L^2\br{\complex^k,\gamma_n}$ and $0\leq\rho\leq 1$, it holds that
		\begin{enumerate}
			\item $\widehat{f}\br{\sigma}\in \reals^k$ if $f\in L^2\br{\reals^k,\gamma_n}$;
			\item $\var{f}=\sum_{\sigma\neq 0^n}\norm{\widehat{f}\br{\sigma}}_2^2$;
			\item $\influence_i\br{f}=\sum_{\sigma:\sigma_i\neq 0}\norm{\widehat{f}\br{\sigma}}_2^2$, and hence for all $i$, $\influence_i\br{f}\leq\var{f}$;
			\item $\influence\br{f}=\sum_{\sigma}\abs{\set{i:\sigma_i>0}}\twonorm{\widehat{f}\br{\sigma}}^2$;
			\item $\influence\br{f}\leq\deg\br{f}\var{f}$;
			\item $U_\rho f=\sum_{\sigma\in\mathbb{Z}_{\geq 0}^n}\widehat{f}\br{\sigma}\rho^{\abs{\sigma}}H_{\sigma}$.
			
		\end{enumerate}
	\end{fact}

	We are also working on the joint distribution $\rho$-correlated Gaussian distribution $\br{\reals\times\reals,\G_{\rho}}$. This is a $2$-dimensional Gaussian distribution $\br{X,Y}$, where $X$ and $Y$ are marginal distributed according to $\gamma_1$ with $\expec{}{XY}=\rho$.
	
	\subsection{Quantum mechanics}

	We briefly review the formalism of quantum mechanics over a finite dimensional system. For a more thorough treatment, readers may refer to~\cite{NC00,Wat08}. For a quantum system $A$, we associate a finite dimensional Hilbert space, which, by abuse of notation, is also denoted by $A$. We denote by $\M\br{A}$ and $\H\br{A}$ the set of all linear operators and the set of  Hermitian operators in the space, respectively. The identity operator in $A$ is denoted by $\id_A$.  If the dimension of $A$ is $d$, then we write $\M\br{A}=\M_d$, $\H\br{A}=\H_d$ and $\id_A=\id_d$. The subscripts may be dropped whenever it is clear from the context. A quantum state in the quantum system $A$ is represent by a {\em density operator} $\rho_A$, a positive semi-definite operator over the Hilbert space $A$ with unit trace. We denote by $\D\br{A}$ the set of all density operators in $A$. A quantum state is {\em pure} if the density operator is a rank-one projector $\ketbra{\psi}$, which is also represented by $\ket{\psi}$ for convenience. Composite quantum systems are associated with the {\em (Kronecker) tensor product space} of the underlying spaces, i.e., for quantum systems $A$ and $B$, the composition of the two systems are represented by $A\otimes B$ with the sets of linear operators, Hermitian  operators and density operators denoted by $\M\br{A\otimes B}$, $\H\br{A\otimes B}$ and $\D\br{A\otimes B}$, respectively. We sometimes use the shorthand $AB$ for $A\otimes B$. The sets of the linear operators and the Hermitian operators in the composition of $n$ $d$-dimensional Hilbert spaces are denoted by $\M_d^{\otimes n}$ and  $\H_d^{\otimes n}$, respectively. With slight abuse of notations, we assume that $\M_d^{\otimes n}=\complex$ and $\H_d^{\otimes n}=\reals$ when $n=0$. A {\em quantum channel} from the input system $A$ to the output system $B$ is represented by a {\em completely positive, trace-preserving linear map} (CPTP map).  The set of all quantum channels from a system $A$ to a system $B$ is denoted $\L\br{A,B}$. A {\em quantum operator} on $A$ is a channel acting $A$. The set of quantum operators on $A$ is denoted $\L\br{A}$. An important operation on a composite system $A\otimes B$ is the {\em partial trace} $\Tr_B\rho_{AB}$ which effectively derives the marginal state of the subsystem $A$ from the quantum state $\rho_{AB}$.  The partial trace is given by $\Tr_B\rho_{AB}\defeq\sum_i\br{\id_A\otimes\bra{i}}\rho_{AB}\br{\id_A\otimes\ket{i}}$ where $\set{\ket{i}}$ is an orthonormal basis for $B$. The partial trace is a valid quantum channel in $\L\br{A\otimes B, A}$ Note that the  action is independent of the choices of the basis chosen to represent it , so we unambiguously write $\rho_A=\Tr_B\rho_{AB}$. In this paper, we may also apply the partial trace to an arbitrary operator in $\M\br{A\otimes B}$. A pure state evolution on system $A$ with state $\ket{\psi}$ is represented by a unitary operator $U^A$, denoted by $U^A\ket{\psi}$.  An evolution of register $B$ of a state $\ket{\psi}_{AB}$ under the action of a unitary $U^B$  is represented by $\br{\id^A\otimes U^B}\ket{\psi}_{AB}$.  The superscripts and the subscripts might be dropped whenever it is clear from the context.  A {\em quantum measurement} is represented by a {\em positive-operator valued measure} (POVM), which is a set of positive semi-definite operators $\set{M_i}_{i=1}^n$ satisfying $\sum_{i=1}^nM_i=\id$, where $n$ is the number of possible measurement outcomes. Suppose the state  of the quantum system is $\rho$. Applying the measurement $\set{M_i}_{i=1}^n$, the probability that outputs $i$ is $\Tr M_i\rho$.

In this paper, we need the following fact.
\begin{fact}\label{fac:cauchyschwartz}
Given registers $A, B$, operators $P\in\H\br{A}, Q\in\H\br{B}$ and a bipartite state $\psi_{AB}$, it holds that
\begin{enumerate}
\item $\Tr\br{P\otimes\id_B}\psi_{AB}=\Tr P\psi_A$.

\item $\abs{\Tr\br{P\otimes Q}\psi_{AB}}\leq\br{\Tr P^2\psi_A}^{1/2}\cdot\br{\Tr Q^2\psi_B}^{1/2}$.
\end{enumerate}
\end{fact}	

    \subsection{Matrix analysis and matrix spaces}\label{subsec:matrixspace}

	Given an $m\times n$ matrix $M$, $\mathsf{abs}\br{M}$ represents the matrix obtained by substituting each entry of $M$ by its absolute value. for $1\leq p\leq\infty$ the $p$-norm of $M$ is defined to be $\norm{M}_p\defeq\br{\sum_{i=1}^{\min\set{m,n}}s_i\br{M}^p}^{1/p}$, where $\br{s_1\br{M},s_2\br{M},\ldots}$ are the singular values of $M$ sorted in a non-increasing order.  $\norm{M}\defeq\norm{M}_{\infty}=s_1\br{M}$ when $p=\infty$ It is easy to verify that $\norm{M}_p\leq\norm{M}_q$ if $p\geq q$. For $M\in\M_d$, $\abs{M}\defeq\sqrt{M^{\dagger}M}$. The {\em normalized $p$-norm} of $M$ is defined as $\nnorm{M}_p\defeq\br{\frac{1}{d}\sum_{i=1}^ds_i\br{M}^p}^{1/p}$ and $\nnorm{M}\defeq\nnorm{M}_{\infty}=s_1\br{M}$. We have $\nnorm{M}_p\geq\nnorm{M}_q$ if $p\geq q$. For any $M\in\H_d$, $\br{\lambda_1\br{M},\ldots,\lambda_d\br{M}}$ represents the eigenvalues of $M$ in a non-increasing order. Given two matrices $A, B$ of the same dimension, the {\em Hadamard product} of $A$ and $B$ is $A\circ B$, where $\br{A\circ B}\br{i,j}\defeq A\br{i,j}\cdot B\br{i,j}$.
	
	Given $P,Q\in\M_d$, we define
	\begin{equation}\label{eqn:innerproduct}	
	\innerproduct{P}{Q}\defeq\frac{1}{d}\Tr~P^{\dagger}Q.
	\end{equation}

	\begin{fact}\label{fac:innerproduct}
		$\innerproduct{\cdot}{\cdot}$ is an inner product. Then $\br{\innerproduct{\cdot}{\cdot},\M_d}$ forms a Hilbert space of dimension $d^2$. For any $M\in\M_d$, $\nnorm{M}_2^2=\innerproduct{M}{M}$.
	\end{fact}

	We say $\set{\B_0,\ldots,\B_{d^2-1}}$ is a {\em standard orthonormal basis} in $\M_d$. if  it is an orthonormal basis with all elements being Hermitian and $\B_0=\id_d$.
	
	\begin{fact}\label{fac:unitarybasis}
		Given a standard orthonormal basis $\set{\B_0,\ldots,\B_{d^2-1}}$, the set $\set{\B'_0,\ldots,\B'_{d^2-1}}$ is also a standard orthonormal basis in $\M_d$ if and only if $\B'_0=\id_d$ and there exists a $\br{d^2-1}\times \br{d^2-1}$ orthogonal matrix $U$ such that $\B'_i=\sum_{j=1}^{d^2-1}U_{i,j}\B_j$ for all $1\leq i\leq d^2-1$.
	\end{fact}

			\begin{fact}\label{fac:paulimutiplecopy}
		Let $\set{\B_i}_{i=0}^{d^2-1}$ be a standard orthonormal basis in $\M_d$, then
		\[\set{\B_{\sigma}\defeq\otimes_{i=1}^n\B_{\sigma_i}}_{\sigma\in[d^2]_{\geq 0}^n}\]
		is a standard orthonormal basis in $\M_d^{\otimes n}$.
	\end{fact}

The following lemma guarantees the existence of the standard orthonormal basis.
	\begin{lemma}\label{lem:paulibasis}
	For any integer $d\geq 2$, there exists a standard orthonormal basis in $\M_d$.
\end{lemma}
\begin{proof}
It is easy to verify that
the set $$\set{\sqrt{d}\cdot\ketbra{j}}_{0\leq j\leq  d-1}\cup\set{\sqrt{d/2}\br{\ketbratwo{j}{k}+\ketbratwo{k}{j}}}_{0\leq j<k\leq d-1}\cup\set{\sqrt{d/2}\br{\mathrm{i}\cdot \ketbratwo{j}{k}-\mathrm{i}\cdot\ketbratwo{k}{j}}}_{0\leq j<k\leq d-1}$$ forms an orthonormal basis in $\M_d$. Let's denote it by $\set{\A_1,\ldots, \A_{d^2}}$. Suppose
\[\id_d=\sum_{i=1}^{d^2} x_i\A_i.\]
Then $x=\br{x_1,\ldots, x_{d^2}}$ is a unit vector in $\reals^{d^2}$. Let $\br{x^{(0)},\ldots,x^{(d^2-1)}}$ be a set of orthonormal vectors in $\reals^{d^2}$ with $x^{(0)}=x$. Define $\B_i=\sum_{j=1}^{d^2}x^{(i)}_j\A_j$. Then $\set{\B_i}_{0\leq i\leq d^2-1}$ is a standard orthonormal basis in $\M_d$ by Fact~\ref{fac:unitarybasis}.
\end{proof}

	Given a standard orthonormal basis $\B=\set{\B_i}_{i=0}^{d^2-1}$ in $\M_d$, every matrix $M\in\M_d^{\otimes n}$ has a {\em Fourier expansion} with respect to the basis $\B$ given by
	\[M=\sum_{\sigma\in[d^2]_{\geq 0}^{n}}\widehat{M}\br{\sigma}\B_{\sigma},\]
	where $\widehat{M}\br{\sigma}$'s are the {\em Fourier coefficients} of $M$ with respect to the basis $\B$, which can be obtained as $\widehat{M}\br{\sigma}=\innerproduct{\B_{\sigma}}{M}$. The basic properties of $\widehat{M}\br{\sigma}$ are summarized in the following fact, which follow from the orthonormality of $\set{\B_{\sigma}}_{\sigma\in[d^2]_{\geq 0}^n}$.
	\begin{fact}\label{fac:basicfourier}
		Given a standard orthonormal basis $\set{\B_i}_{i=0}^{d^2-1}$ in $\M_d$ and $M,N\in\M_d$, it holds that
		\begin{enumerate}
			\item $\widehat{M}\br{\sigma}$ is real if $M$ is Hermitian;
			\item $\innerproduct{M}{N}=\innerproduct{\id}{M^{\dagger}N}=\innerproduct{MN^{\dagger}}{\id}=\sum_{\sigma}\conjugate{\widehat{M}\br{\sigma}}\widehat{N}\br{\sigma}$;
			\item $\nnorm{M}_2^2=\innerproduct{M}{M}=\innerproduct{M^{\dagger}M}{\id}=\innerproduct{\id}{M^{\dagger}M}=\sum_{\sigma}\abs{\widehat{M}\br{\sigma}}^2$;
			\item $\innerproduct{\id}{M}=\widehat{M}\br{0}$.	
		\end{enumerate}
	\end{fact}
	The {\em variance} of a matrix $M\in\M_d$ is defined to be  $\var{M}\defeq\innerproduct{M}{M}-\innerproduct{M}{\id}\innerproduct{\id}{M}$.
	The following lemma is easily verified.
	\begin{lemma}\label{lem:variance}
	Given a standard orthonormal basis $\set{\B_i}_{i=0}^{d^2-1}$ in $\M_d$ and $M\in\M_d$, it holds that
	$\var{M}=\sum_{\sigma\neq 0}\abs{\widehat{M}\br{\sigma}}^2$.
	\end{lemma}

%\begin{definition}\label{def:heisenberg}
%  Given an integer $d>0$, we fix an orthonormal basis $\set{\ket i}_{i=0}^{d-1}$ in $\complex^d$ and define the {\em Heisenberg-Weyl operators} $\set{X^jZ^k}_{0\leq j,k\leq d-1}$, where $X^j\ket{k}\defeq\ket{\br{k+j}\mod d}, Z^j\ket{k}\defeq e^{\mathrm{i}\cdot 2\pi\frac{jk}{d}}\ket{k}$.
%\end{definition}

		In this paper, we will be working on a particular basis in $\M_2$, {\em Pauli basis}, defined as  \[\P\defeq\set{\P_0\defeq \id_2, \P_1\defeq\begin{pmatrix}
		0 & 1\\1 & 0
		\end{pmatrix}, \P_2\defeq\begin{pmatrix}
		0 & -i\\i & 0
		\end{pmatrix}, \P_3\defeq\begin{pmatrix}
		1 & 0\\0 & -1
		\end{pmatrix}},\]
which is the set of {\em Heisenberg-Weyl operators} in $\M_2$.

\begin{fact}
  The set of Pauli basis $\P$ is a standard orthonormal basis in $\M_2$.
\end{fact}

	\begin{definition}
		Let $\B=\set{\B_i}_{i=0}^{d^2-1}$ be a standard orthonormal basis in $\M_d$, $P,Q\in\M_d^{\otimes n}$ and a subset $S\subseteq[n]$.
		\begin{enumerate}
			\item The degree of $P$ is defined to be $\deg P\defeq\max\set{\abs{\sigma}:\widehat{P}\br{\sigma}\neq 0}.$ where $|\sigma|$ represents the number of nonzeros in $\sigma$.
			\item For any $S\subseteq[n]$, $P_S\defeq\frac{1}{d^{|S^c|}}\Tr_{S^c}P$;
			\item For any $S\subseteq[n]$, $\innerproduct{P}{Q}_S\defeq\frac{1}{d^{|S|}}\Tr_S~P^{\dagger}Q.$
			\item For any $S\subseteq[n]$,  $\text{Var}_S[P]=\br{P^{\dagger}P}_{S^c}-\br{P_{S^c}}^{\dagger}\br{P_{S^c}}$. If $S=\set{i}$, we use $\text{Var}_i[P]$ in short.
			\item For any $i\in[n]$, $\influence_i\br{P}\defeq\innerproduct{\id}{\text{Var}_i[P]}.$
			\item $\influence\br{P}\defeq\sum_i\influence_i\br{P}$.
		\end{enumerate}
	\end{definition}

With the notion of degree, we define the low degree part and the high degree part of an operator.

\begin{definition}\label{def:lowdegreehighdegree}
Given integers $d,t>0$, a standard orthonormal basis $\B=\set{\B_i}_{i=0}^{d^2-1}$ in $\M_d$ and $P\in\M_d^{\otimes n}$, we define
\[P^{\leq t}\defeq\sum_{\sigma\in[d^2]_{\geq 0}^n:\abs{\sigma}\leq t}\widehat{P}\br{\sigma}\B_{\sigma};\]
\[P^{\geq t}\defeq\sum_{\sigma\in[d^2]_{\geq 0}^n:\abs{\sigma}\geq t}\widehat{P}\br{\sigma}\B_{\sigma}\]
and
\[P^{=t}\defeq\sum_{\sigma\in[d^2]_{\geq 0}^n:\abs{\sigma}=t}\widehat{P}\br{\sigma}\B_{\sigma};\]
where $\widehat{P}\br{\sigma}$'s are the Fourier coefficients of $P$ with respect to the basis $\B$.
\end{definition}

    \begin{lemma}\label{lem:pt}
    	The degree of $P$ is independent of the choices of the basis. Moreover, $P^{\leq t}, P^{\geq t}$ and $P^{=t}$ are also independent of the choices of the basis.
    \end{lemma}

\begin{proof}
	Let $\set{\B_{\sigma}}_{\sigma\in[d^2]_{\geq 0}}$ and $\set{\B'_{\sigma}}_{\sigma\in[d^2]_{\geq 0}}$ be two standard orthonormal basis in $\M_d$. From Fact~\ref{fac:unitarybasis},  there exists a $\br{d^2-1}\times \br{d^2-1}$ orthogonal matrix $U$ satisfying that $\B_{\sigma}=\sum_{\sigma'=1}^{d^2-1}U_{\sigma,\sigma'}\B'_{\sigma'}$ for any $\sigma\in[d^2-1]$. Suppose $P=P^{=t}$ with respect to the basis $\set{\B_{\sigma}}_{\sigma\in[d]_{\geq 0}}$. By linearity, we may assume that $P=\B_{\sigma}=\bigotimes_{i=1}^n\B_{\sigma_i}$  without loss of generality. It is easy to verify that each term in the expansion of $P$ in terms of the basis $\set{\B_i'}_{i=0}^{d^2-1}$ is of degree $\abs{\sigma}$.
\end{proof}

	\begin{lemma}\label{lem:partialvariance}
			Given $P\in\M_d^{\otimes n}$ a standard orthonormal basis $\B=\set{\B_i}_{i=0}^{d^2-1}$ in $\M_d$ and a subset $S\subseteq[n]$, it holds that
		\begin{enumerate}
			\item $P_S=\innerproduct{\id}{P}_{S^c}=\sum_{\sigma:\sigma_{S^c}=\mathbf{0}}\widehat{P}\br{\sigma}\B_{\sigma_S}$, $\norm{P_S}\leq\norm{P}$ and $\nnorm{P_S}_2\leq\nnorm{P}_2$;
			\item $\innerproduct{\id_{S^c}}{\text{Var}_S[P]}=\sum_{\sigma:\sigma_S\neq\mathbf{0}}\abs{\widehat{P}\br{\sigma}}^2.$
			\item $\influence_i\br{P}=\sum_{\sigma:\sigma_i\neq0}\abs{\widehat{P}\br{\sigma}}^2.$
			\item $\influence\br{P}=\sum_{\sigma}\abs{\sigma}\abs{\widehat{P}\br{\sigma}}^2\leq\deg P\cdot\nnorm{P}^2_2$.
		\end{enumerate}
	\end{lemma}
	\begin{proof}
		\begin{enumerate}
			\item $P_S=\frac{1}{d^{|S^c|}}\sum_{\sigma}\widehat{P}\br{\sigma}\Tr_{S^c}\B_{\sigma}=\text{RHS}.$
			
			For the second inequality, it suffices to show that $\norm{P_{-i}}\leq\norm{P}$ for any $i\in[n]$. Without loss of generality, we may assume $i=1$. For any unit vector $\ket{v}$
			\[\bra{v}P_{-1}\ket{v}=\frac{1}{d}\bra{v}\br{\Tr_{\set{1}}P}\ket{v}=\Tr P\br{\frac{\id_d}{d}\otimes\ketbra{v}}\leq\norm{P},\]
			where the last inequality is from the fact that $\abs{\Tr PQ}\leq\norm{P}\onenorm{Q}$ .
			
			For the last inequality,
			\[\nnorm{P_S}^2_2=\sum_{\sigma:\sigma_{S^c}={\mathbf{0}}}\abs{\widehat{P}\br{\sigma}}^2\leq\sum_{\sigma}\abs{\widehat{P}\br{\sigma}}^2=\nnorm{P}^2_2,\]
where the equalities are both from Fact~\ref{fac:basicfourier} item 3.

\item
			From the item 1,
			\[\br{P^{\dagger}P}_{S^c}=\sum_{\sigma,\sigma'}\conjugate{\widehat{P}\br{\sigma}}\widehat{P}\br{\sigma'}\br{\B_{\sigma}\B_{\sigma'}}_{S^c}=\sum_{\sigma,\sigma':\sigma_S=\sigma'_S}\conjugate{\widehat{P}\br{\sigma}}\widehat{P}\br{\sigma'}\B_{\sigma_{S^c}}\B_{\sigma'_{S^c}}.\]
			Meanwhile,
			\[\br{P_{S^c}}^{\dagger}\br{P_{S^c}}=\sum_{\sigma,\sigma':\sigma_S=\sigma'_S=\mathbf{0}}\conjugate{\widehat{P}\br{\sigma}}\widehat{P}\br{\sigma'}\B_{\sigma_{S^c}}\B_{\sigma'_{S^c}}.\]
			Therefore,
			\begin{eqnarray*}
			\innerproduct{\id_{S^c}}{\text{Var}_S[P]}=\sum_{\sigma:\sigma_S\neq\mathbf{0}}\abs{\widehat{P}\br{\sigma}}^2.
			\end{eqnarray*}
			\item It follows from the item 2 and the definition of $\influence_i\br{\cdot}$.
			\item It follows by a direct calculation.
		\end{enumerate}
	\end{proof}

		\begin{definition}\label{def:efronstein}\br{\textbf{Efron-Stein decomposition}} Given integers $n,d>0$, an operator $P\in\M_d^{\otimes n}$, a standard orthonormal basis $\set{\B_i}_{i=0}^{d^2-1}$ and $S\subseteq[n]$, set $P[S]\defeq\sum_{\sigma\in[d^2]_{\geq 0}^n:\supp{\sigma}=S}\widehat{P}\br{\sigma}\B_{\sigma}$, where $\supp{\sigma}\defeq\set{i\in[n]:\sigma_i>0}$. The Efron-Stein decompostion of $P$ is $P=\sum_{S\subseteq[n]}P[S].$
	\end{definition}
Again, the definition of $P[S]$ is independent of the choices of the basis $\set{\B_i}_{i=0}^{d^2-1}$, followed by the same argument for Lemma~\ref{lem:pt}.

	The following proposition follows from the orthogonality of $\B_i$'s.
	\begin{prop}\label{prop:enfronsteinortho}
		Given integers $d,n>0$, $S\neq T\subseteq[n]$ and $P,Q\in\M_d^{\otimes n}$, it holds that $\innerproduct{P[S]}{Q[T]}=0$.
	\end{prop}
	\begin{prop}\label{prop:efronstein}
		Given integers $d,n>0$,  $P\in\M_d^{\otimes n}$ and $S,T\subseteq[n], S\not\subseteq T$, it holds that
		\[\Tr_{T^c}~P[S]=0.\]
	\end{prop}
	\begin{proof}
		\[\Tr_{T^c}~P[S]=\Tr_{T^c}~\br{\sum_{\sigma:\supp{\sigma}=S}\widehat{P}\br{\sigma}\B_{\sigma}}=0,\]
		where the second equality is because $S\cap T^c\neq\emptyset$.
	\end{proof}

	\begin{lemma}\label{lem:jointbasis}
		Given $\psi_{AB}$ with $\psi_A=\frac{\id_{d_A}}{d_A}$ and $\psi_B=\frac{\id_{d_B}}{d_B}$, where $d_A$ and $d_B$  are the dimensions of $A$ and $B$, respectively, there exist  standard orthonormal basis $\set{\X_{\alpha}}_{\alpha\in[d_A^2]_{\geq 0}}$ and $\set{\Y_{\beta}}_{\beta\in[d_B^2]_{\geq 0}}$ in  $\M\br{A}$ and $\M\br{B}$, respectively, such that
		\[\Tr\br{\X_{\alpha}\otimes\Y_{\beta}}\psi_{AB}=0,\]
		whenever $\alpha\neq\beta$.
	\end{lemma}

	\begin{proof}
		Let $\set{\A_{\alpha}}_{\alpha\in[d_A^2]_{\geq 0}}$ and $\set{\B_{\beta}}_{\beta\in[d_B^2]_{\geq 0}}$ be arbitrary standard orthonormal basis in $\M_{d_A}$ and $\M_{d_B}$, respectively. Define  $\br{M_{\alpha,\beta}}_{\alpha\in[d_A^2]_{\geq 0},\beta\in[d_B^2]_{\geq 0}}$  where $M_{\alpha,\beta}=\Tr\br{\A_{\alpha}\otimes\B_{\beta}}\psi_{AB}$,
is a $d_A^2\times d_B^2$ real matrix. Then
		\begin{equation}
		M=\begin{pmatrix}
		1 & 0 & \cdots & 0 \\
		0 & & &          &    \\
		\raisebox{15pt}{\vdots}  & & \raisebox{15pt}{{\huge\mbox{{$M'$}}}}  & \\
		0 &  & &
		\end{pmatrix}
		\end{equation}
		Let $M'=U^{\dagger}DV^{\dagger}$ be a singular eigenvalue decomposition of $M'$ where $U, V$ are both orthogonal matrices and $D$ is a diagonal matrix.  For any $\alpha\in[d_A^2]_{\geq 0}$ and $\beta\in[d_B^2]_{\geq 0}$ set
		\[\X_{\alpha}\defeq\begin{cases}\sum_{\alpha'=1}^{d_A^2-1}U_{\alpha,\alpha'}\A_{\alpha'}~&\mbox{if $\alpha\neq 0$}\\
		\id_{d_A}~\mbox{otherwise},&\end{cases}
		~\mbox{and}~
		\Y_{\beta}\defeq\begin{cases}\sum_{\beta'=1}^{d_B^2-1}V_{\beta',\beta}\B_{\beta'}~&\mbox{if $\beta\neq 0$}\\
		\id_{d_B}~&\mbox{otherwise}.\end{cases}.\]
		From Fact~\ref{fac:unitarybasis}, $\set{\X_{\alpha}}_{\alpha=0}^{d_A^2-1}$ and $\set{\Y_{\beta}}_{\beta=0}^{d_B^2-1}$ are standard orthonormal basis in $\M\br{A}$ and $\M\br{B}$, respectively. Then
		\[\Tr\br{\X_{\alpha}\otimes\Y_{\beta}}\psi_{AB}=\begin{cases}\br{UM'V}_{\alpha,\beta}=\delta_{\alpha,\beta}D_{\alpha,\alpha}~&\mbox{if $\alpha,\beta>0$}\\ \delta_{\br{0,0},\br{\alpha,\beta}}~&\mbox{otherwise}.\end{cases}\]

	\end{proof}

\subsection{Random operators}\label{subsec:randomoperators}
    From the previous subsections, we see that the matrix space $\M_d^{\otimes n}$ and Gaussian space $L^2\br{\complex, \gamma_n}$ are both Hilbert spaces. In this subsection, we unify both spaces by {\em random operators}. In this paper, we only concern the case that the dimension $d=2$. However, the results in this subsection can be extended to arbitrary dimension $d$ directly.
    \begin{definition}\label{def:randomoperators}
    	Given integers $h, n>0$, we say $\mathbf{P}$ is a random operator if it is expressed as
    \begin{equation}\label{eqn:randomoperatorexpansion}
      \mathbf{P}=\sum_{\sigma\in[4]_{\geq 0}^h}p_{\sigma}\br{\mathbf{g}}\B_{\sigma},
    \end{equation}
    	where $\set{\B_i}_{i=0}^3$ is a standard orthonormal basis in $\M_2$, $p_{\sigma}\in L^2\br{\complex,\gamma_n}$ for all $\sigma\in[4]_{\geq 0}^h$ and $\mathbf{g}\sim \gamma_n.$
    	$\mathbf{P}\in L^2\br{\M_2^{\otimes h},\gamma_n}$ if $p_{\sigma}\in L^2\br{\complex,\gamma_n}$ for all $\sigma\in[4]_{\geq 0}^h$. Moreover, $\mathbf{P}\in L^2\br{\H_2^{\otimes n},\gamma_n}$ if $p_{\sigma}\in L^2\br{\reals,\gamma_n}$. Define a vector-valued function $p\defeq\br{p_{\sigma}}_{\sigma\in[4]_{\geq 0}^h}:\reals^n\rightarrow\complex^{4^h}$. We say $p$ is the {\em associated vector-valued function of $\mathbf{P}$ under the basis $\set{\B_i}_{i=0}^3$}.
       \end{definition}

   The following is a generalization of $p$-norm in $L^2\br{\M_2^{\otimes h},\gamma_n}$.

   \begin{definition}\label{def:randop}\footnote{To clarify the potential ambiguity, we consider $\nnorm{\mathbf{P}}_p$ to be a random variable and use $N_p\br{\cdot}$ to represent the normalized $p$-norm of a random operator.}\label{def:randoperatorsbasic}
   	Given integers $n,h\geq 0$ and $\mathbf{P}\in L^2\br{\M_2^{\otimes h},\gamma_n}$, for $p\geq 1$, the normalized $p$-norm of $\mathbf{P}$ is $N_p\br{\mathbf{P}}\defeq\expec{}{\nnorm{\mathbf{P}}_p^p}^{\frac{1}{p}}$. The degree of $\mathbf{P}$, denoted by $\deg\br{\mathbf{P}}$, is $\max_{\sigma\in[4]_{
   	\geq 0}^h}\deg\br{p_{\sigma}}$. We say $\mathbf{P}$ is multilinear if $p_{\sigma}\br{\cdot}$ is multilinear for all $\sigma\in[4]_{\geq 0}^h$.
   	 \end{definition}

   \begin{lemma}\label{lem:randoperator}
     Given integers $n,h\geq 0$ let $\mathbf{P}\in L^2\br{\M_2^{\otimes n},\gamma_h}$ with the associated vector-valued function $p$. It holds that  $N_2\br{\mathbf{P}}=\twonorm{p}.$
   \end{lemma}
\begin{proof}
	Consider
	\[N_2\br{\mathbf{P}}^2=\expec{}{\nnorm{\mathbf{P}}_2^2}=\expec{\mathbf{g}\sim \gamma_n}{\sum_{\sigma\in[4]_{\geq 0}^h}\abs{p_{\sigma}\br{\mathbf{g}}}^2}=\twonorm{p}^2,\]
	where the second equality is from Fact~\ref{fac:basicfourier} item 3.
\end{proof}
\begin{lemma}\label{lem:influencerandomoperator}
	Given a multilinear random operator $\mathbf{P}\in L^2\br{\M_2^{\otimes h},\gamma_n}$ with degree $d$ and the associated vector-valued function $p$ under a standard orthonormal basis, it holds that
	\[\influence\br{p}\leq \deg\br{\mathbf{P}}N_2\br{\mathbf{P}}^2.\]
\end{lemma}
\begin{proof}
	Consider
	\[\influence\br{p}=\sum_{i=1}^n\influence_i\br{p}\leq\deg\br{p}\var{p}\leq\deg\br{\mathbf{P}}\twonorm{p}^2=\deg\br{\mathbf{P}}N_2\br{\mathbf{P}}^2,\]
	where the first inequality is from Fact~\ref{fac:vecfun} item 5; the second equality is from Lemma~\ref{lem:randoperator}.
\end{proof}
  We say a pair of random operators $\br{\mathbf{P},\mathbf{Q}}\in L^2\br{\M_2^{\otimes h},\gamma_n}\times L^2\br{\M_2^{\otimes h},\gamma_n} $ are {\em joint random operators} if the random variables $\br{\mathbf{g},\mathbf{h}}$ in $\br{\mathbf{P},\mathbf{Q}}$ are drawn from the joint distribution $\G_{\rho}^{\otimes n}$ for $0\leq\rho\leq 1$.
\subsection{Miscellaneous}\label{subsec:misc}

Throughout this paper, any function $f:\reals\rightarrow\reals$ is also viewed as a map $f:\H_d\rightarrow\H_d$ defined as
$f\br{P}=\sum_if\br{\lambda_i}\ketbra{v_i}$, where $P=\sum_i\lambda_i\ketbra{v_i}$ is a spectral decomposition of $P$.

	Given a convex body $\Delta\subseteq\reals^k$, we say a map $\R:\reals^k\rightarrow\reals^k$ is the rounding map of $\Delta$ if for any $x\in\reals^k$, $\R\br{x}$ is the element in $\Delta$ that is closest to $x$ in $\twonorm{\cdot}$ distance. The following well-known fact states that the Lipschitz coefficient of a rounding map is at most $1$.
	
	\begin{fact}\label{fac:rounding}
		Let $\Delta$ be a convex set in $\reals^k$ with the rounding map $\R$. It holds that
		\[\twonorm{\R\br{x}-\R\br{y}}\leq\twonorm{x-y},\]
		for any $x,y\in\reals^k$.

Thus, if $\Delta$ contains the element $\br{0,\ldots, 0}$, then $\R$ is a contraction map. Namely, $\twonorm{\R\br{x}}\leq\twonorm{x}$ for any $x\in\reals^k$.
	\end{fact}

\section{Main results}\label{sec:mainresult}
\begin{theorem}\label{thm:nijs}
  Given $0\leq \rho<1$, $\epsilon\in(0,1)$, integers $n,s>0$, a bipartite state $\psi_{AB}$ with $\psi_A=\psi_B=\frac{\id_2}{2}<1$ and the maximal correlation $\rho=\rho\br{\psi_{AB}}$ defined in Definition~\ref{def:maximalcorrelation}, sequences of possibly repetitive Hermitian operators $P_1,\ldots, P_s\in\H_2^{\otimes n}$ and $Q_1,\ldots, Q_s\in\H_2^{\otimes n}$ satisfying $0\leq P_u,Q_u\leq\id$ for any $1\leq u\leq s$, there exists an explicitly computable $D=D\br{\rho,\epsilon,s}$ and maps $f,g:\H_2^{\otimes n}\rightarrow \H_2^{\otimes D}$ with $\widetilde{P}_u=f\br{P_u}$ and $\widetilde{Q}_u=g\br{Q_u}$ for  $1\leq  u\leq s$, such that the following holds.
  \begin{enumerate}
    \item $0\leq\widetilde{P}_u\leq\id$ and $0\leq \widetilde{Q}_u\leq\id$.
    \item $\abs{\frac{1}{2^n}\Tr~P_u-\frac{1}{2^{D}}\Tr~\widetilde{P}_u}\leq\epsilon$ and $\abs{\frac{1}{2^n}\Tr~Q_u-\frac{1}{2^{D}}\Tr~\widetilde{Q}_u}\leq\epsilon$.
    \item $\abs{\Tr~\br{P_u\otimes Q_u}\psi_{AB}^{\otimes n}-\Tr~\br{\widetilde{P}_u\otimes \widetilde{Q}_u}\psi_{AB}^{\otimes D}}\leq\epsilon$.
  \end{enumerate}
   In particular, one may choose $D=\exp\br{\mathrm{poly}\br{s,\exp\br{\mathrm{poly}\br{\frac{1}{\epsilon},\frac{1}{1-\rho}}}}}$.
\end{theorem}

\begin{proof}
  Let $\delta,\tau $ be parameters which are chosen later. The proof is composed of several steps.
  \begin{itemize}
    \item \textbf{Smoothing operators}. For $u\in[s]$, we apply Lemma~\ref{lem:smoothing of strategies} to operators $P_u$ and $Q_u$ to get $P_u^{(1)}$ and $Q_u^{(1)}$ and $ d_1=\frac{2\log^2(1/\delta)}{C\br{1-\rho}\delta}$ for some constant $C$ satisfying that
        \bigskip
        \begin{enumerate}
          \item $0\leq P_u^{(1)}\leq\id~\mbox{and}~0\leq Q_u^{(1)}\leq\id;$
          \item $\Tr~P_u^{(1)}=\Tr~P_u~\mbox{and}~\Tr~Q_u^{(1)}=\Tr~Q_u;$
          \item $\nnorm{P_u^{(1)}}_2\leq\nnorm{P_u}_2~\mbox{and}~\nnorm{Q_u^{(1)}}_2\leq\nnorm{Q_u}_2;$
          \item $\abs{\Tr\br{P_u^{(1)}\otimes Q_u^{(1)}}\psi^{\otimes n}_{AB}-\Tr\br{P_u\otimes Q_u}\psi^{\otimes n}_{AB}}\leq\delta;$
          \item $\nnorm{\br{P_u^{(1)}}^{>{d_1}}}^2_2\leq\delta~\mbox{and}~\nnorm{\br{Q_u^{(1)}}^{>{d_1}}}^2_2\leq\delta.$
          \item If $P_u=P_v$, then $P_u^{(1)}=P_v^{(1)}$. If $Q_u=Q_v$, then $Q_u^{(1)}=Q_v^{(1)}$.
        \end{enumerate}
        \bigskip

    \item\textbf{Regularization}. For any $u\in[s]$, applying Lemma~\ref{lem:regular} to $P_u^{(1)}$ and $Q_u^{(1)}$ with $\delta\leftarrow\delta, \epsilon\leftarrow\tau, d\leftarrow d_1$, we obtain a set $H_u\subseteq[n]$ of size $h_u=\abs{H}\leq\frac{2d_1}{\tau}$ such that

        \[\br{\forall i\notin H_u}~\influence_i\br{\br{P_u^{(1)}}^{\leq d_1}}\leq\tau,~\mbox{and}~\influence_i\br{\br{Q_u^{(1)}}^{\leq d_1}}\leq\tau.\]
        Set $H=\cup_uH_u$. Then $h=\abs{H}\leq\frac{2sd_1}{\tau}$. It holds that for any $u\in[s]$,

            \begin{enumerate}
    	\item $\br{\forall i\notin H}~\influence_i\br{\br{P_u^{(1)}}^{\leq d_1}}\leq\tau,~\mbox{and}~\influence_i\br{\br{Q_u^{(1)}}^{\leq d_1}}\leq\tau;$
    	\item $\Tr~\br{P_u^{(1)}\otimes Q_u^{(1)}}\psi_{AB}^{\otimes h}=\sum_{\sigma\in[4]_{\geq 0}^h}c_{\sigma}\Tr~\br{P^{(1)}_{u,\sigma}\otimes Q^{(1)}_{u,\sigma}}\psi_{AB}^{\otimes (n-h)},$

   	where
    	$\br{c_i}_{i=0}^3$ are the singular values of the matrix $\mathsf{Corr}\br{\psi_{AB}}$ defined in Definition~\ref{def:covariancematrix} and $c_{\sigma}=c_{\sigma_1}\cdot c_{\sigma_2}\cdots c_{\sigma_h}$ and
    	\[P^{(1)}_{\sigma}=\sum_{\tau\in[4]_{\geq 0}^n:\tau_H=\sigma}\widehat{P}\br{\sigma}\A_{\sigma}~\mbox{and}~Q_{\sigma}^{(1)}=\sum_{\tau\in[4]_{\geq 0}^n:\tau_H=\sigma}\widehat{Q}\br{\sigma}\B_{\sigma},\]
    \end{enumerate}
    for standard orthonormal basis $\set{\A_i}_{i=0}^3$ and $\set{\B_i}_{i=0}^3$.

    \item\textbf{Invariance from $\H_2^{\otimes n}$ to $L^2\br{\H_2^{\otimes h},\gamma_{3\br{n-h}}}$}. ~For any $u\in[s]$, applying Lemma~\ref{lem:jointinvariance} to $P_u^{(1)}$ and $Q_u^{(1)}$ and $H$, we obtain  degree-$d_1$ multilinear joint random operators $\br{\mathbf{P}_u^{(2)},\mathbf{Q}_u^{(2)}}\in L^2\br{\M_2^{\otimes h},\gamma_{3\br{n-h}}}\times  L^2\br{\M_2^{\otimes h},\gamma_{3\br{n-h}}}$ with joint random variables $\br{\mathbf{g}_i,\mathbf{h}_i}_{i=1}^{3\br{n-h}}\sim \G_{\rho}^{\otimes 3(n-h)}$ such that,
        \bigskip
    \begin{enumerate}
    \item $2^{n-h}\expec{}{\Tr~\mathbf{P}_u^{(2)}}=\Tr~P_u^{(1)}~\mbox{and}~2^{n-h}\expec{}{\Tr~\mathbf{Q}_u^{(2)}}=\Tr~Q_u^{(1)};$

    \item $N_2\br{\mathbf{P}_u^{(2)}}\leq\nnorm{P^{(1)}}_2~\mbox{and}~N_2\br{\mathbf{Q}_u^{(2)}}\leq\nnorm{Q^{(1)}}_2;$
    	
    \item $\expec{}{\Tr~\zeta\br{\mathbf{P}_u^{(2)}}}\leq O\br{2^h\br{\br{3^{d_1}\sqrt{\tau}d_1}^{2/3}+\sqrt{\delta}}}$ and
    	
    $\expec{}{\Tr~\zeta\br{\mathbf{Q}_u^{(2)}}}\leq O\br{2^h\br{\br{3^{d_1}\sqrt{\tau}d_1}^{2/3}+\sqrt{\delta}}},$
    where $\zeta\br{\cdot}$ is defined in Eq.~\eqref{eqn:zeta};

	    \item $\Tr~\br{P_u^{(1)}\otimes Q_u^{(1)}}\psi_{AB}^{\otimes n}=\expec{}{\Tr~\br{\br{\mathbf{P}_u^{(2)}\otimes\mathbf{Q}_u^{(2)}}\psi_{AB}^{\otimes h}}}.$
	
	    \item If $P_u^{(1)}=P_v^{(1)}$, then $\mathbf{P}_u^{(2)}=\mathbf{P}_v^{(2)}$. If $Q_u^{(1)}=Q_v^{(1)}$, then $\mathbf{Q}_u^{(2)}=\mathbf{Q}_v^{(2)}$.
	    	
    \end{enumerate}
    \bigskip
    \item\textbf{Dimension reduction}. ~For any $u\in[s]$, applying Lemma~\ref{lem:dimensionreduction} to $\br{\mathbf{P}_u^{(2)},\mathbf{Q}_u^{(2)}}$ with $\delta\leftarrow\delta/2s, d\leftarrow d_1, n\leftarrow 3(n-h),\alpha\leftarrow 1/4s$, and the union bound on the $u$'s, we obtain
        joint random operators $\br{\mathbf{P}_{u}^{(3)},\mathbf{Q}_{u}^{(3)}}\in L^2\br{\H_2^{\otimes h},\gamma_{n_0}}\times L^2\br{\H_2^{\otimes h},\gamma_{n_0}}$ with the random variables drawn from $\G_{\rho}^{\otimes n_0}$ such that for all $u\in[s]$ the following holds.
        \bigskip
        \begin{enumerate}
         \item $\abs{\expec{}{\Tr~\mathbf{P}_u^{(3)}}-\expec{}{\Tr~\mathbf{P}_u^{(2)}}}\leq\delta2^hN_2\br{\mathbf{P}_u^{(2)}}$
         and
         $\abs{\expec{}{\Tr~\mathbf{Q}_u^{(3)}}-\expec{}{\Tr~\mathbf{Q}_u^{(2)}}}\leq\delta2^hN_2\br{\mathbf{Q}_u^{(2)}}.$
           \item $N_2\br{\mathbf{P}_u^{(3)}}\leq\br{1+\delta}N_2\br{\mathbf{P}_u^{(2)}}
               ~\mbox{and}~ N_2\br{\mathbf{Q}_u^{(3)}}\leq\br{1+\delta}N_2\br{\mathbf{Q}^{(2)}_u}.$
             \item $\expec{}{\Tr~\zeta\br{\mathbf{P}_u^{\br{3}}}}\leq 2\sqrt{s}\expec{}{\Tr~\zeta\br{\mathbf{P}_u^{\br{2}}}}~\mbox{ and}~ \expec{}{\Tr~\zeta\br{\mathbf{Q}_u^{\br{3}}}}\leq 2\sqrt{s}\expec{}{\Tr~\zeta\br{\mathbf{Q}_u^{\br{2}}}}.$
  \item $\abs{\expec{}{\Tr\br{\mathbf{P}_u^{\br{3}}\otimes \mathbf{Q}_u^{\br{3}}}\psi_{AB}^{\otimes h}}-\expec{}{\Tr\br{\mathbf{P}_u^{(2)}\otimes \mathbf{Q}_u^{(2)}}\psi_{AB}^{\otimes h}}}\leq\delta N_2\br{\mathbf{P}_u^{(2)}}N_2\br{\mathbf{Q}_u^{(2)}}.$
  \item If $\mathbf{P}_u^{(2)}=\mathbf{P}_v^{(2)}$, then $\mathbf{P}_u^{(3)}=\mathbf{P}_v^{(3)}$. If $\mathbf{Q}_u^{(2)}=\mathbf{Q}_v^{(2)}$, then $\mathbf{Q}_u^{(3)}=\mathbf{Q}_v^{(3)}$.
        \end{enumerate}
        \bigskip
  Here $n_0=\frac{4^{3h+4}d_1^{O\br{d_1}}s^2}{\delta^2}$.
    \item\textbf{Smoothing random operators}. For any $u\in[s]$, applying Lemma~\ref{lem:smoothgaussian} to $\br{\mathbf{P}_u^{(3)},\mathbf{Q}_u^{(3)}}$ with $ h\leftarrow h, n\leftarrow n_0$ we obtain joint random operators $\br{\mathbf{P}_u^{(4)},\mathbf{Q}_u^{(4)}}\in L^2\br{\H_2^{\otimes h},\gamma_{n_0}}\times L^2\br{\H_2^{\otimes h},\gamma_{n_0}}$ such that,
        \bigskip
    \begin{enumerate}
     	\item $\deg\br{\mathbf{P}_u^{(4)}}\leq d_2~\mbox{and}~\deg\br{\mathbf{Q}_u^{(4)}}\leq d_2.$
\item $\expec{}{\Tr~\br{\mathbf{P}_u^{(4)}}}=\expec{}{\Tr~\br{\mathbf{P}_u^{(3)}}}~\mbox{and}~\expec{}{\Tr~\br{\mathbf{Q}_u^{(4)}}}=\expec{}{\Tr~\br{\mathbf{Q}_u^{(3)}}},$
			\item $N_2\br{\mathbf{P}_u^{(4)}}\leq N_2\br{\mathbf{P}_u^{(3)}}~\mbox{ and}~N_2\br{\mathbf{Q}_u^{(4)}}\leq N_2\br{\mathbf{Q}_u^{(3)}}.$
\item $\expec{}{\Tr~\zeta\br{\mathbf{P}_u^{(4)}}}\leq2\br{\expec{}{\Tr~\zeta\br{\mathbf{P}_u^{(3)}}}+\delta2^hN_2\br{\mathbf{P}_u^{(3)}}^2}$  and

    $\expec{}{\Tr~\zeta\br{\mathbf{Q}_u^{(4)}}}\leq2\br{\expec{}{\Tr~\zeta\br{\mathbf{P}_u^{(3)}}}+\delta2^hN_2\br{\mathbf{Q}_u^{(3)}}^2}.$
\item $\abs{\expec{}{\Tr\br{\mathbf{P}_u^{(3)}\otimes\mathbf{Q}_u^{(3)}}\psi_{AB}^{\otimes h}}-\expec{}{\Tr\br{\mathbf{P}_u^{(4)}\otimes\mathbf{Q}_u^{(4)}}\psi_{AB}^{\otimes h}}}\leq \delta N_2\br{\mathbf{P}_u^{(3)}}N_2\br{\mathbf{Q}_u^{(3)}}.$
\item If $\mathbf{P}_u^{(3)}=\mathbf{P}_v^{(3)}$, then $\mathbf{P}_u^{(4)}=\mathbf{P}_v^{(4)}$. If $\mathbf{Q}_u^{(3)}=\mathbf{Q}_v^{(3)}$, then $\mathbf{Q}_u^{(4)}=\mathbf{Q}_v^{(4)}$.
    \end{enumerate}
    \bigskip
    Here $d_2=O\br{\frac{\log^2\frac{1}{\delta}}{\delta\br{1-\rho}}}$.

    \item\textbf{Multilinearization}. For any $u\in[s]$, suppose
    \[\br{\mathbf{P}_u^{(4)},\mathbf{Q}_u^{(4)}}=\br{\sum_{\sigma\in[4]_{\geq 0}^h}p^{(4)}_{u,\sigma}\br{\mathbf{g}}\A_{\sigma},\sum_{\sigma\in[4]_{\geq 0}^h}q^{(4)}_{u,\sigma}\br{\mathbf{h}}\B_{\sigma}}_{\br{\mathbf{g},\mathbf{h}}\sim\G_{\rho}^{\otimes n_0}}.\]

    Applying Lemma~\ref{lem:multiliniearization} with $d\leftarrow d_2,h\leftarrow h, n\leftarrow n_0,\delta\leftarrow\tau$, we obtain multilinear joint random operators
        \[\br{\mathbf{P}_u^{(5)},\mathbf{Q}_u^{(5)}}=\br{\sum_{\sigma\in[4]_{\geq 0}^h}p^{(5)}_{u,\sigma}\br{\mathbf{x}}\A_{\sigma},\sum_{\sigma\in[4]_{\geq 0}^h}q^{(5)}_{u,\sigma}\br{\mathbf{y}}\B_{\sigma}}_{\br{\mathbf{x},\mathbf{y}}\sim\G_{\rho}^{\otimes n_0t}},\]
       with $t=O\br{\frac{d_2^2}{\tau^2}}$ such that the following holds.
       \bigskip
        \begin{enumerate}
          \item $\deg\br{\mathbf{P}_u^{(5)}}\leq d_2$ and $\deg\br{\mathbf{Q}_u^{(5)}}\leq d_2$.
          \item For all $\br{i,j}\in[n_0]\times [t]$
          \[\influence_{in_0+j}\br{p^{(5)}_{u,\sigma}}\leq\tau\influence_i\br{p^{(4)}_{u,\sigma}}~\mbox{          and}~
          \influence_{in_0+j}\br{q^{(5)}_{u,\sigma}}\leq\tau\influence_i\br{q^{(4)}_{u,\sigma}}.\]
          \item
          $\expec{}{\Tr~\mathbf{P}_u^{(5)}}=\expec{}{\Tr~\mathbf{P}_u^{(4)}}~\mbox{and}~\expec{}{\Tr~\mathbf{Q}_u^{(5)}}=\expec{}{\Tr~\mathbf{Q}_u^{(4)}}.$
          \item $N_2\br{\mathbf{P}_u^{(5)}}\leq N_2\br{\mathbf{P}_u^{(4)}}~\mbox{and}~N_2\br{\mathbf{Q}_u^{(5)}}\leq N_2\br{\mathbf{Q}_u^{(4)}}.$

\item  $\abs{\expec{}{\Tr~\zeta\br{\mathbf{P}_u^{(5)}}}-\expec{}{\Tr~\zeta\br{\mathbf{P}_u^{(4)}}}}\leq\tau2^{h+2}N_2\br{\mathbf{P}_u^{(4)}}^2$

and

$\abs{\expec{}{\Tr~\zeta\br{\mathbf{Q}_u^{(5)}}}-\expec{}{\Tr~\zeta\br{\mathbf{Q}_u^{(4)}}}}\leq\tau2^{h+2}N_2\br{\mathbf{Q}_u^{(4)}}^2.$
\item $\abs{\expec{}{\Tr\br{\mathbf{P}_u^{(5)}\otimes\mathbf{Q}_u^{(5)}}\psi_{AB}^{\otimes h}}-\expec{}{\Tr\br{\mathbf{P}_u^{(4)}\otimes\mathbf{Q}_u^{(4)}}\psi_{AB}^{\otimes h}}}\leq \tau N_2\br{\mathbf{P}_u^{(4)}}N_2\br{\mathbf{Q}_u^{(4)}}.$
\item If $\mathbf{P}_u^{(4)}=\mathbf{P}_v^{(4)}$, then $\mathbf{P}_u^{(5)}=\mathbf{P}_v^{(5)}$. If $\mathbf{Q}_u^{(4)}=\mathbf{Q}_v^{(4)}$, then $\mathbf{Q}_u^{(5)}=\mathbf{Q}_v^{(5)}$.
        \end{enumerate}
        \bigskip
\item\textbf{Invariance from $L^2\br{\H_2^{\otimes h},\gamma_{n_0t}}$ to $\H_2^{\otimes h+n_0t}$}. From the item 1 and item 2 above and Lemma~\ref{lem:partialvariance} item 4 and Lemma~\ref{lem:randoperator}, we have
\[\sum_{\sigma}\influence_i\br{p_{u,\sigma}^{(5)}}\leq\tau d_2N_2\br{\mathbf{P}_u^{(4)}}^2.\]
    Similarly, we have
    \[\sum_{\sigma}\influence_i\br{q_{u,\sigma}^{(5)}}\leq\tau d_2N_2\br{\mathbf{Q}_u^{(4)}}^2.\]
    For any $u\in[s]$, we apply Lemma~\ref{lem:invariancejointgaussian} to $\br{\mathbf{P}_u^{(5)},\mathbf{Q}_u^{(5)}}$ with $n\leftarrow n_0t,h\leftarrow h,d\leftarrow d_2$,
    $$\tau\leftarrow\tau_0\defeq\max_u\set{\max\set{\tau d_2N_2\br{\mathbf{P}_u^{(4)}}^2,\tau d_2N_2\br{\mathbf{Q}_u^{(4)}}^2}~:~u\in[s]}$$
    to get $\br{P_u^{(6)},Q_u^{(6)}}\in\H_2^{\otimes h+n_0t}\times \H_2^{\otimes h+n_0t}$ satisfying that

    \bigskip
    \begin{enumerate}
      \item $2^{n_0t}\expec{}{\Tr~\mathbf{P}_u^{(5)}}=\Tr~P_u^{(6)}~\mbox{and}~2^{n_0t}\expec{}{\Tr~\mathbf{Q}_u^{(5)}}=\Tr~Q_u^{(6)};$

      \item $N_2\br{\mathbf{P}_u^{(5)}}=\nnorm{P_u^{(6)}}_2~\mbox{and}~N_2\br{\mathbf{Q}_u^{(5)}}=\nnorm{Q_u^{(6)}}_2.$

      \item $\abs{\expec{}{2^{n_0t}\Tr~\zeta\br{\mathbf{P}_u^{(5)}}}-\Tr~\zeta\br{P_u^{(6)}}}\leq O\br{2^{n_0t+h}\br{3^{d_2}d_2\sqrt{\tau_0}}^{2/3}},$
          and

      $\abs{\expec{}{2^{n_0t}\Tr~\zeta\br{\mathbf{Q}_u^{(5)}}}-\Tr~\zeta\br{Q_u^{(6)}}}\leq O\br{2^{n_0t+h}\br{3^{d_2}d_2\sqrt{\tau_0}}^{2/3}}.$

      \item $\Tr~\br{P_u^{(6)}\otimes Q_u^{(6)}}\psi_{AB}^{\otimes h+n_0t}=\expec{}{\Tr~\br{\mathbf{P}_u^{(5)}\otimes \mathbf{Q}_u^{(5)}}\psi_{AB}^{\otimes h}}.$
      \item If $\mathbf{P}_u^{(5)}=\mathbf{P}_v^{(5)}$, then $P_u^{(6)}=P_v^{(6)}$. If $\mathbf{Q}_u^{(5)}=\mathbf{Q}_v^{(5)}$, then $Q_u^{(6)}=Q_v^{(6)}$.
    \end{enumerate}
    \bigskip

\item\textbf{Rounding to measurement operators}. Finally, we let $\widetilde{P}_u$ and $\widetilde{Q}_u$ be the roundings of $P_u^{(6)}$ and $Q_u^{(6)}$ to the Hermitian matrices between $0$ and $\id$, respectively. Namely, suppose $P_u^{(6)}=\sum_ip_i\ketbra{u_i}$ and $Q_u^{(6)}=\sum_iq_i\ketbra{v_i}$ are the spectral decompositions of $P_u^{(6)}$ and $Q_u^{(6)}$, respectively. Then
$\widetilde{P}_u=\sum_i\widetilde{p_i}\ketbra{u_i}$ and $\widetilde{Q}_u=\sum_i\widetilde{q_i}\ketbra{v_i}$, where
$$\widetilde{x}\defeq\begin{cases}
1~&\mbox{if $x\geq 1$}\\
x~&\mbox{if $0\leq x\leq 1$}\\
0~&\mbox{otherwise}
\end{cases}.$$
Then
\[\frac{1}{2^{n_0t+h}}\abs{\Tr~P_u^{(6)}-\Tr~\widetilde{P}_u}\leq\nnorm{P_u^{(6)}-\widetilde{P}_u}_1\leq\nnorm{P_u^{(6)}-\widetilde{P}_u}_2=\br{\frac{1}{2^{n_0t+h}}\Tr~\zeta\br{P_u^{(6)}}}^{1/2};\]
\[\frac{1}{2^{n_0t+h}}\abs{\Tr~Q_u^{(6)}-\Tr~\widetilde{Q}_u}\leq\nnorm{Q_u^{(6)}-\widetilde{Q}_u}_1\leq\nnorm{Q_u^{(6)}-\widetilde{Q}_u}_2=\br{\frac{1}{2^{n_0t+h}}\Tr~\zeta\br{Q_u^{(6)}}}^{1/2},\]
where the equalities are from the definition of the function $\zeta\br{\cdot}$. And
\begin{eqnarray*}
	&&\abs{\Tr~\br{P_u^{(6)}\otimes Q_u^{(6)}}\psi_{AB}^{\otimes n_0t+h}-\Tr~\br{\widetilde{P}_u\otimes \widetilde{Q}_u}\psi_{AB}^{\otimes n_0t+h}}\\
	&\leq&\abs{\Tr~\br{P_u^{(6)}\otimes \br{Q_u^{(6)}-\widetilde{Q}_u}}\psi_{AB}^{\otimes n_0t+h}}+\abs{\Tr~\br{\br{P_u^{(6)}-P_u}\otimes \widetilde{Q}_u}\psi_{AB}^{\otimes n_0t+h}}\\
	&\leq&\nnorm{P_u^{(6)}}_2\nnorm{Q_u^{(6)}-\widetilde{Q}_u}_2+\nnorm{\widetilde{Q}_u}_2\nnorm{P_u^{(6)}-\widetilde{P}_u}_2\\
	&=&\nnorm{P_u^{(6)}}_2\br{\frac{1}{2^{n_0t+h}}\Tr~\zeta\br{Q_u^{(6)}}}^{1/2}+\nnorm{Q_u^{(6)}}_2\br{\frac{1}{2^{n_0t+h}}\Tr~\zeta\br{P_u^{(6)}}}^{1/2},
\end{eqnarray*}
where the second equality is from Fact~\ref{fac:cauchyschwartz} and the  equality is from the definition of $\zeta\br{\cdot}$.

  \end{itemize}
Choosing $\delta=\frac{\epsilon^2}{10000}, \tau=\epsilon^3/\br{s^4\exp\br{\frac{4\log^2\frac{1}{\delta}}{\br{1-\rho}\delta}}}$ and $D=n_0t+h$, we conclude the desired result.
\end{proof}

\begin{theorem}
  Given parameters $0<\epsilon,\rho<1$ a mono-state binary game $\br{G,\psi_{AB}}$ with question sets $\X,\Y$, where $\psi_{AB}$ is a noisy EPR state, i.e.,  $\psi_A=\psi_B=\frac{\id_2}{2}$ and the maximal correlation $\rho=\rho\br{\psi_{AB}}<1$ as defined in Definition~\ref{def:maximalcorrelation}, there exists an explicitly computable bound $D=D\br{\abs{\X},\abs{\Y},\epsilon,\rho}$ such that it suffices for the players to share $D$ copies of $\psi_{AB}$ to achieve the probability of winning the game at least $\omega(G,\psi_{AB})-\epsilon$. In particular, one may choose $D=\exp\br{\mathrm{poly}\br{\abs{\X},\abs{\Y},\exp\br{\mathrm{poly}\br{\frac{1}{\epsilon},\frac{1}{1-\rho}}}}}$.
\end{theorem}
\begin{proof}
  Suppose the players share $n$ copies of $\psi_{AB}$ and employ the strategies $$\br{\set{P^x_a}_{x\in\X,a\in\set{0,1}},\set{Q^y_b}_{y\in\Y,b\in\set{0,1}}}$$ with the winning probability $\omega$. Without loss of generality, we assume that $\X=\set{1,2,\ldots,\abs{\X}}$ and $\Y=\set{1,2,\ldots,\abs{\Y}}$.

  Apply Theorem~\ref{thm:nijs} to the following two sequences of measurement operators
   $$\br{\underbrace{P_0^1,\ldots, P_0^1}_{|\Y|~\text{times}},\underbrace{P_0^2,\ldots, P_0^2}_{|\Y|~\text{times}},\ldots, \underbrace{P_0^{\abs{X}},\ldots, P_0^{\abs{X}}}_{|\Y|~\text{times}}}$$
   and
   $$\br{Q_0^1,\ldots, Q_0^{\abs{\Y}},Q_0^1,\ldots, Q_0^{\abs{\Y}},\ldots,Q_0^1,\ldots, Q_0^{\abs{\Y}}},$$
   with parameter $\epsilon\leftarrow\epsilon/8,s\leftarrow\abs{\X}\cdot\abs{\Y}$. Let $f$ and $g$ be the maps induced by Theorem~\ref{thm:nijs}. Set $\widetilde{P^x_0}\defeq f\br{P^x_0}$ and $\widetilde{Q^y_0}\defeq g\br{Q^y_0}$ for $x\in\X$ and $y\in\Y$. We claim that the strategy $$\br{\set{\widetilde{P^x_0},\widetilde{P^x_1}\defeq\id_D-\widetilde{P^x_0}}_{x\in\X},\set{\widetilde{Q^y_0},\widetilde{Q^y_1}\defeq\id_D-\widetilde{Q^y_0}}_{y\in\Y}}$$ wins the game with probability $\widetilde{\omega}\geq \omega-\epsilon$.
  Theorem~\ref{thm:nijs} item 1 guarantees that the operators above are valid measurements. Let $\nu_{xy}\br{a,b}\defeq\Tr\br{P^x_a\otimes Q^y_b}\psi_{AB}^{\otimes n}$ and $\widetilde{\nu}_{xy}\br{a,b}\defeq\Tr\br{\widetilde{P^x_a}\otimes \widetilde{Q^y_b}}\psi_{AB}^{\otimes D}$. From Theorem~\ref{thm:nijs}, for any $x\in\X$ and $y\in\Y$,
  \begin{eqnarray*}
        &&\abs{\nu_{xy}\br{0,0}+\nu_{xy}\br{0,1}-\widetilde{\nu}_{xy}\br{0,0}-\widetilde{\nu}_{xy}\br{0,1}} \leq\epsilon/8; \\
    &&\abs{\nu_{xy}\br{0,0}+\nu_{xy}\br{1,0}-\widetilde{\nu}_{xy}\br{0,0}-\widetilde{\nu}_{xy}\br{1,0}} \leq\epsilon/8;\\
    &&\abs{\nu_{xy}\br{0,0}-\widetilde{\nu}_{xy}\br{0,0}}\leq\epsilon/8.
  \end{eqnarray*}
  where the first and the second inequalities are implied by the item 2 in Theorem~\ref{thm:nijs}. The last inequality is from the item 3 in Theorem~\ref{thm:nijs}. The three inequalities above together imply that $\abs{\nu_{xy}\br{a,b}-\widetilde{\nu}_{xy}\br{a,b}}\leq\epsilon/4$ for any $a,b\in\set{0,1}$. Thus
  \begin{eqnarray*}
    &&\abs{\omega-\tilde{\omega}}=\abs{\sum_{xy}\mu\br{x,y}\br{\nu_{xy}\br{a,b}-\widetilde{\nu}_{xy}\br{a,b}}V(x,y,a,b)} \\&\leq&\sum_{xy}\mu\br{x,y}\sum_{a,b}\abs{\nu_{xy}\br{a,b}-\widetilde{\nu}_{xy}\br{a,b}}\\
    &\leq&\epsilon.
  \end{eqnarray*}

\end{proof}

\section{Open questions}\label{sec:openproblems}

In this work, we prove the decidability of mono-state binary games $\br{G,\psi}$ for any noisy EPR state $\psi$, by reducing the problem to the decidability of the quantum non-interactive simulations of joint distributions. Several interesting open questions are followed by this work.

An immediate open question is the decidability of general mono-state games. To remove the restrictions in the main result, it seems that several new ideas are required. For instance, if the shared state $\psi$ has maximal correlation $1$, we cannot use the the framework of the non-interactive simulation, because such a state possibly can generate any distribution without communication, such as EPR states. For the case that $\psi$ is a high dimensional state, we need a hypercontractive inequality for qudit quantum channels. For non-binary games, we need to work on several-matrix-variable functions possibly with much more involved calculation.

There are many other "tensored" quantities in quantum information theory and quantum complexity theory  not known to be computable, such as the regularisations of the various entanglement measures~\cite{Plbnio:2007:IEM:2011706.2011707}, quantum information complexity~\cite{Touchette:2015:QIC:2746539.2746613}, etc.

	\section{Markov super-operators, noise operators and maximal correlation}\label{sec:markov}
	Given $\psi\in\D_d$, $\psi>0$, $P,Q\in\M_d$, we define
	\[\innerproduct{P}{Q}_{\psi}\defeq\frac{1}{2}\Tr\br{P^{\dagger}Q+QP^{\dagger}}\psi,\]
	for any $P,Q\in\M_d$.

	The following fact can be easily verified.
		\begin{fact}\label{fac:innerproductpsi}
		$\innerproduct{\cdot}{\cdot}_{\psi}$ is an inner product and $\norm{\cdot}_{\psi}\defeq\sqrt{\innerproduct{P}{P}_{\psi}}$ is a norm whenever $\psi$ is a positive definite density operator.
	\end{fact}
	
 The inner product defined in Eq.~\eqref{eqn:innerproduct} can be viewed as the inner product above with $\psi=\frac{\id_d}{d}$, where $d$ is the dimension.
	For any integer $n>0$, $\psi^{\otimes n}$ induces an inner product in $\M_d^{\otimes n}$. To keep the notations short, the inner product is represented as
	\[\innerproduct{P}{Q}_{\psi}=\frac{1}{2}\Tr~\br{P^{\dagger}Q+QP^{\dagger}}\psi^{\otimes n},\]
	for $P,Q\in\M_d^{\otimes n}$.

For any quantum state $\psi>0$ in $\M_d$, we denote the space induced by the inner product defined in Eq.~\eqref{eqn:innerproduct} by $\br{\M_d,\psi}$. Note that $\norm{\id_d}_{\psi}=1$. Similar to Section~\ref{subsec:matrixspace}, we say an orthonormal basis in $\br{\M_d,\psi}$ is standard if all the operators are Hermitian and it contains $\id$ as an element. The Efron-Stein decomposition in Definition~\ref{def:efronstein} can be extended to $\br{\M_d,\psi}$ as well. Similar to Lemma~\ref{lem:efronsteinortho}, the terms in Efron-Stein decompositions are orthogonal to each other with respect to the inner product $\innerproduct{\cdot}{\cdot}_{\psi}$.
	\begin{definition}\label{def:markovoperator}
		Given quantum systems $A$ and $B$ and a bipartite quantum state $\psi_{AB}\in\D\br{A\otimes B}$,  we define the {\em Markov super-operator} $\T:\M\br{B}\rightarrow\M\br{A}$ as follows.
		\[\Tr\br{M^{\dagger}\otimes Q}\psi_{AB}=\innerproduct{M}{\T\br{Q}}_{\psi_A},\]
		for any $M\in\M\br{A}$ and $Q\in\M\br{B}$.
	\end{definition}

	\begin{lemma}\label{lem:markovoperator}
		Given quantum systems $A$, $B$ a bipartite state $\psi_{AB}$ and $Q\in\M\br{B}$, it holds that
		\[\T\br{Q}=2L\br{\psi_A,2\Tr_B\br{\id\otimes Q}\psi},\]
		where $L\br{\cdot,\cdot}$ is the solution to the Lyapunov equation given in Definition~\ref{def:sylvester}. In particular, if $\dim A=\dim B=2$ and $\psi_{AB}$ is a state satisfying that $\psi_A=\frac{\id_2}{2}$. Then $\T\br{Q}=2\Tr_B\br{\id\otimes Q}\psi_{AB}$.
	\end{lemma}
	\begin{proof}
		By Definition~\ref{def:markovoperator}, $\T\br{Q}$ must satisfy that
		\[\Tr M^{\dagger}\br{\Tr_B\br{\id\otimes Q}\psi}=\Tr M^{\dagger}\frac{\T\br{Q}\psi_A+\psi_A\T\br{Q}}{2}\] for any $M\in\M\br{A}$. Thus
		\begin{equation}\label{eqn:tb}
			\T\br{Q}\psi_A+\psi_A\T\br{Q}=2\Tr_B\br{\id\otimes Q}\psi_{AB}.
		\end{equation}
		We conclude the first part of the lemma.
		The second part follows from Eq.~\eqref{eqn:tb} with $\psi_A=\frac{\id_2}{2}$.
	\end{proof}

	\begin{definition}\label{def:bonamibeckner}
		For any quantum system $A$ with dimension $d$, a quantum state $\psi\in\D\br{A}$ with $\psi>0$ and $\rho\in[0,1]$, the noise operator $\T_{\rho}:\M\br{A}\rightarrow\M\br{A}$ on $\br{\M(A),\psi}$ is defined as follows. For any $M, P\in\M\br{A}$,
		\[\T_{\rho}\br{P}=\rho P+\br{1-\rho}\Tr P\psi\cdot\id_d.\]

For the space $\M\br{A^n}$, with slight abuse of notations, we define $\T_{\rho}\defeq\otimes_{i=1}^n\T_{\rho}$.
	\end{definition}
\noindent The noise operator $\T_{\rho}$ is also named {\em depolarizing channel}~\cite{NC00} in the quantum information theory, which is an analog of the {\em Bonami-Beckner operator} in Fourier analysis~\cite{10.2307/1970980,AIF_1970__20_2_335_0}.
	
	\begin{lemma}\label{lem:bonamibecknerdef}
	Given integers $d,n>0$, $\rho\in[0,1]$, space $\br{\M_d,\psi}$ with a standard orthonormal basis  $\B=\set{\B_i}_{i=0}^{d^2-1}$, the following holds.
		\begin{enumerate}
			\item For any $P\in\M_d^{\otimes n}$ with the Fourier expansion $P=\sum_{\sigma\in[d^2]_{\geq 0}^n}\widehat{P}\br{\sigma}\B_{\sigma}$, it holds that
			\[\T_{\rho}\br{P}=\sum_{\sigma\in[d^2]_{\geq 0}^n}\rho^{\abs{\sigma}}\widehat{P}\br{\sigma}\B_{\sigma}.\]
			\item For any $P\in\M_d^{\otimes n}$ $\norm{\T_{\rho}\br{P}}\leq\norm{P}$ and  $\norm{\T_{\rho}\br{P}}_{\psi}\leq \norm{P}_{\psi}$.
		\end{enumerate}
	\end{lemma}
	\begin{proof}
		Note that $\B_0=\id_d$. Item 1 follows from the definition directly.
		
		For item 2, we define $\T^{\br{i}}$ be the operator on $\M_d^{\otimes n}$ which applies $\T_{\rho}$ to the $i$-th system and leaves other systems untouched. Then $\T_{\rho}=\T_{\rho}^{(n)}\circ\cdots\circ\T^{(1)}_{\rho}$. From item 1, we have
		\begin{eqnarray*}
			&&\T^{\br{i}}_{\rho}\br{P}\defeq\sum_{\sigma:\sigma_i=0}\widehat{P}\br{\sigma}\B_{\sigma}+\rho\sum_{\sigma:\sigma_i\neq0}\widehat{P}\br{\sigma}\B_{\sigma}\\
			&=&\rho P+\br{1-\rho}\sum_{\sigma:\sigma_i=0}\widehat{P}\br{\sigma}\B_{\sigma}=\rho P+\br{1-\rho}P_{-i}\otimes\B_0^{(i)},
		\end{eqnarray*}
		where $\B_0^{(i)}$ means that it is in the $i$-th register.
	
		Note that the spectral norm of the first term is at most $\rho\norm{P}$. The spectral norm of the second term is at most $\br{1-\rho}\norm{P}$ by Lemma~\ref{lem:partialvariance} item 1.
		Hence $\norm{\T_{\rho}^{\br{i}}\br{P}}\leq\norm{P}$. Thus the first inequality in item 3 follows. To prove the second inequality, consider
		\[\norm{\T_{\rho}\br{P}}_{\psi}^2=\sum_{\sigma\in[d^2]_{\geq 0}^n}\rho^{2\abs{\sigma}}\abs{\widehat{P}\br{\sigma}}^2\leq\sum_{\sigma\in[d^2]_{\geq 0}^n}\abs{\widehat{P}\br{\sigma}}^2=\norm{P}_{\psi}^2.\]
	\end{proof}

	The notion of quantum maximal correlation introduced by Beigi~\cite{Beigi:2013}, which generalizes the maximal correlation coefficients~\cite{hirschfeld:1935,Gebelein:1941,Renyi1959} in classical information theory to the quantum setting, is crucial to our analysis.
	\begin{definition}[Maximal correlation]~\cite{Beigi:2013}\label{def:maximalcorrelation}
		Given quantum systems $A, B$ and a bipartite state $\psi_{AB}\in\D\br{A\otimes B}$, the maximal correlation of $\psi_{AB}$ is defined to be
		\[\rho\br{\psi_{AB}}\defeq\sup\set{\abs{\Tr\br{P^{\dagger}\otimes Q}\psi_{AB}}~:P\in\M\br{A}, Q\in\M\br{B},\Tr~P\psi_A=\Tr~Q\psi_B=0,\atop \norm{P}_{\psi_A}=\norm{Q}_{\psi_B}=1}.\]
	\end{definition}
\begin{fact}~\cite{Beigi:2013}\label{fac:maximalcorrlationone}
	Given quantum systems $A, B$ and a bipartite quantum state $\psi_{AB}$, it holds that
	\begin{enumerate}	
		\item $0\leq \rho\br{\psi_{AB}}\leq 1$.
		
		\item $\rho\br{\psi_{AB}}=1$ if and only if there exist local measurements $\set{M_A,\id-M_A}$ and $\set{N_B,\id-N_B}$ such that $0<\Tr\br{\psi_{AB}\br{M_A\otimes N_B}}<1$, and
		\[\Tr\br{\psi_{AB}\br{M_A\otimes\br{\id-N_B}}}=\Tr\br{\psi_{AB}\br{\br{\id-M_A}\otimes N_B}}=0.\]		
	\end{enumerate}
\end{fact}

The following proposition characterizes all the two-qubit states with maximal correlation being $1$.

\begin{prop}\label{prop:maximalcorrelationone}
	Given a bipartite state $\psi_{AB}\in\D\br{\M_2\times\M_2}$, $\rho\br{\psi_{AB}}=1$ if and only if there exist local unitaries $U_A$ and $V_B$ such that
	\begin{equation}\label{eqn:maximalcorrelation}
	\br{U_A^{\dagger}\otimes V_B^{\dagger}}\psi_{AB}\br{U_A\otimes V_B}=\sum_{a,b=0}^1 c_{a,b}\ketbratwo{aa}{bb},
	\end{equation}
	where $\br{c_{ab}}_{a,b\in\set{0,1}}$ satisfies that $\begin{pmatrix}c_{00} & c_{01}\\ c_{10} & c_{11}\end{pmatrix}$ is a density operator.
	
	Moreover, if $\psi_A=\psi_B=\frac{\id_2}{2}$, then $\rho\br{\psi_{AB}}=1$ if and only if there exist local unitaries $U_A$ and $V_B$ such that
	\[\br{U_A^{\dagger}\otimes V_B^{\dagger}}\psi^{AB}\br{U_A\otimes V_B}=p\ketbra{\Phi}+\br{1-p}\ketbra{\Psi},\]
	for $0\leq p\leq 1$, where $\ket{\Phi}=\frac{1}{\sqrt{2}}\br{\ket{00}+\ket{11}}$ and $\ket{\Psi}=\frac{1}{\sqrt{2}}\br{\ket{00}-\ket{11}}$.
\end{prop}

\begin{proof}
	
	Let $\set{M_A,\id-M_A}$ and $\set{N_B,\id-N_B}$ be the measurements induced by Fact~\ref{fac:maximalcorrlationone}. We may assume that
	$M_A=a_0\ketbra{0}+a_1\ketbra{1}$ and $N_B=b_0\ketbra{0}+b_1\ketbra{1}$ up to local unitaries.
	We claim that $M_A\neq\id_2$ and $N_B\neq\id_2$. Suppose $M_A=\id_2$. Then we have $0<\Tr\psi_BN_B<1$ and $1-\Tr\psi_BN_B=0$, contradicting to the fact that $\psi_B$ is a density operator. The argument for $N_B$ is similar. Thus we may further assume that $a_1\leq a_0\leq 1$ and $b_1\leq b_0\leq 1$ and $a_1, b_1<1, a_0,b_0>0$ Then $\Tr~\br{M_A\otimes\br{\id-M_B}}\psi_{AB}=0$ implies that $\bra{01}\psi_{AB}\ket{01}=0$. $\Tr~\br{\br{\id-M_A}\otimes M_B}=0$ implies that $\bra{10}\psi_{AB}\ket{10}=0$. We conclude Eq.~\eqref{eqn:maximalcorrelation}.
	The second part follows by elementary calculation.
\end{proof}

\begin{definition}\label{def:noisyepr}
	A bipartite state $\psi_{AB}\in\D(\M_2\times\M_2)$ is a noisy EPR state if $\psi_A=\psi_B=\frac{\id_2}{2}$ and its maximal correlation $\rho=\rho\br{\psi_{AB}}<1$.
\end{definition}
	Proposition~\ref{prop:maximalcorrelationone} gives a tight characterization of noisy EPR states. Probably the most interesting case is an EPR state with an arbitrary depolarizing noise $\epsilon>0$, i.e., $\br{1-\epsilon}\ketbra{\Phi}+\epsilon\frac{\id_2}{2}\otimes\frac{\id_2}{2}$, where $\ket{\Phi}=\frac{1}{\sqrt{2}}\br{\ket{00}+\ket{11}}$ is an EPR state. Beigi proved that the maximal correlation of this state is $1-\epsilon$ in~\cite{Beigi:2013}.

The following proposition provides a useful characterization of the quantum maximal correlation.
	\begin{prop}\label{prop:maximalvariance}
		Given quantum systems $A,B$ and a bipartite state $\psi_{AB}$, for any $Q\in\M\br{B}$,
\begin{equation}\label{eqn:max}
  \max\set{\abs{\Tr\br{P^{\dagger}\otimes Q}\psi_{AB}}: P\in\M\br{\P},\norm{P}_{\psi_A}=1}
\end{equation}
		is achieved by
		\[P^*=\frac{\T\br{Q}}{\norm{\T\br{Q}}_{\psi_A}},\]
		with the maximum value $\norm{\T\br{Q}}_{\psi_A}$, where $\T:\M\br{B}\rightarrow\M\br{A}$ is the Markov super-operator in Definition~\ref{def:markovoperator}.
		Thus,
		\begin{equation*}
			\rho\br{\psi_{AB}}=\max\set{\norm{\T\br{Q}}_{\psi_A}: Q\in\M\br{B}, \innerproduct{\id}{Q}_{\psi_B}=0,\norm{Q}_{\psi_B}=1}.
		\end{equation*}
		In particular, if $\psi_A=\frac{\id_{d_A}}{d_A}$ and $\psi_B=\frac{\id_{d_B}}{d_B}$, then
		\begin{equation}\label{eqn:TQ}
			\rho\br{\psi_{AB}}=\max\set{\nnorm{\T\br{Q}}_2:\Tr~Q=0, \nnorm{Q}_2=1},
		\end{equation}

		where $d_A=\dim\br{A}$ and $d_B=\dim\br{B}$.
		
		Moreover, the maximal correlation in Definition~\ref{def:maximalcorrelation} can be achieved by a pair of Hermitian operators $\br{P,Q}$ if $\psi_A=\frac{\id_{d_A}}{d_A}$ and $\psi_B=\frac{\id_{d_B}}{d_B}$.
	\end{prop}
	
	\begin{proof}
		The proof follows closely to the one of Lemma 2.8 in~\cite{Mossel:2010}. Let $P\in\M\br{A}$ achieves the maximum in Eq.~\eqref{eqn:max}. Then it satisfies that $\norm{P}_{\psi_A}=1$. Write $P=\alpha P^*+\beta P'$, where $\abs{\alpha}^2+\abs{\beta}^2=1$, $\norm{P'}_{\psi_A}=1$ and $\innerproduct{P^*}{P'}_{\psi_A}=0$. By the definition of the Markov super-operator
		\[0=\innerproduct{P'}{\T\br{Q}}_{\psi_A}=\Tr\br{\br{P'}^{\dagger}\otimes Q}\psi_{AB}.\]
		So we should set $\abs{\alpha}=1.$ Moreover,
		\[\Tr\br{\T\br{Q}^{\dagger}\otimes Q}\psi_{AB}=\norm{\T\br{Q}}^2_{\psi_A}.\]	
		
		To prove that the maximal correlation in Definition~\ref{def:maximalcorrelation} can be achieved by a pair of Hermitian operators $\br{P,Q}$, it suffices to prove that the maximum in Eq.~\eqref{eqn:TQ} can be achieved by a Hermitian matrix $Q$ by Lemma~\ref{lem:markovoperator} and Lemma~\ref{lem:lyapunovsol}. Suppose $Q=Q_1+\mathrm{i}\cdot Q_2$ achieves the maximum in Eq.~\eqref{eqn:TQ} with Hermitian matrices $Q_1\neq 0$ and $Q_2\neq 0$. Then $\Tr~Q_1\psi_B=\Tr~Q_2\psi_B=0$ and $1=\nnorm{Q}_2^2=\nnorm{Q_1}_2^2+\nnorm{Q_2 }_2^2$ Then
		\[\nnorm{\T\br{Q}}_2=\br{\frac{\nnorm{\T\br{Q_1}}_2^2+\nnorm{\T\br{Q_2}}_2^2}{\nnorm{Q_1}_2^2+\nnorm{Q_2}_2^2}}^{\frac{1}{2}}\leq\max\set{\nnorm{\T\br{\frac{Q_1}{\nnorm{Q_1}_2}}}_2,\nnorm{\T\br{\frac{Q_2}{\nnorm{Q_2}_2}}}_2}.\]
		Thus at least one of $\frac{Q_1}{\nnorm{Q_1}_2}$ and $\frac{Q_2}{\nnorm{Q_2}_2}$ also achieves the maximum in Eq.~\eqref{eqn:TQ}

	\end{proof}

	\begin{lemma}\label{lem:efronsteinortho}
	Given quantum systems $A, B$ with $\dim A=d_A$ and $\dim B= d_B$, a bipartite quantum state $\psi_{AB}\in\D\br{A\otimes B}$, let $\set{\A_{\sigma}}_{\sigma\in[d_A^2]_{\geq 0}}$ and $\set{\B_{\sigma}}_{\sigma\in[d_B^2]_{\geq 0}}$ be standard orthonormal basis in $\M\br{A}$ and $\M\br{B}$ with respect to the inner product $\innerproduct{\cdot}{\cdot}_{\psi_A}$ and $\innerproduct{\cdot}{\cdot}_{\psi_B}$, respectively. It holds that
	\[\Tr\br{\A_{\sigma}\otimes\B_{\tau}}\psi_{AB}^{\otimes n}=0,\]
	whenever $\supp{\sigma}\neq\supp{\tau}$. Thus
	\[\Tr\br{A[S]\otimes Q[T]}\psi_{AB}^{\otimes n}=0,\]
	whenever $S\neq T$.
\end{lemma}
\begin{proof}
	It follows from the equalities $\Tr\br{\A_0\otimes\B_{\tau}}\psi_{AB}=\Tr\br{\A_{\sigma}\otimes\B_0}\psi_{AB}=0$ whenver $\sigma\neq 0$ and $\tau\neq0$.
\end{proof}

	\begin{prop}\label{prop:markovenfronstein}
	Given an integer $n>0$, quantum systems $A$ and $B$, a bipartite quantum state $\psi_{AB}$ , $Q\in\M\br{B^n}$ and $S\subseteq[n]$, it holds that
	\[\T\br{Q[S]}=\T\br{Q}[S],\]
	where $\T:\M\br{B}\rightarrow\M\br{A}$ is the Markov super-operator in Definition~\ref{def:markovoperator} with respect to $\psi_{AB}$.
\end{prop}
\begin{proof} It suffices to show that
	\[\innerproduct{X}{\T\br{Q[S]}}_{\psi_A}=\innerproduct{X}{\T\br{Q}[S]}_{\psi_A},\] for any $X\in\M\br{A^n}$.
	By the definition,
	\[\text{LHS}=\Tr\br{X^{\dagger}\otimes Q[S]}\psi_{AB}^{\otimes n}.\]
	By Definition~\ref{def:efronstein} and Proposition~\ref{prop:enfronsteinortho}, we have
	\[\text{RHS}=\innerproduct{X[S]}{\T\br{Q}}_{\psi_A}=\Tr\br{X[S]^{\dagger}\otimes Q}\psi_{AB}^{\otimes n}.\]
	By Lemma~\ref{lem:efronsteinortho} and Definition~\ref{def:efronstein},
	\[\Tr\br{X[S]^{\dagger}\otimes Q}\psi_{AB}^{\otimes n}=\Tr\br{X^{\dagger}\otimes Q[S]}\psi_{AB}^{\otimes n}=\Tr\br{X[S]^{\dagger}\otimes Q[S]}\psi_{AB}^{\otimes n}.\]
\end{proof}

	\begin{prop}\label{prop:markovoperatornorm}
		Given an integer $n>0$, quantum systems $A, B$, a bipartite quantum state $\psi_{AB}\in\D\br{A\otimes B}$, $Q\in\M\br{B^n}$ and $S\subseteq[n]$. It holds that
		\[\norm{\T\br{Q[S]}}_{\psi_A}\leq\rho^{|S|}\norm{Q[S]}_{\psi_B},\]
		where $\rho=\rho\br{\psi_{AB}}$.
	\end{prop}
	\begin{proof}
	We may assume that $Q=Q[S]$ without loss of generality. It suffices to show the case that $S=[n]$. Let $d_A=\dim\br{A}, d_B=\dim\br{B}$, $\set{\A_i}_{i\in[d_A]_{\geq 0}}$ and $\set{\B_i}_{i\in[d_B]_{\geq 0}}$ be standard orthonormal basis in $\M\br{A}$ and $\M\br{B}$, respectively. For $r\in[n]$, set $\T^{\br{r}}\defeq \id_{\M\br{A}}^{\otimes\br{r-1}}\otimes\T\otimes\id_{\M\br{B}}^{\otimes\br{n-r}}$, $\psi^{(r)}\defeq\psi_A^{\otimes r}\otimes\psi_B^{\otimes(n-r)}$, $Q^{(0)}\defeq Q$ and $Q^{\br{r}}=\T^{(r)}\br{Q^{(r-1)}}$, where $\id_{\M\br{A}}$ and $\id_{\M\br{B}}$ are the identity maps mapping $\M\br{A}$ to $\M\br{A}$ and $\M\br{B}$ to $\M\br{B}$, respectively. Note that $Q^{(n)}=\T\br{Q}$. Hence it suffices to show that
		\begin{equation}\label{eqn:markovnormeq1}
		\Tr~\br{Q^{(r)}}^{\dagger}Q^{(r)}\psi^{(r)}\leq\rho^2 \Tr~\br{Q^{(r-1)}}^{\dagger}Q^{(r-1)}\psi^{(r-1)},
		\end{equation}
		and
		\begin{equation}\label{eqn:markovnormeq2}
		Q^{(r)}=Q^{(r)}[n],
		\end{equation}
		where $Q^{(r)}[n]$ is defined by expanding $Q^{(r)}$ over $\set{\A_{\sigma_{\leq r}}\otimes\B_{\sigma_{> r}}}_{\sigma\in[d_A^2]_{\geq 0}^r\times[d_B^2]_{\geq 0}^{n-r}}$,
		because $\T=\T^{(n)}\circ\cdots\circ\T^{\br{1}}$.
		Let $\psi_A=\sum_{i=1}^{d_A}\lambda_i^A\ketbra{u_i}$ and $\psi_B=\sum_{i=1}^{d_B}\lambda_i^B\ketbra{v_i}$ be spectral decompositions of $\psi_A$ and $\psi_B$, respectively. For any $s\in[d_A]^{r-1}\times[d_B]^{n-r}$, we define
		$$\theta^{(r)}_s\defeq\sqrt{\lambda_{s_1}^A\cdot \lambda_{s_2}^A\cdots\lambda_{s_{r-1}}^A\cdot\lambda_{s_{r+1}}^B\cdot\lambda_{s_{r+2}}^B\cdots\lambda_{s_n}^B},$$
		and
		\[\ket{w_s}\defeq\ket{u_{s_1}}\otimes\ldots\otimes\ket{u_{s_{r-1}}}\otimes\ket{v_{s_{r+1}}}\otimes\ldots\otimes\ket{v_{s_n}}.\]
		 For any $s, t\in[d_A]^{r-1}\times[d_B]^{n-r}$, we define
	 \[P^{(r)}_{s,t}\defeq\theta_t^{(r)}\br{\bra{w_s}\otimes\id}Q^{(r)}\br{\ket{w_t}\otimes\id} ~\mbox{and}~ Q^{(r-1)}_{s,t}\defeq\theta_t^{(r)}\br{\bra{w_s}\otimes\id}Q^{(r-1)}\br{\ket{w_t}\otimes\id},\]
	 where $\ket{w_s}$ and $\ket{w_t}$ lie in the registers $\set{1,\ldots, r-1,r+1,\ldots, n}$.
	 Then
	 \begin{equation}\label{eqn:tpq}
	 	P^{(r)}_{s,t}=\T\br{Q_{s,t}^{(r-1)}}.
	 \end{equation}
		And we have
		\begin{eqnarray*}
		&&\innerproduct{\id_B}{Q_{s,t}^{(r-1)}}_{\psi_B}=0,
		\end{eqnarray*}
		by the induction $Q^{(r-1)}=Q^{(r-1)}[n]$.
		By Proposition~\ref{prop:maximalvariance},  $\norm{P^{(r)}_{s,t}}_{\psi_A}\leq\rho\norm{Q^{(r-1)}_{s,t}}_{\psi_B}$.
		Consider
		\[\sum_{s,t}\norm{P^{(r)}_{s,t}}_{\psi_A}^2=\sum_{s,t}\Tr\br{P^{(r)}_{s,t}}^{\dagger}P^{(r)}_{s,t}\psi_A=\Tr\br{Q^{(r)}}^{\dagger}Q^{(r)}\psi^{(r)},
		\]
	and
	\[\sum_{s,t}\norm{Q^{(r-1)}_{s,t}}_{\psi_B}^2=\sum_{s}\br{\theta^{(r)}_{s}}^2\Tr\br{Q^{(r-1)}}^{\dagger}Q^{(r-1)}\psi^{(r-1).}\]
		Combining with Eq.~\eqref{eqn:tpq} and Proposition~\ref{prop:maximalvariance}, we conclude Eq.~\eqref{eqn:markovnormeq1}.
		
		For Eq.~\eqref{eqn:markovnormeq2}, compute
		\begin{eqnarray*}
		&&\widehat{Q^{(r)}}\br{\sigma}=\widehat{\T^{\br{r}}\br{Q^{(r-1)}}}\br{\sigma}=\innerproduct{\X_{\sigma_{\leq r}}\otimes\Y_{\sigma_{>r}}}{\T^{\br{r}}\br{Q}}_{\psi^{(r)}}\\
		&=&\sum_{\tau:\abs{\tau}=n}\widehat{Q}\br{\tau}\innerproduct{\X_{\sigma_{\leq r}}\otimes\Y_{\sigma_{>r}}}{\T^{\br{r}}\br{\X_{\tau_{<r}}\otimes\Y_{\tau_{\geq r}}}}_{\psi^{(r)}}\\
		&=&\sum_{\tau:\abs{\tau}=n,\tau_{-r}=\sigma_{-r}}\widehat{Q}\br{\tau}\innerproduct{\X_{\sigma_r}}{\T\br{\Y_{\tau_r}}}_{\psi_A}.
		\end{eqnarray*}
		Note that
		\[\innerproduct{\X_0}{\T\br{\Y_{\tau_r}}}_{\psi_A}=\Tr\br{\id\otimes\Y_{\tau_r}}\psi_{AB}=0,\]
		as $\abs{\tau}=n$. Therefore, $\widehat{\T^{\br{r}}\br{Q}}\br{\sigma}=0$ if $\abs{\sigma}<n$. We conclude Eq.~\eqref{eqn:markovnormeq2}.
	\end{proof}
	A key property of the (classical) maximal correlation is {\em tensorization}, which states that the maximal correlation of multiple independent identical copies of a distribution is same as the one of one copy. The same property also holds for the quantum maximal correlation shown by Beigi in~\cite{Beigi:2013}. Here we provide a different proof.
	\begin{fact}~\cite{Beigi:2013}\label{prop:maximalcorrelationtensorisation}Given quantum systems $A, B$ with $\dim A=d_A$ and $\dim B= d_B$, a bipartite quantum state $\psi_{AB}\in\D\br{A\otimes B}$,  it holds that
		\[\rho\br{\psi_{AB}^{\otimes n}}=\rho\br{\psi_{AB}}.\]
	\end{fact}
	\begin{proof}
		Given $Q\in\M\br{\B^n}$ with $\innerproduct{\id}{Q}_{\psi^B}=0$ and $\norm{Q}_{\psi^B}=1$. Using the Efron-Stein decomposition $Q=\sum_{S\neq\emptyset}Q[S],$
		we have
		\begin{eqnarray*}
		&&\norm{\T\br{Q}}_{\psi^B}^2\\
		&=&\norm{\sum_{S\neq\emptyset}\T\br{Q[S]}}_{\psi^B}^2\quad\quad\mbox{(linearlity of $\T$)}\\
		&=&\norm{\sum_{S\neq\emptyset}\T\br{Q}[S]}_{\psi^B}^2\quad\quad\mbox{(Proposition~\ref{prop:markovenfronstein})}\\
		&=&\sum_{S\neq\emptyset}\norm{\T\br{Q}[S]}_{\psi^B}^2\quad\quad\mbox{(orthogonality of the Efron-Stein decomposition)}\\
		&\leq&\sum_{S\neq\emptyset}\rho^{2\abs{S}}\norm{Q[S]}_{\psi^B}^2\quad\quad\mbox{(Proposition~\ref{prop:markovoperatornorm})}\\
		&\leq&\rho^2\sum_{S\neq\emptyset}\norm{Q[S]}_{\psi^B}^2\\
		&=&\rho^2\norm{Q}_{\psi_B}^2\quad\quad\mbox{(orthogonality of the Efron-Stein decomposition)}\\
		&=&\rho^2.
		\end{eqnarray*}
	From Proposition~\ref{prop:maximalvariance}, $\rho\br{\psi^{\otimes n}_{AB}}\leq\rho\br{\psi_{AB}}$. The other direction is trivial.
	\end{proof}

	\section{Smoothing operators}\label{sec:smoothoperators}
	The main lemma in this section is the following.
	
		\begin{lemma}\label{lem:smoothing of strategies}
		Given parameters $0\leq\rho<1$, $0<\delta<1$, integer $n>0$ a noisy EPR state $\psi_{AB}$  with the maximal correlation $\rho=\rho\br{\psi_{AB}}$, there exist $d=d\br{\rho,\delta}$ and a map $f:\H_2^{\otimes n}\rightarrow \H_2^{\otimes n}, $such that for any $P\in\H_2^{\otimes n}, Q\in\H_2^{\otimes n}$ satisfying $0\leq P\leq \id$ and $0\leq Q\leq \id$, the operators $P^{(1)}=f\br{P}$ and $Q^{(1)}=f\br{Q}$ satisfy that

\begin{enumerate}
  \item $0\leq P^{(1)}\leq \id~\mbox{and}~0\leq Q^{(1)}\leq \id;$
  \item $\Tr~P^{(1)}=\Tr~P~\mbox{and}~\Tr~Q^{(1)}=\Tr~Q;$
  \item $\nnorm{P^{(1)}}_2\leq\nnorm{P}_2~\mbox{and}~\nnorm{Q^{(1)}}_2\leq\nnorm{Q}_2;$
  \item $\abs{\Tr\br{P^{(1)}\otimes Q^{(1)}}\psi^{\otimes n}_{AB}-\Tr\br{P\otimes Q}\psi^{\otimes n}_{AB}}\leq\delta,$
  \item $\nnorm{\br{P^{(1)}}^{>d}}^2_2\leq\delta~\mbox{and}~\nnorm{\br{Q^{(1)}}^{>d}}^2_2\leq\delta,$
\end{enumerate}
In particular, we can take $d=\frac{2\log^2\frac{1}{\delta}}{C\br{1-\rho}\delta}$ for some constant $C$.
	\end{lemma}
	
	Before proving Lemma~\ref{lem:smoothing of strategies}, we need the following lemma, which generalizes the smoothing of the strategies in the classical setting in~\cite{Mossel:2010} to the quantum setting.
	
	\begin{lemma}\label{lem:Tsmooth}
		Given a noisy EPR state $\psi_{AB}$ with the maximal correlation $\rho=\rho(\psi_{AB})<1$, a parameter $0<\epsilon<1$ and operators $P\in\H_2^{\otimes n}, Q\in\H_2^{\otimes n}$, let $\gamma$ be chosen sufficiently close to $0$ so that
		\[\gamma\leq1-\br{1-\epsilon}^{\log\rho/\br{\log\epsilon+\log\rho}}.\]
		Then
		\[\abs{\Tr\br{P\otimes Q}\psi_{AB}^{\otimes n}-\Tr\br{\T_{1-\gamma}\br{P}\otimes\T_{1-\gamma}\br{Q}}\psi_{AB}^{\otimes n}}\leq 2\epsilon\sqrt{\var{P}\var{Q}}.\]
		In particular, there exists an absolute constant $C$ such that it suffices to take
		\[\gamma=C\frac{\br{1-\rho}\epsilon}{\log\br{1/\epsilon}}.\]
	\end{lemma}
	\begin{proof}
		Note that $\var{\T_{1-\gamma}\br{Q}}\leq\var{Q}$ by Lemma~\ref{lem:variance} and Lemma~\ref{lem:bonamibecknerdef}. Thus it suffices to show
		\[\abs{\Tr\br{P\otimes Q}\psi_{AB}^{\otimes n}-\Tr\br{P\otimes \T_{1-\gamma}\br{Q}}\psi_{AB}^{\otimes n}}=\abs{\Tr\br{P\otimes\br{\id-\T_{1-\gamma}}\br{Q}}\psi_{AB}^{\otimes n}}\leq\epsilon\sqrt{\var{P}\var{Q}}.\]
		By Lemma~\ref{lem:bonamibecknerdef} item 1, we have
		\[\br{\id-\T_{1-\gamma}}\br{Q[S]}=\br{1-\br{1-\gamma}^{\abs{S}}}Q[S].\]
		Using Proposition~\ref{prop:markovenfronstein} and Proposition~\ref{prop:markovoperatornorm},
		\begin{equation}\label{eqn:tqs}
			\twonorm{\T\br{Q[S]}}=\twonorm{\T\br{Q}[S]}\leq\rho^{\abs{S}}\twonorm{Q[S]},
		\end{equation}
		and $\innerproduct{\T\br{Q[S]}}{P[S']}=0$ if $S\neq S'$.
		Let $\T'\defeq\T\circ\br{\id_{\M\br{B}}-\T_{1-\gamma}}$.
		From Definition~\ref{def:markovoperator},
		\[\Tr\br{P\otimes\br{\id-\T_{1-\gamma}}\br{Q}}\psi_{AB}^{\otimes n}=\innerproduct{P}{\T'\br{Q}}\]	
		Combining Lemma~\ref{lem:bonamibecknerdef} and Eq.~\eqref{eqn:tqs} and the choice of $\gamma$,
		\[\twonorm{\T'\br{Q[S]}}\leq\min\br{\rho^{\abs{S}},1-\br{1-\gamma}^{\abs{S}}}\twonorm{Q[S]}\leq\epsilon\twonorm{Q[S]}.\]
		From Lemma~\ref{lem:bonamibecknerdef},  Proposition~\ref{prop:enfronsteinortho} and  Proposition~\ref{prop:markovenfronstein}, $\innerproduct{\T'\br{Q[S]}}{P[S']}=0$ if $S\neq S'$.
		Therefore,
		\begin{eqnarray*}
		&&\abs{\innerproduct{P}{\T'\br{Q}}}=\abs{\sum_{S\neq\emptyset}\innerproduct{P[S]}{\T'\br{Q[S]}}}\leq\sqrt{\var{P}}\sqrt{\frac{1}{2^n}\sum_{S\neq\emptyset}\twonorm{\T'\br{Q[S]}}^2}\\
		&\leq&\epsilon\sqrt{\var{P}}\sqrt{\frac{1}{2^n}\sum_{S\neq\emptyset}\twonorm{\br{Q[S]}}^2}\leq\epsilon\sqrt{\var{P}\var{Q}},
		\end{eqnarray*}
		where the last inequality is from the orthogonality of the Efron-Stein decomposition.
	\end{proof}	

We are now ready to prove Lemma~\ref{lem:smoothing of strategies}.

	\begin{proof}[Proof of Lemma~\ref{lem:smoothing of strategies}]
		Given $\rho$ and $\delta$, we choose $\epsilon=\delta/2$ and $\gamma$ in Lemma~\ref{lem:Tsmooth}. We choose $d$ sufficiently large such that $(1-\gamma)^{2d}\leq\delta$, that is, $d=\frac{\log\frac{1}{\delta}}{2\gamma}$. Given $P,Q$ as in Lemma~\ref{lem:smoothing of strategies}, we set $P^{(1)}\defeq\T_{1-\gamma}\br{P}$ and $Q^{(1)}\defeq\T_{1-\gamma}\br{Q}$. From Definition~\ref{def:bonamibeckner}, $0\leq P^{(1)}\leq 1$ and $0\leq Q^{(1)}\leq 1$. $\Tr~P^{(1)}=\Tr~P$ and $\Tr~Q^{(1)}=\Tr~Q$ follows by Lemma~\ref{lem:bonamibecknerdef} item 1. By Lemma~\ref{lem:bonamibecknerdef} item 2, $\nnorm{P^{(1)}}_2\leq\nnorm{P}_2$ and $ \nnorm{Q^{(1)}}_2\leq\nnorm{Q}_2$. From the inequalities that $\var{P}\leq\nnorm{P}_2^2\leq 1,\var{Q}\leq\nnorm{Q}_2^2\leq 1$ due to Lemma~\ref{lem:partialvariance} and Lemma~\ref{lem:Tsmooth}, we get
		\[\abs{\Tr\br{P\otimes Q}\psi_{AB}^{\otimes n}-\Tr\br{P^{(1)}\otimes Q^{(1)}}\psi_{AB}^{\otimes n}}\leq 2\epsilon=\delta.\]
		Also, $\widehat{P^{(1)}}\br{\sigma}=\br{1-\gamma}^{\abs{\sigma}}\widehat{P}\br{\sigma}$ and $\widehat{Q^{(1)}}\br{\sigma}=\br{1-\gamma}^{\abs{\sigma}}\widehat{Q}\br{\sigma}$. Thus, we get that
		\begin{eqnarray*}
		&&\sum_{\abs{\sigma}>d}\widehat{P^{(1)}}\br{\sigma}^2\leq\br{1-\gamma}^{2d}\sum_{\abs{\sigma}>d}\widehat{P}\br{\sigma}^2\leq\br{1-\gamma}^{2d}\nnorm{P}_2^2\leq\delta;\\
		&&\sum_{\abs{\sigma}>d}\widehat{Q^{(1)}}\br{\sigma}^2\leq\br{1-\gamma}^{2d}\sum_{\abs{\sigma}>d}\widehat{Q}\br{\sigma}^2\leq\br{1-\gamma}^{2d}\nnorm{Q}_2^2\leq\delta,
		\end{eqnarray*}
	where the second inequalities in both equations are from Fact~\ref{fac:basicfourier} item 3.
	\end{proof}

\section{Joint regularity lemma}\label{sec:jointregular}
In the proof the decidability of the classical non-interactive joint simulation~\cite{7782969,doi:10.1137/1.9781611975031.174,Ghazi:2018:DRP:3235586.3235614}, a subset $H$ of $[n]$ with bounded size are chosen such that all the coordinates not in $H$ are low influential even if the values of the coordinates in $H$ are fixed. However, the same strategy cannot be applied in the quantum setting because there is no common basis among the all coordinates. Instead of fixing values, we expand the operators in a proper chosen standard orthonormal basis. Before getting into the details, we introduce notion of {\em correlated matrices}.

	\begin{definition}\label{def:covariancematrix}
		Given quantum systems $A$ and $B$ with dimension $d_A$ and $d_B$, respectively, and a bipartite quantum state $\psi_{AB}$, Let $\A=\set{\A_i}_{i\in[d_A^2]_{\geq 0}}$  and $\B=\set{\B_i}_{i\in[d_B^2]_{\geq 0}}$ be standard orthonormal basis in the space $\br{\M\br{A},\psi_A}$ and $\br{\M\br{B},\psi_B}$ defined in Section~\ref{sec:markov}, respectively.
		The correlation matrix of $\br{\psi_{AB}, \A,\B}$ is defined as
		\[\mathsf{Corr}\br{\psi_{AB},\A,\B}_{i,j}\defeq\Tr\br{\A_i\otimes\B_j}\psi_{AB}.\]
		for $i\in[d_A]_{\geq 0}, j\in[d_B]_{\geq 0}$. The correlation matrix of $\br{\psi_{AB}^{\otimes n},\A,\B}$ is defined to be
		\[\mathsf{Corr}\br{\psi_{AB}^{\otimes n},\A,\B}_{\sigma,\tau}\defeq\Tr\br{\A_{\sigma}\otimes\B_{\tau}}\psi_{AB}^{\otimes n},\]
		for $\sigma\in[d_A]^n_{\geq 0}$ and $\tau\in[d_B]^n_{\geq 0}$.
	\end{definition}
The following lemma follows by the definition.
\begin{lemma}\label{lem:convariancetensor}
	Given an integer $n>0$, quantum systems $A$ and $B$ with dimension $d_A$ and $d_B$, respectively, and a bipartite quantum state $\psi_{AB}$, let $\A=\set{\A_i}_{i\in[d_A^2]_{\geq 0}}$  and $\B=\set{\B_i}_{i\in[d_B^2]_{\geq 0}}$ be standard orthonormal basis in $\br{\M\br{A},\psi_A}$ and $\br{\M\br{B},\psi_B}$, respectively. It holds that \[\mathsf{Corr}\br{\psi_{AB}^{\otimes n},\A,\B}=\mathsf{Corr}\br{\psi_{AB},\A,\B}^{\otimes n}.\]
\end{lemma}
%\begin{proof}
%    Let $M_n\defeq \mathsf{Corr}\br{\psi_{AB}^{\otimes n},\A^{\otimes n},\B^{\otimes n}}$. Then for any $\sigma\in[d_A^2]_{\geq 0}^n$ and $\tau\in[d_B^2]_{\geq 0}^n$
%	\begin{eqnarray*}
%		&&\br{M_1}^{\otimes n}\br{\sigma,\tau}=\prod_{i=1}^{n}M_1\br{\sigma_i,\tau_i}\\
%		&=&\prod_{i=1}^n\br{\Tr\br{\A_{\sigma_i}\otimes\B_{\tau_i}}\psi_{AB}-\br{\Tr\A_{\sigma_i}\psi_A}\br{\Tr\B_{\tau_i}\psi_B}}\\
%		&=&\prod_{i=1}^n\br{\Tr\br{\A_{\sigma_i}\otimes\B_{\tau_i}}\psi_{AB}-\delta_{\sigma_i,0}\cdot\delta_{\tau_i,0}}.
%	\end{eqnarray*}
%	And
%	\[M_n\br{\sigma,\tau}=\prod_{i=1}^n\Tr~\br{\A_{\sigma_i}\otimes\B_{\tau_i}}\psi_{AB}-\delta_{\sigma,\mathbf{0}}\cdot\delta_{\tau,\mathbf{0}}.\]
%	Thus it is obvious $M_n\br{\sigma,\tau}=\br{M_1}^{\otimes n}\br{\sigma,\tau}$ if $\sigma\neq\mathbf{0}$ and $\tau\neq\mathbf{0}$. If $\sigma=\mathbf{0}$ and $\tau\neq\mathbf{0}$ or $\tau=\mathbf{0}$ and $\sigma\neq\mathbf{0}$, then $M_n\br{\sigma,\tau}=\br{M_1}^{\otimes n}\br{\sigma,\tau}=0$. The remaining case is that $\M_n\br{\mathbf{0},\mathbf{0}}=M_1\br{0,0}=1$
%\end{proof}
\begin{lemma}\label{lem:normofM}
	Given quantum systems $A$ and $B$ with dimension $d_A$ and $d_B$, respectively, and a noisy EPR state $\psi_{AB}$, for any standard orthonormal basis $\A=\set{\A_i}_{i\in[d_A^2]_{\geq 0}}$  and $\B=\set{\B_i}_{i\in[d_B^2]_{\geq 0}}$ in $\br{\M\br{A},\psi_A}$ and $\br{\M\br{B},\psi_B}$, respectively. It holds that $s_1\br{\mathsf{Corr}\br{\psi_{AB},\A,\B}}=1$ and $s_2\br{\mathsf{Corr}\br{\psi_{AB},\A,\B}}=\rho$, where $\rho=\rho\br{\psi_{AB}}$ and $s_i\br{\cdot}$ is the $i$-th largest singular value.

By Fact~\ref{fac:unitarybasis}, there exist standard orthonormal basis $\A=\set{\A_i}_{i\in[d_A^2]_{\geq 0}}$  and $\B=\set{\B_i}_{i\in[d_B^2]_{\geq 0}}$ in $\br{\M\br{A},\psi_A}$ and $\br{\M\br{B},\psi_B}$, respectively, such that
\[\mathsf{Corr}\br{\psi_{AB},\A,\B}_{i,j}=\begin{cases}c_i~&\mbox{if $i=j$}\\0~&\mbox{otherwise},\end{cases}\]
where $c_1=1, c_2=\rho\br{\psi_{AB}}$ and $c_1\geq c_2\geq c_3\geq\ldots$.
\end{lemma}
\begin{proof}
	We assume that both the dimensions of $A$ and $B$ are $d$. The arguments are similar when the dimensions of the two systems are different. Let $M\defeq\mathsf{Corr}\br{\psi_{AB},\A,\B}$. It is easy to verify that $M_{0,i}=M_{i,0}=0$ for all $i\in[d^2]_{\geq 0}$. Note that $M$ is a real matrix. Thus we set $M=U^{\dagger}DV$ to be a singular value decomposition of $M$ where $U$ and $V$ are orthogonal matrices and $U_{0,0}=V_{0,0}=1, U_{0,i}=U_{i,0}=V_{0,i}=V_{i,0}=0$ for $1\leq i\leq d^2-1$. Let $\P_i\defeq\sum_{j=1}^{d^2-1}U_{ij}\A_j$ and $\Q_i\defeq\sum_{j=1}^{d^2-1}V_{ji}^{\dagger}\B_j$ for $1\leq i\leq d^2-1$. Then
	$\Tr~\P_i=\Tr~\Q_i=0$. $\var{\P_i}=\var{\Q_i}=1$. Thus by the definition of the quantum maximal correlation,
	\begin{eqnarray*}
		\rho&\geq&\Tr~\br{\P_i\otimes \Q_{i'}}\psi_{AB}\\
		&=&\sum_{j,k=1}^{d^2-1}U_{ij}V^{\dagger}_{ki'}\br{\A_j\otimes\B_k}\psi_{AB}\\
		&=&\sum_{j,k=1}^{d^2-1}U_{ij}V_{ki'}^{\dagger}M_{jk}=\br{UMV^{\dagger}}_{ii'}\\
		&=&\delta_{i,i'}D_{ii}.
	\end{eqnarray*}
Hence $\norm{M}=1$ and $s_2\br{M}\leq\rho$.

From Proposition~\ref{prop:maximalvariance}, we assume that $P, Q$ are two Hermitian operators in $\H\br{A}, \H\br{B}$ which achieve $\rho$. Then from Definition~\ref{def:maximalcorrelation}, both $\set{\id_{d_A}, P}$ and $\set{\id_{d_B}, Q}$ can be extended to orthonormal basis, say $\set{\P_i}_{i=0}^{d^2-1}$ and $\set{\Q_i}_{i=0}^{d^2-1}$ where $\P_1=P$ and $\Q_1=Q$. Let $M$ be the corresponding correlated matrix. Then $M\br{1,1}=1, M\br{2,2}=\rho$. Thus, $s_2\br{M}=\rho$.
\end{proof}

We reach the main lemma in this section.
\begin{lemma}\label{lem:regular}
	Given a noisy EPR state $\psi_{AB}$ and operators $P\in\H_2^{\otimes n}, Q\in\H_2^{\otimes n}$, and parameters $d,\delta,\epsilon>0$ satisfying that $\nnorm{P}_2\leq 1, \nnorm{Q}_2\leq 1$ and $\nnorm{P^{>d}}^2_2\leq\delta$ and $\nnorm{Q^{>d}}^2_2\leq\delta$, let $\br{c_i}_{i=0}^3$ be the singular values of $\mathsf{Corr}\br{\psi_{AB}}$ in non-increasing order and $\A=\set{\A_i}_{i=0}^3$ and $\B=\set{\B_i}_{i=0}^3$ be the standard orthonormal basis induced by Lemma~\ref{lem:normofM}. Then there exists a subset $H\subseteq[n]$ of size $h\defeq\abs{H}\leq\frac{2d}{\epsilon}$ such that for any $i\notin H$, $\influence_i\br{P^{\leq d}}\leq\epsilon$, $\influence_i\br{Q^{\leq d}}\leq\epsilon$ and
	\[\Tr\br{P\otimes Q}\psi_{AB}^{\otimes n}=\sum_{\sigma\in[4]_{\geq 0}^h}c_{\sigma}\Tr\br{P_{\sigma}\otimes Q_{\sigma}}\psi_{AB}^{\otimes(n-h)},\]
	where $c_{\alpha}\defeq\prod_{i=1}^tc_{\alpha_i}$ for any $\alpha\in[4]_{\geq 0}^h$;
	and
	$$P_{\sigma}=\sum_{\tau\in[4]_{\geq 0}^n:\tau_H=\sigma}\widehat{P}\br{\tau}\A_{\tau}$$
	and
	$$Q_{\sigma}=\sum_{\tau\in[4]_{\geq 0}^n:\tau_H=\sigma}\widehat{Q}\br{\tau}\B_{\tau}.$$
\end{lemma}
\begin{proof}
	Set $H=\set{i:\influence_i\br{P^{\leq d}}\geq\epsilon~\mbox{or}~\influence_i\br{Q^{\leq d}}\geq\epsilon}$. From Lemma~\ref{lem:partialvariance} item 4, $\abs{H}\leq\frac{2d}{\epsilon}$.  Expanding $P$ and $Q$ in terms of the basis $\A$ and $\B$, respectively, we conclude the result.
\end{proof}

\section{Fr\'echet Derivative and Taylor expansion}\label{sec:derivative}

In this section, we derive a Taylor expansion of matrix functions, for which we adopts Fr\'echet derivatives.  The Fr\'echet derivatives are derivatives defined on Banach spaces. In this paper, we only concern about the Fr\'echet derivatives on matrix spaces. Readers may refer to~\cite{Coleman} for a more thorough treatment.

\begin{definition}\label{def:frechetderivative}
	Given a map $f:\M_d\rightarrow\M_d$ and Hermitian matrices $P, Q$, the Fr\'echet derivative of $f$ at $P$ with direction $Q$ is defined to be
	\[Df\br{P}\br{Q}\defeq\frac{d}{dt}f\br{P+tQ}|_{t=0}.\]
	The $k$-th order Fr\'echet derivative of $f$ at $P$ with direction $\br{Q_1,\ldots, Q_k}$ is defined to be
	\[D^kf\br{P}\br{Q_1,\ldots, Q_k}\defeq\frac{d}{dt}D^{k-1}f\br{P+tQ_k}\br{Q_1,\ldots, Q_{k-1}}|_{t=0}.\]
\end{definition}

The Fr\'echet derivatives share many common properties with the derivatives in Euclidean spaces, such as linearity, composition rules, etc. The most basic properties are summarized in Appendix~\ref{sec:frechet}.

The function we are interested in this paper is the function $\zeta:\reals\rightarrow\reals$ defined as follows, which was also studied in~\cite{Mossel:2010,MosselOdonnell:2010}.

	\begin{eqnarray}
	&&\zeta\br{x}\defeq\begin{cases}x^2~&\mbox{if $x\leq 0$}\\ (x-1)^2~&\mbox{if $x\geq 1$}\\ 0~&\mbox{otherwise}\end{cases}.\label{eqn:zeta}
	\end{eqnarray}

The main reason to study $\zeta\br{\cdot}$ is that $\Tr~\zeta\br{\cdot}$ characterizes the minimum $\twonorm{\cdot}$-distance between a Hermitian matrix and the set of all positive semidefinite matrices with eigenvalues at most $1$.
\begin{lemma}\label{lem:closedelta}
	Given integer $d>0$, $M\in\H_d$, $\Delta=\set{X\in\H_d:0\leq X\leq\id}$, let $\R$ be the rounding map of $\Delta$ with respect to the distance $\twonorm{\cdot}$. It holds that
	\[\Tr~\zeta\br{M}=\twonorm{M-\R\br{M}}^2.\]
\end{lemma}
\begin{proof}
	Without loss of generality, we assume that $M$ is diagonal. Let $$X_0=\arg\min\set{\twonorm{M-X}^2:X\in\Delta}.$$
	The lemma follows by easy calculation if $X_0$ is also diagonal. We now show that $X_0$ is indeed diagonal. Note that	\[\twonorm{M-X_0}^2=\Tr~X_0^2+\Tr~M^2-2\sum_i\lambda_i\br{M}X_0\br{i,i}.\]
	It is known that $\br{X_0\br{1,1},\ldots, X_0\br{d,d}}$ is majorized by $\br{\lambda_1\br{X},\ldots,\lambda_n\br{X}}$~\cite{Bhatia}. Namely, $\sum_{j=1}^i\lambda_j\br{X}\geq\sum_{j=1}^iX_0\br{j,j}$. Note that $X_0\geq 0$. It is easy to verify that
	\[\sum_i\lambda_i\br{M}X_0\br{i,i}\leq\sum_i\lambda_i\br{M}\lambda_i\br{X}.\]
	The equality is achieved only if $X_0$ is also diagonal.
	
\end{proof}
Note that the function $\zeta$ is in $\C^1$ but not in $\C^2$. We define a $\C^2$-approximation of $\zeta$ in the following, whose Fr\'echet derivatives are easier to calculate comparing with the  $\C^{\infty}$-approximation considered in~\cite{Mossel:2010,MosselOdonnell:2010}.
	For any $0\leq\lambda<1$, define $\zeta_{\lambda}:\reals\rightarrow\reals$ to be \footnote{The definition of $\zeta_{\lambda}$ is derived from the following construction.

 \[\psi\br{x}\defeq\begin{cases}
 \frac{1}{2}~&\mbox{if $-1\leq x\leq 1$}\\
 0~&\mbox{otherwise.}
 \end{cases}\]
 $f\br{x}=x\cdot\id_{x\geq 0}$.
 \[\psi_{\lambda}\br{x}\defeq\psi\br{x/\lambda}/\lambda,\]
 for $0<\lambda\leq\frac{1}{2}$.
 \[f_{\lambda}\defeq f*\psi_{\lambda}.\]
 \[\zeta_{\lambda}\br{x}\defeq\begin{cases}
 2\int_{-\infty}^{-x}f_{\lambda}\br{t}dt~&\mbox{if $x\leq1/2$}\\
 2\int_{-\infty}^{x-1}f_{\lambda}\br{t}dt~&\mbox{if $x\geq1/2$}.
 \end{cases}\]
 }

	\begin{eqnarray}
	&&\zeta_{\lambda}\br{x}\defeq\begin{cases}
	x^2+\frac{1}{3}\lambda^2~&\mbox{if $x\leq-\lambda$}\\
	\frac{\br{\lambda-x}^3}{6\lambda}~&\mbox{if $-\lambda\leq x\leq \lambda$}\\
	0~&\mbox{if $\lambda\leq x\leq 1-\lambda$}\\
	\frac{\br{x-1+\lambda}^3}{6\lambda}~&\mbox{if $1-\lambda\leq x\leq1+\lambda$}\\
	\br{1-x}^2+\frac{1}{3}\lambda^2~&\mbox{if $x\geq1+\lambda$}.
	\end{cases}\label{eqn:zetalambda}
	\end{eqnarray}

The following lemma can be verified by elementary calculus.

 \begin{lemma}\label{lem:zeta}
 	For any $0<\lambda<1$, it holds that
 	\begin{enumerate}
 		\item $\norm{\zeta_{\lambda}-\zeta}_{\infty}\leq 4\lambda^2.$
 		\item $\zeta_{\lambda}\in\mathcal{C}^2$. $\zeta_{\lambda}''$ is a piecewisely linear function. $\zeta'''_{\lambda}\br{\cdot}$ exists in $\reals$ except for finite points. And $\abs{\zeta'''_{\lambda}\br{x}}\leq \frac{1}{\lambda}$ for any $x$ where $\zeta'''_{\lambda}\br{x}$ exists. 	
 	\end{enumerate}
 \end{lemma}

 \begin{lemma}\label{lem:zetataylor}
  For any Hermitian $P, Q$, $0<\lambda<1/2$, it holds that
 	\begin{equation}\label{eqn:zetataylor}
 	\Tr~\zeta_{\lambda}\br{P+ Q}=\Tr~\zeta_{\lambda}\br{P}+\Tr~ D\zeta_{\lambda}\br{P}\br{Q}+\frac{1}{2}\Tr~D^2\zeta_{\lambda}\br{P}\br{Q}+O\br{\frac{\twonorm{Q}\norm{Q}_4^2}{\lambda}}.
 	\end{equation}
 \end{lemma}
The proof is via calculating each order of the Fr\'echet derivatives combining with several inequalities in matrix analsis, which is  deferred to Appendix~\ref{sec:zetataylor}.

The following lemma enables us to remove the part of an operator with low $2$-norm without changing the value of $\Tr~\zeta\br{\cdot}$ much. The proof is also deferred to Appendix~\ref{sec:zetataylor}.
\begin{lemma}\label{lem:zetaadditivity}
  For any Hermitian matrices $P$ and $Q$, it holds that $\abs{\Tr\br{\zeta\br{P+Q}-\zeta\br{P}}}\leq4\br{\twonorm{P}\twonorm{Q}+\twonorm{Q}^2}$.
\end{lemma}

   	\section{Hypercontractive inequality for random operators}\label{sec:hypercontractive}

	We first introduce a noise operator $\Gamma_{\rho}$ acting on $L^2\br{\M_2^{\otimes h},\gamma_n}$, which is a hybrid of the Ornstein-Uhlenbeck operator $U_{\rho}$ in Definition~\ref{def:ornstein} and the noise operator $\T_{\rho}$ in Definition~\ref{def:bonamibeckner}. Recall that any $\mathbf{P}\in L^2\br{\M_2^{\otimes h},\gamma_n}$ can be expressed as
\begin{equation}\label{eqn:randomoperator}
\mathbf{P}=\sum_{\sigma\in[4]_{\geq 0}^h}p_{\sigma}\br{\mathbf{g}}\B_{\sigma},
    \end{equation}
    	where $\set{\B_i}_{i=0}^3$ is a standard orthonormal basis in $\M_2$, $p_{\sigma}\in L^2\br{\complex,\gamma_n}$ and $\mathbf{g}\sim \gamma_n.$

\begin{definition}\label{def:gamma}
Given $0\leq\rho\leq 1$ and integers $h,n\geq 0$, the noise operator $\Gamma_{\rho}:L^2\br{\M_2^{\otimes h},\gamma_n}\rightarrow L^2\br{\M_2^{\otimes h},\gamma_n}$ is defined to be	\[\Gamma_{\rho}\br{\mathbf{P}}=\sum_{\sigma\in[4]_{\geq 0}^h}\br{U_{\rho}p_{\sigma}}\br{\mathbf{g}}\T_{\rho}\br{\B_{\sigma}},\]
where $\set{\B_i}_{i=0}^3$ is a standard orthonormal basis in $\M_2$.\footnote{From Lemma~\ref{fac:unitarybasis}, the definition is independent of the choices of the basis as the Gaussian distribution $\gamma_n$ is invariant under orthogonal transformation.}
	\end{definition}
The following lemma directly follows from  Fact~\ref{fac:vecfun} and Lemma~\ref{lem:bonamibecknerdef} item 1.

\begin{lemma}\label{lem:gammaoperator}
  Given $0\leq\rho\leq 1$, integers $n,h\geq 0$ and a random operator $\mathbf{P}\in L^2\br{\M_2^{\otimes h},\gamma_n}$ that has expansion in Eq.~\eqref{eqn:randomoperator}, it holds that
  \begin{equation}\label{eqn:gamma}
    \Gamma_{\rho}\br{\mathbf{P}}=\sum_{\sigma\in[4]_{\geq0}^h}\sum_{\tau\in\mathbb{Z}_{\geq 0}^n}\rho^{\abs{\sigma}+\abs{\tau}}\widehat{p_{\sigma}}\br{\tau}H_{\tau}\br{\mathbf{g}}\B_{\sigma},
  \end{equation}
  where $\abs{\sigma}=\abs{\set{i:\sigma_i\neq 0}}$ and $\abs{\tau}=\sum_i\tau_i$ and $H_{\tau}$'s are the Hermite polynomials defined in Eqs.~\eqref{eqn:hermitebasis}\eqref{eqn:hermite}.
\end{lemma}

The main result in this section is a hypercontractive inequality for random operators stated as follows.

\begin{lemma}\label{lem:hypercontractivity}
	Given $0\leq\rho\leq \frac{1}{\sqrt{3}}$, integers $n,h\geq 0$, for any multilinear random operator $\mathbf{P}\in L^2\br{\M_2^{\otimes h},\gamma_n}$, it holds that
	\[N_{4}\br{\Gamma_{\rho}\br{\mathbf{P}}}\leq N_2\br{\mathbf{P}},\]
where $\Gamma_{\rho}$ is the noise operator acting on $L^2\br{\M_2^{\otimes h},\gamma_n}$ defined in Definition~\ref{def:gamma} and $N_p$ is the normalized $p$-norm of a random operator in Definition~\ref{def:randop}.
\end{lemma}	

The concept of random operators is a hybrid of the operators in $\M_2^{\otimes h}$ and the the functions in the Gaussian space $L^2\br{\complex,\gamma_n}$. Thus the proof of Lemma~\ref{lem:hypercontractivity} is a combination of the hypercontractive inequality for unital quantum operators due to King~\cite{King2014} and the hypercontractive inequality for Gaussian variables~\cite{PawelWolff2007,MosselOdonnell:2010}. The proof is deferred to the end of this section.
The following is an application of Lemma~\ref{lem:hypercontractivity}.
\begin{lemma}\label{lem:Xhypercontractivity}
	Given integers $h,n\geq 0$, for any multilinear random operator $\mathbf{P}\in L^2\br{\M_2^{\otimes h},\gamma_n}$  with the associated vector-valued polynomial $p=\br{p_{\sigma}}_{\sigma\in[4]_{\geq 0}^h}$, it holds that
	\[N_4\br{\mathbf{P}}\leq 3^{d/2}N_2\br{\mathbf{P}},\]
where $d=\max_{\sigma\in[4]_{\geq 0}^h}\br{\deg\br{p_{\sigma}}+\abs{\sigma}}$.
\end{lemma}
\begin{proof}
	Suppose
$\mathbf{P}=\sum_{\sigma\in[4]_{\geq 0}^h}p_{\sigma}\br{\mathbf{g}}\B_{\sigma},$
    	where $\set{\B_i}_{i=0}^3$ is a standard orthonormal basis in $\M_2$, $p_{\sigma}\in L^2\br{\complex,\gamma_n}$ is multilinear and $\mathbf{g}\sim \gamma_n$. Set $$\mathbf{P}^{=i}\defeq\sum_{\br{\sigma,\tau}\in[4]_{\geq0}^h\times\mathbb{Z}_{\geq 0}^n:\atop\abs{\sigma}+\abs{\tau}=i}p_{\sigma}\br{\mathbf{g}}\B_{\sigma}.$$
	Applying Lemma~\ref{lem:hypercontractivity},
	\begin{eqnarray*}
		&&N_4\br{\mathbf{P}}=N_4\br{\Gamma_{\frac{1}{\sqrt{3}}}\br{\sum_{i=1}^{d}\br{\sqrt{3}}^i\mathbf{P}^{=i}}}\leq N_2\br{\sum_{i=1}^{d}\br{\sqrt{3}}^i\mathbf{P}^{=i}}=\expec{}{\nnorm{\sum_{i=1}^{d}\br{\sqrt{3}}^i\mathbf{P}^{=i}}^2_2}^{1/2}.
	\end{eqnarray*}
Note that
\[\expec{}{\Tr\br{\mathbf{P}^{=i}}^{\dagger}\mathbf{P}^{=j}}=0,\]
whenever $i\neq j$.
Therefore,
\[N_4\br{\mathbf{P}}=\br{\sum_{i=1}^{d}\br{\sqrt{3}}^{2i}\expec{}{\nnorm{\mathbf{P}^{=i}}^2_2}}^{\frac{1}{2}}\leq3^{d/2}\br{\sum_{i=1}^{d}\expec{}{\nnorm{\mathbf{P}^{=i}}^2_2}}^{\frac{1}{2}}=3^{d/2}N_2\br{\mathbf{P}}.\]
\end{proof}

%\begin{definition}\label{def:ptoqnorm}
%	For any map $\Phi:\M_d\rightarrow\M_d$ and $1\leq p\leq q\leq\infty$, the {\em  $p$-to-$q$ norm} of $\Phi$ is defined to be
%	\[\norm{\Phi}_{p\rightarrow q}\defeq\max_{M\in\M_d}\frac{\norm{\Phi\br{M}}_q}{\norm{M}_p}.\]
%\end{definition}

The following fact is a direct consequence of a hypercontractive inequality for qubit channel due to King.
\begin{fact}\label{fac:hyperqubit}~\cite{King2014}
	Let $\T_{\rho}:\M_2\rightarrow\M_2$ be a noise operator in Definition~\ref{def:bonamibeckner}. For any integer $n\geq 1$, $0\leq \rho\leq \sqrt{\frac{1}{3}}$ and $M\in\M_2^{\otimes n}$, it holds that
	\[\norm{\Gamma_{\rho}\br{M}}_4\leq\twonorm{M}.\]
\end{fact}

%
%\begin{fact}\label{lem:hann}~\cite{King2003}
%	Given a block-matrix $P=\begin{pmatrix}
%	P_{11} & P_{12}\\
%	P_{21} & P_{22}
%	\end{pmatrix}$. It holds that
%	\[\norm{P}_4\leq\norm{\begin{pmatrix}
%		\norm{P_{11}}_4 & \norm{P_{12}}_4\\
%		\norm{P_{21}}_4 & \norm{P_{22}}_4
%		\end{pmatrix}}_4.\]
%\end{fact}

The following fact is a well known hypercontractive inequality in Gaussian space.
\begin{fact}\footnote{The result in~\cite{PawelWolff2007,MosselOdonnell:2010} is for $f\in L^2\br{\reals, \gamma_n}$. But it can be extended to $L^2\br{\complex,\gamma_n}$ easily.
	
Let $f=f_1+\mathrm{i}\cdot f_2\in L^2\br{\complex,\gamma_n}$, where $f_i\in L^2\br{\reals,\gamma_n}$. Then \begin{eqnarray*}
	&&\norm{U_{\rho}f}_4=\expec{\mathbf{x}\sim \gamma_n}{\abs{\br{U_{\rho}f}\br{\mathbf{x}}}^4}^{\frac{1}{4}}\\
	&=&\expec{\mathbf{x}\sim \gamma_n}{\abs{\br{U_{\rho}f_1}\br{\mathbf{x}}}^4+\abs{\br{U_{\rho}f_2}\br{\mathbf{x}}}^4+2\abs{\br{U_{\rho}f_1}\br{\mathbf{x}}}^2\abs{\br{U_{\rho}f_2}\br{\mathbf{x}}}^2}^{\frac{1}{4}}\\
	&\leq&\br{\norm{U_{\rho}f_1}_4^4+\norm{U_{\rho}f_2}_4^4+2\norm{U_{\rho}f_1}_4^2\norm{U_{\rho}f_2}_4^2}^{\frac{1}{4}}\quad\quad\mbox{(Cauchy-Schwarz inequality)}\\
	&\leq&\br{\norm{f_1}_2^4+\norm{f_2}_2^4+2\norm{f_1}_2^2\norm{f_2}_2^2}^{\frac{1}{4}}\quad\quad\mbox{(Hypercontractive inequality in $L^2\br{\reals,\gamma_n}$)}\\
	&=&\br{\twonorm{f_1}^2+\twonorm{f_2}^2}^{\frac{1}{2}}=\twonorm{f}.
\end{eqnarray*}}\label{fac:gaussianhypercontractivity}~\cite{PawelWolff2007,MosselOdonnell:2010}
	For any $0\leq\rho\leq\frac{1}{\sqrt{3}}$, it holds that
	\[\sup_{f\in L^2\br{\complex, \gamma_n}}\norm{U_{\rho}f}_4\leq\twonorm{f},\]
	where $U_{\rho}$ is the Ornstein-Uhlenbeck operator acting on $L^2\br{\complex,\gamma_n}$ Definition~\ref{def:ornstein}.
\end{fact}

The following lemma is a generalization of Fact~\ref{fac:gaussianhypercontractivity} for technical purposes.

\begin{lemma}\label{lem:gausshyper}
Given $p_1,\ldots p_n\in L^2\br{\complex,\gamma_n}$, it holds that
\[\expec{\mathbf{x}\sim \gamma_n}{\br{\sum_{i=1}^n \abs{\br{U_{\rho}p_i}\br{\mathbf{x}}}^2}^2}^{\frac{1}{4}}\leq\expec{\mathbf{x}\sim \gamma_n}{\sum_{i=1}^n\abs{p_i\br{\mathbf{x}}}^2}^{\frac{1}{2}}.\]
\end{lemma}
\begin{proof}
  Let $q_i\defeq U_{\rho}p_i$. Then
  \begin{eqnarray*}
  	&&\expec{\mathbf{x}\sim \gamma_n}{\br{\sum_i \abs{\br{U_{\rho}p_i}\br{\mathbf{x}}}^2}^2}^{\frac{1}{4}}\\
  	&=&\br{\sum_{i=1}^n\expec{\mathbf{x}\sim \gamma_n}{\abs{q_i\br{\mathbf{x}}}^4}+\sum_{i\neq j}\expec{\mathbf{x}}{\abs{q_i\br{\mathbf{x}}^2q_j\br{\mathbf{x}}^2}}}^{\frac{1}{4}}\\
  	&\leq&\br{\sum_{i=1}^n\norm{q_i}_4^4+\sum_{i\neq j}\norm{q_i}_4^2\norm{q_j}_4^2}^{\frac{1}{4}}\quad\quad\mbox{(Cauchy-Schwarz inequality)}\\
  	&\leq&\br{\sum_{i=1}^n\twonorm{p_i}^4+\sum_{i\neq j}\twonorm{p_i}^2\twonorm{p_j}^2}^{\frac{1}{4}}\quad\quad\mbox{(Fact~\ref{fac:gaussianhypercontractivity})}\\
  	&=&\br{\sum_i\twonorm{p_i}^2}^{\frac{1}{2}}\\
  	&=&\expec{\mathbf{x}\sim \gamma_n}{\sum_{i=1}^n\abs{p_i\br{\mathbf{x}}}^2}^{\frac{1}{2}}.
  \end{eqnarray*}
\end{proof}

It is now ready to prove Lemma~\ref{lem:hypercontractivity}.
\begin{proof}[Proof of Lemma~\ref{lem:hypercontractivity}]
	 Let $\mathbf{P}=\sum_{\sigma\in[4]^h_{\geq 0}}p_{\sigma}\br{\mathbf{g}}\B_{\sigma}$, where
	$\set{\B_i}_{i=0}^3$ is a standard orthonormal basis.
	Set $\mathbf{Q}=\sum_{\sigma\in[4]^h_{\geq 0}}\br{U_{\rho}p_{\sigma}}\br{\mathbf{g}}\B_{\sigma}$. Then by the definition of $\Gamma_{\rho}$,
	\[\Gamma_{\rho}\br{\mathbf{P}}=\T_{\rho}\br{\mathbf{Q}}.\]
\noindent Using Fact~\ref{fac:hyperqubit},
\begin{eqnarray*}
	&&N_4\br{\Gamma_{\rho}\br{\mathbf{P}}}=\expec{}{\nnorm{\T_{\rho}\br{\mathbf{Q}}}_4^4}^{\frac{1}{4}}\leq\expec{}{\nnorm{\mathbf{Q}}_2^4}^{\frac{1}{4}}.
\end{eqnarray*}
%	for any $M\in\M_2^{\otimes h-1}$. Thus
%	\begin{eqnarray*}
%		N_4\br{\br{\Omega\otimes\T_{\rho}}\br{\mathbf{P}}}\leq \expec{}{\twonorm{\begin{pmatrix}
%					\norm{\mathbf{Q}_{11}}_2 & \norm{\mathbf{Q}_{12}}_2\\
%					\norm{\mathbf{Q}_{21}}_2 & \norm{\mathbf{Q}_{22}}_2
%			\end{pmatrix}}^4}^{\frac{1}{4}}.
%	\end{eqnarray*}
	Let $p_{ij}\in L^2\br{\complex,\gamma_n}$ and $q_{ij}\in L^2\br{\complex,\gamma_n}$ be the entries of $\mathbf{P}$ and $\mathbf{Q}$, respectively, for $1\leq i,j\leq 2^h$. Then $q_{ij}=U_{\rho}p_{ij}$.
	Notice that
	\begin{eqnarray*}
		&&N_4\br{\mathbf{Q}}=\expec{}{\nnorm{\mathbf{Q}}_2^4}^{\frac{1}{4}}=\expec{\mathbf{x}\sim \gamma_n}{\br{\sum_{ij}\abs{q_{ij}\br{\mathbf{x}}}^2}^2}^{\frac{1}{4}}\leq\expec{\mathbf{x}\sim \gamma_n}{\sum_{ij}\abs{p_{ij}\br{\mathbf{x}}}^2}^{\frac{1}{2}}=N_2\br{\mathbf{P}},
	\end{eqnarray*}
where the inequality is from Lemma~\ref{lem:gausshyper}. We conclude the result.
	
\end{proof}
\section{Quantum invariance principle}\label{sec:invariance}
Throughout this section we define the following joint random variables.
\[\br{\g_{1,0},\g_{1,1},\g_{1,2},\g_{1,3},\ldots, \g_{n,0},\g_{n,1},\g_{n,2},\g_{n,3}}\sim\br{\set{1}\times \gamma_3}^{\otimes n}.\]
For any $0\leq i\leq n$, define the hybrid basis and the hybrid random operators
\begin{align}
&\X_{\sigma}^{\br{i}}\defeq\br{\mathbf{g}_{1,\sigma_1}\P_0}\otimes\br{\mathbf{g}_{2,\sigma_2}\P_0}\otimes\ldots\otimes\br{\mathbf{g}_{i,\sigma_i}\P_0}\otimes \P_{\sigma_{>i}}\label{eqn:hybridxi}\\
&\mathbf{M}^{\br{i}}\defeq\sum_{\sigma\in[4]^n}\widehat{M}\br{\sigma}\X^{\br{i}}_{\sigma}.\label{eqn:hybridmi}
\end{align}
\begin{lemma}\label{lem:mg}
	$\mathbf{M}^{(i)}$ is independent of the choices of the basis. Namely, for any standard orthonormal basis $\set{\B_i}_{i=0}^{3}$ in $\M_2$ and $M=\sum_{\sigma\in[d]^n}\lambda_{\sigma}\B_{\sigma}$, set $\mathbf{N}=\sum_{\sigma\in[4]_{\geq 0}^n}\lambda_{\sigma}\br{\prod_{j=1}^i\g_{j,\sigma_j}}\B_0^{\otimes i}\otimes \B_{\sigma>i}$. Then $\mathbf{N}=\mathbf{M}^{\br{i}}.$
\end{lemma}
\begin{proof}
		From Fact~\ref{fac:unitarybasis}, all orthonormal basis are equivalent up to orthogonal transformations. The lemma follows from the well known fact that the Gaussian distribution $\gamma_n$ is invariant under any orthogonal transformation.		
\end{proof}
\begin{lemma}\label{lem:hybrid}
	For any integer $n>0$ and $0\leq i\leq n-1$ and $M\in\M_2^{\otimes n}$, it holds that
	\begin{align*}
		\abs{\expec{}{\Tr\zeta_{\lambda}\br{\mathbf{M}^{\br{i+1}}}-\Tr\zeta_{\lambda}\br{\mathbf{M}^{\br{i}}}}}
		\leq O\br{\frac{2^n3^d}{\lambda}\influence_{i+1}\br{M}^{3/2}},
	\end{align*}
	where $d=\deg M$.
\end{lemma}
\begin{proof} Note that

	\begin{align*}
	&\mathbf{M}^{\br{i}}=\sum_{\sigma:\sigma_{i+1}=0}\widehat{M}\br{\sigma}\X_{\sigma}^{\br{i}}+\sum_{\sigma:\sigma_{i+1}\neq 0}\widehat{M}\br{\sigma}\X_{\sigma}^{\br{i}},\\
	&\mathbf{M}^{\br{i+1}}=\sum_{\sigma:\sigma_{i+1}= 0}\widehat{M}\br{\sigma}\X_{\sigma}^{\br{i+1}}+\sum_{\sigma:\sigma_{i+1}= 0}\widehat{M}\br{\sigma}\X_{\sigma}^{\br{i+1}},
	\end{align*}
	and
	\[\sum_{\sigma:\sigma_{i+1}=0}\widehat{M}\br{\sigma}\X_{\sigma}^{\br{i}}=\sum_{\sigma:\sigma_{i+1}= 0}\widehat{M}\br{\sigma}\X_{\sigma}^{\br{i+1}}.\]
	Set
	\begin{align*}
	&\mathbf{A}\defeq\sum_{\sigma:\sigma_{i+1}=0}\widehat{M}\br{\sigma}\X_{\sigma}^{\br{i}}\\
	&\mathbf{B}\defeq\sum_{\sigma:\sigma_{i+1}\neq 0}\widehat{M}\br{\sigma}\X_{\sigma}^{\br{i}},\\
	&\mathbf{C}\defeq\sum_{\sigma:\sigma_{i+1}\neq 0}\widehat{M}\br{\sigma}\X_{\sigma}^{\br{i+1}}.
	\end{align*}
	Then we have
	\begin{align*}
	\mathbf{M}^{\br{i}}=\mathbf{A}+\mathbf{B};~\mathbf{M}^{\br{i+1}}=\mathbf{A}+\mathbf{C}.
	\end{align*}
	From Lemma~\ref{lem:zetataylor},
	\begin{align*} &\abs{\expec{}{\Tr~\zeta_{\lambda}\br{\mathbf{M}^{\br{i+1}}}-\Tr~\zeta_{\lambda}\br{\mathbf{M}^{\br{i}}}}}\\
	\leq&\abs{\expec{}{\br{\Tr~D\zeta_{\lambda}\br{\mathbf{A}}\br{\mathbf{C}}+\frac{1}{2}\Tr~D^2\zeta_{\lambda}\br{\mathbf{A}}\br{\mathbf{C}}}-\atop\br{\Tr~D\zeta_{\lambda}\br{\mathbf{A}}\br{\mathbf{B}}+\frac{1}{2}\Tr~D^2\zeta_{\lambda}\br{\mathbf{A}}\br{\mathbf{B}}}}}\\
&+O\br{\expec{}{\frac{\twonorm{\mathbf{C}}\norm{\mathbf{C}}_4^2}{\lambda}}}+O\br{\expec{}{\frac{\twonorm{\mathbf{B}}\norm{\mathbf{B}}_4^2}{\lambda}}}\\
&\leq\abs{\expec{}{\br{\Tr~D\zeta_{\lambda}\br{\mathbf{A}}\br{\mathbf{C}}+\frac{1}{2}\Tr~D^2\zeta_{\lambda}\br{\mathbf{A}}\br{\mathbf{C}}}-\atop\br{\Tr~D\zeta_{\lambda}\br{\mathbf{A}}\br{\mathbf{B}}+\frac{1}{2}\Tr~D^2\zeta_{\lambda}\br{\mathbf{A}}\br{\mathbf{B}}}}}\\
&+O\br{\frac{2^n}{\lambda}\br{N_2\br{\mathbf{C}}N_4\br{\mathbf{C}}^2+N_2\br{\mathbf{B}}N_4\br{\mathbf{B}}^2}}\quad\quad\mbox{(Cauchy-Schwarz inequality)}\\
&=O\br{\frac{2^n}{\lambda}\br{N_2\br{\mathbf{C}}N_4\br{\mathbf{C}}^2+N_2\br{\mathbf{B}}N_4\br{\mathbf{B}}^2}}.\quad\quad\mbox{(Claim~\ref{claim:1})}
	\end{align*}
	 Applying Lemma~\ref{lem:Xhypercontractivity}, we have
	\[\abs{\expec{}{\Tr~\zeta_{\lambda}\br{\mathbf{M}^{\br{i+1}}}-\Tr~\zeta_{\lambda}\br{\mathbf{M}^{\br{i}}}}}\leq O\br{\frac{3^d2^n}{\lambda}\br{N_2\br{\mathbf{B}}^3+N_2\br{\mathbf{C}}^3}}.\]
Notice that
	\[N_2\br{\mathbf{B}}=N_2\br{\mathbf{C}}=\br{\sum_{\sigma:\sigma_{i+1}\neq 0}\abs{\widehat{M}\br{\sigma}^2}}^{1/2}=\influence_{i+1}\br{M}^{1/2}.\]
	Therefore,
	\[\abs{\expec{}{\Tr\zeta_{\lambda}\br{\mathbf{M}^{\br{i+1}}}-\Tr\zeta_{\lambda}\br{\mathbf{M}^{\br{i}}}}}\leq O\br{\frac{2^n3^d}{\lambda}\influence_{i+1}\br{M}^{3/2}}.\]
\end{proof}

\begin{claim}\label{claim:bc}
		\begin{eqnarray}
	&&\expec{}{\Tr~\mathbf{B}f\br{\mathbf{A}}}=\expec{}{\Tr~\mathbf{C} f\br{\mathbf{A}}}\label{eqn:bc}\\
	&&\expec{}{\Tr~\mathbf{B}f\br{\mathbf{A}}\mathbf{B}g\br{\mathbf{A}}}=\expec{}{\Tr~\mathbf{C}f\br{\mathbf{A}}\mathbf{C}g\br{\mathbf{A}}}\label{eqn:bc2}
	\end{eqnarray}
	for any $f,g\in L^2\br{\reals, \gamma_1}$.
\end{claim}

\begin{claim}\label{claim:1}
	It holds that
	\[\expec{}{\br{\Tr~D\zeta_{\lambda}\br{\mathbf{A}}\br{\mathbf{C}}}}=\expec{}{\br{\Tr~D\zeta_{\lambda}\br{\mathbf{A}}\br{\mathbf{B}}}};\]
	\[\expec{}{\br{\Tr~D^2\zeta_{\lambda}\br{\mathbf{A}}\br{\mathbf{C}}}}=\expec{}{\br{\Tr~D^2\zeta_{\lambda}\br{\mathbf{A}}\br{\mathbf{B}}}}.\]		
\end{claim}
The proofs of the both claims above are deferred to Appendix~\ref{sec:appinvariance}. Combining Lemma~\ref{lem:hybrid} and Lemma~\ref{lem:zeta}, we conclude the following lemma.

\begin{lemma}\label{lem:invariance}
	Given $M\in\H_2^{\otimes n}$, let $\set{\B_i}_{i=0}^3$ be a standard orthonormal basis in $\M_2$. Then for any $0<\lambda\leq\frac{1}{2}$ and $H\subseteq[n]$, it holds that
	\[\abs{\expec{}{\Tr\zeta\br{\sum_{\sigma\in[4]_{\geq 0}^n}\widehat{M}\br{\sigma}\br{\bigotimes_{i\in H}\B_{\sigma_i}}\otimes\br{\bigotimes_{i\notin H}\mathbf{g}_{i,\sigma_i}\id_2}}}-\Tr\zeta\br{M}}\leq 2^n\br{8\lambda^2+\frac{3^d}{\lambda}\sum_{i\notin H}\influence_i\br{M}^{3/2}}.\]
		where $d=\deg M$.
\end{lemma}

\begin{lemma}\label{lem:invarianceH}
	Given $0<\tau,\delta<1$, $M\in\H_2^{\otimes n}$, $0\leq M\leq \id$ $H\subseteq[n]$, an integer $d>0$ and standard orthonormal basis $\B=\set{\B_i}_{i=0}^3$, suppose $\influence_i\br{M}\leq \tau$ for all $i\notin H$ and $\twonorm{M^{>d}}^2\leq\delta2^n$, where $M^{>d}$ is defined in Definition~\ref{def:lowdegreehighdegree}. Set
	\[\mathbf{M}\defeq\sum_{\sigma\in[4]_{\geq 0}^n:\abs{\sigma}\leq d}\widehat{M}\br{\sigma}\br{\bigotimes_{i\in H}\B_{\sigma_i}}\otimes\br{\bigotimes_{i\notin H}\mathbf{g}_{i,\sigma_i}\id_2}.\]
	Then it holds that
	
\[\expec{}{\Tr~\zeta\br{\mathbf{M}}}\leq O\br{2^n\br{\br{3^d\sqrt{\tau} d}^{2/3}+\sqrt{\delta}}}\]
\end{lemma}

\begin{proof}
Without loss of generality, we may assume that $H=\set{1,2,\ldots, s}$.
Then
\begin{equation}\label{eqn:etamlow}
\expec{}{\Tr~\zeta\br{\mathbf{M}}}\leq\abs{\expec{}{\Tr~\br{\zeta\br{\mathbf{M}}-\zeta\br{M^{\leq d}}}}}+\Tr~\zeta\br{M^{\leq d}}
\end{equation}
Applying Lemma~\ref{lem:invariance},
\begin{eqnarray}
&&\abs{\expec{}{\Tr~\br{\zeta\br{\mathbf{M}}-\zeta\br{M^{\leq d}}}}}\nonumber\\
&\leq&O\br{2^n\br{\lambda^2+\frac{3^d}{\lambda}\sum_{i\notin H}\influence_i\br{M^{\leq d}}^{3/2}}}\nonumber\\
&\leq&O\br{2^n\br{\lambda^2+\frac{3^d\sqrt{\tau}}{\lambda}\influence\br{M^{\leq d}}}}\nonumber\\
&\leq&O\br{2^n\br{\lambda^2+\frac{3^d\sqrt{\tau} d}{\lambda}}},\label{eqn:etamlowdiff}
\end{eqnarray}
where the last inequality is from Lemma~\ref{lem:partialvariance} item 4.

Note that $\zeta\br{M}=0$ since $0\leq M\leq 1$. Again applying Lemma~\ref{lem:zetaadditivity},
\begin{eqnarray}
&&\Tr~\zeta\br{M^{\leq d}}=\abs{\Tr~\zeta\br{M^{\leq d}}-\Tr~\zeta\br{M}}\nonumber\\
&\leq&4\br{\twonorm{M}\twonorm{M^{>d}}+\twonorm{M^{>d}}^2}\leq\sqrt{\delta}2^{n+3}.\label{eqn:msigmalow}
\end{eqnarray}
Combing Eqs.~\eqref{eqn:etamlow}\eqref{eqn:etamlowdiff}\eqref{eqn:msigmalow}, we have
\begin{equation}\label{eqn:etamlow2}
\expec{}{\Tr~\zeta\br{\mathbf{M}^{\leq d}}}\leq O\br{2^n\br{\lambda^2+\frac{3^d\sqrt{\tau} d}{\lambda}+8\sqrt{\delta}}}.
\end{equation}
%For the second term in Eq.~\eqref{eqn:tretamsigmah}, we have
%\begin{eqnarray}\label{eqn:2ndterm}
%&&\expec{}{\twonorm{\mathbf{M}^{\text{low}}}\twonorm{\mathbf{M^{\text{high}}}}}\leq2^n\sqrt{\delta}.
%\end{eqnarray}
%The final term in  Eq.~\eqref{eqn:tretamsigmah} is upper bounded by $\delta2^n$ from the assumption. Thus combining Eqs.~\eqref{eqn:etamlow2}\eqref{eqn:2ndterm}, we have
%\[\expec{}{\Tr~\eta\br{\mathbf{M}}}\leq O\br{2^n\br{\lambda^2+\frac{3^d\sqrt{\tau} d}{\lambda}+10\sqrt{\delta}}}.\]
Choosing $\lambda=\br{3^d\sqrt{\tau} d}^{1/3}$, we conclude the result.
\end{proof}

The following lemma converts random operators to operators.

\begin{lemma}\label{lem:invariancegaussian}
	Given integers $d, h,n>0$ and a degree-$d$ multilinear random operators $\mathbf{M}\in L^2\br{\H_2^{\otimes h},\gamma_n}$ with the associated vector-valued function $p:\reals^n\rightarrow\reals^{4^h}$ under a standard orthonormal basis $\set{\B_i}_{i=0}^3$. Then there exists $M^{(1)}\in\H_2^{\otimes (n+h)}$ satisfying that
	\[\abs{\expec{}{2^n\cdot\Tr~\zeta\br{\mathbf{M}}}-\Tr~\zeta\br{M^{(1)}}}\leq O\br{2^{n+h}\br{3^d\sum_{i=1}^n\influence_i\br{p}^{3/2}}^{2/3}}.\]
	In particular, if for all $i$, $\influence_i\br{p}\leq\tau$ for some $\tau\in(0,1)$, then from Lemma~\ref{lem:influencerandomoperator},
		\[\abs{\expec{}{2^n\cdot\Tr~\zeta\br{\mathbf{M}}}-\Tr~\zeta\br{M^{(1)}}}\leq O\br{2^{n+h}\br{3^dd\sqrt{\tau}}^{2/3}}.\]
\end{lemma}
\begin{proof}
	 Let $\set{\B_i}_{i=0}^3$ be an arbitrary standard orthonormal basis in $\M_2$. Let $\set{\mathbf{g}_i}_{i=1}^n$ be the random variables in $\mathbf{M}$. Substitute each random variable by $\B_1$ and the products of random variables by tensor products. The following proof is same as the one of Lemma~\ref{lem:invariance}.
\end{proof}
We finally reach the main lemma in this section.
\begin{lemma}\label{lem:jointinvariance}
  Given $0<\tau,\delta,\rho<1$, integers $n>h>0, d>0$, $P,Q\in\H_2^{\otimes n}, 0\leq P,Q\leq\id$, $H\subseteq[n]$ of size $\abs{H}=h$, a noisy EPR state $\psi_{AB}$ with the maximal correlation $\rho=\rho\br{\psi_{AB}}$,  suppose $\influence_i\br{P}\leq\tau$ and $\influence_i\br{Q}\leq\tau$ for all $i\notin H$ and $\twonorm{P^{>d}}^2\leq\delta2^n,\twonorm{Q^{>d}}^2\leq\delta2^n$. Then there exist maps $f,g:\H_2^{\otimes n}\times\reals^{n-h}\rightarrow\H_2^{\otimes h}$ satisfying that
  \[\br{\mathbf{P},\mathbf{Q}}\defeq\br{f\br{P,\mathbf{g}},g\br{Q,\mathbf{h}}}_{\br{\mathbf{g},\mathbf{h}}\sim\G_{\rho}^{\otimes n-h}}\in L^2\br{\H_2^{\otimes h},\gamma_{3\br{n-h}}}\times  L^2\br{\H_2^{\otimes h},\gamma_{3\br{n-h}}}\]
are  degree-$d$ multilinear joint random operators with the joint random variables drawn from $\G_{\rho}^{\otimes 3\br{n-h}}$. And
  \begin{enumerate}
    \item $N_2\br{\mathbf{P}}\leq\nnorm{P}_2$ and $N_2\br{\mathbf{Q}}\leq\nnorm{Q}_2$;
    \item $\Tr~\br{P\otimes Q}\psi_{AB}^{\otimes n}=\expec{}{\Tr~\br{\br{\mathbf{P}\otimes\mathbf{Q}}\psi_{AB}^{\otimes h}}}$.
    \item $2^{n-h}\expec{}{\Tr~\mathbf{P}}=\Tr~P$ and $2^{n-h}\expec{}{\Tr~\mathbf{Q}}=\Tr~Q$.
    \item $\expec{}{\Tr~\zeta\br{\mathbf{P}}}\leq O\br{2^h\br{\br{3^d\sqrt{\tau}d}^{2/3}+\sqrt{\delta}}}$ and $\expec{}{\Tr~\zeta\br{\mathbf{Q}}}\leq O\br{2^h\br{\br{3^d\sqrt{\tau}d}^{2/3}+\sqrt{\delta}}}$.
  \end{enumerate}
\end{lemma}

\begin{proof}
	Let $\set{\A_i}_{i=0}^3$ and $\set{\B_i}_{i=0}^3$ be standard orthonormal basis in $\M_2$ satisfying that $\Tr\br{\A_i\otimes\B_j}\psi_{AB}=c_i\delta_{i,j}$ for $0\leq i,j\leq 3$, where $1=c_0>c_1=\rho\geq c_2\geq c_3$. Set joint random variables
	\[\br{\br{\mathbf{g}_{i,0}^{(0)},\mathbf{h}_{i,0}^{(0)}},\br{\mathbf{g}_{i,1}^{(0)},\mathbf{h}_{i,1}^{(0)}},\br{\mathbf{g}_{i,2}^{(0)},\mathbf{h}_{i,2}^{(0)}},\br{\mathbf{g}_{i,3}^{(0)},\mathbf{h}_{i,3}^{(0)}}}_{i=1}^h\sim\br{\set{\br{1,1}}\times\G_{c_1}\times \G_{c_2}\times\G_{c_3}}^{\otimes h}.\]
	Define
	\[\mathbf{P}^{(0)}\defeq\sum_{\sigma\in[4]_{\geq 0}^n}\widehat{P}\br{\sigma}\prod_{i\notin H}\mathbf{g}^{(0)}_{i,\sigma_i}\A_{\sigma_H},\]
	and
	\[\mathbf{Q}^{(0)}\defeq\sum_{\sigma\in[4]_{\geq 0}^n}\widehat{Q}\br{\sigma}\prod_{i\notin H}\mathbf{h}^{(0)}_{i,\sigma_i}\B_{\sigma_H}.\]
	Then
    \[N_2\br{\mathbf{P}^{(0)}}^2=\sum_{\sigma}\abs{\widehat{P}\br{\sigma}}^2=\nnorm{P}_2^2,N_2\br{\mathbf{Q}^{(0)}}^2=\sum_{\sigma}\abs{\widehat{Q}\br{\sigma}}^2=\nnorm{Q}_2^2,\]

    and

	\[\Tr~\br{P\otimes Q}\psi_{AB}^{\otimes n}=\expec{}{\Tr~\br{\br{\mathbf{P}^{(0)}\otimes\mathbf{Q}^{(0)}}\psi_{AB}^{\otimes h}}}=\sum_{\sigma\in[4]_{\geq 0}^n}c_{\sigma},\]
	and
	\[2^{n-h}\expec{}{\Tr~\mathbf{P}^{(0)}}=\Tr~P, 2^{n-h}\expec{}{\Tr~\mathbf{Q}^{(0)}}=\Tr~Q.\]
	From Lemma~\ref{lem:invarianceH},
	\[\expec{}{\Tr~\zeta\br{\mathbf{P}^{(0)}}}\leq O\br{2^h\br{\br{3^d\sqrt{\tau}d}^{2/3}+\sqrt{\delta}}}~\mbox{and}~\expec{}{\Tr~\zeta\br{\mathbf{Q}^{(0)}}}\leq O\br{2^h\br{\br{3^d\sqrt{\tau}d}^{2/3}+\sqrt{\delta}}}.\]
However, the correlation of $\br{\mathbf{g}_{i,j}^{(0)},\mathbf{h}_{i,j}^{(0)}}$ is not exactly the one we need. Given $\br{\mathbf{g}_i,\mathbf{h}_i}_{i=1}^{3\br{n-h}}\sim \G_{\rho}^{\otimes 3(n-h)}$, we perform the following substitution in $\mathbf{P}^{(0)}$ and $\mathbf{Q}^{(0)}$
\[\mathbf{g}_{i,b}^{(0)}\leftarrow\begin{cases}
                                    1, & \mbox{if $b=0$} \\
                                    \mathbf{g}_{3(i-1)+b}, & \mbox{otherwise};
                                  \end{cases}~\mbox{and}
~\mathbf{h}_{i,b}^{(0)}\leftarrow\begin{cases}
                                    1, & \mbox{if $b=0$} \\
                                    \frac{c_i}{\rho}\mathbf{h}_{3(i-1)+b}, & \mbox{otherwise}
                                  \end{cases}\]
to get $\mathbf{P}$ and $\mathbf{Q}$, respectively. The items 1,2,3 still hold obviously. To argue item 4, note that the absolute values of all Fourier coefficients do not increase. By Lemma~\ref{lem:invariance}, item 4 follows.
\end{proof}

Analogously, the following lemma converts joint random operators back to operators.

\begin{lemma}\label{lem:invariancejointgaussian}
	Given $0\leq\rho<1$, $\delta,\tau\in(0,1)$, integers $n,h,d>0$, a noisy EPR state $\psi_{AB}$ with the maximal correlation $\rho=\rho\br{\psi_{AB}}$, standard orthonormal basis $\set{\A_i}_{i=0}^3$ and $\set{\B_i}_{i=0}^3$ in $\M_2$, there exist maps $f,g:L^2\br{\H_2^{\otimes h},\gamma_n}\rightarrow\H_2^{\otimes n+h}$ such that for any degree-$d$ multilinear joint random operators
	\[\br{\mathbf{P},\mathbf{Q}}=\br{\sum_{\sigma\in[4]_{\geq 0}^h}p_{\sigma}\br{\mathbf{g}}\A_{\sigma},\sum_{\sigma\in[4]_{\geq 0}^h}q_{\sigma}\br{\mathbf{h}}\B_{\sigma}}_{\br{\mathbf{g},\mathbf{h}}\sim\G_{\rho}^{\otimes n}}\in L^2\br{\H_2^{\otimes h},\gamma_n}\times L^2\br{\H_2^{\otimes h},\gamma_n},\]
	satisfying that
		\[(\forall i\in[n]):~\sum_{\sigma\in[4]_{\geq 0}^h}\influence_i\br{p_{\sigma}}\leq\tau~\mbox{and}~\sum_{\sigma\in[4]_{\geq 0}^h}\influence_i\br{q_{\sigma}}\leq\tau.\]
		Let $\br{P,Q}=\br{f\br{\mathbf{P}},g\br{\mathbf{Q}}}$. The following holds.
		\begin{enumerate}

            \item $\Tr~P=2^n\expec{}{\Tr~\mathbf{P}} ~\mbox{and}~\Tr~Q=2^n\expec{}{\Tr~\mathbf{Q}};$

			\item $\Tr~\br{P\otimes Q}\psi_{AB}^{\otimes (n+h)}=\expec{}{\Tr~\br{\mathbf{P}\otimes\mathbf{Q}}\psi_{AB}^{\otimes h}};$

\item $N_2\br{\mathbf{P}}=\nnorm{P}_2~\mbox{and}~N_2\br{\mathbf{Q}}=\nnorm{Q}_2;$
			
			\item $\abs{\expec{}{2^n\cdot\Tr~\zeta\br{\mathbf{P}}}-\Tr~\zeta\br{P}}\leq O\br{2^{n+h}\br{3^dd\sqrt{\tau}}^{2/3}}$

			and

			$\abs{\expec{}{2^n\cdot\Tr~\zeta\br{\mathbf{Q}}}-\Tr~\zeta\br{Q}}\leq O\br{2^{n+h}\br{3^dd\sqrt{\tau}}^{2/3}}.$
		\end{enumerate}
\end{lemma}

\begin{proof}
	From Proposition~\ref{prop:maximalcorrelationone}, let $\A$ and $\B$ be the Hermitian matrices achieved the maximal correlation of $\psi_{AB}$ in Definition~\ref{def:maximalcorrelation}. Substitute each pair $\br{\mathbf{g}_i,\mathbf{h}_i}$ by $\br{\A,\B}$ and the products of random variables by tensor products. The following proof is same as the one of Lemma~\ref{lem:jointinvariance}.
\end{proof}
\section{Dimension reduction for random operators}\label{sec:dimensionreduction}
The following is the main lemma in this section.

\begin{lemma}\label{lem:dimensionreduction}
Given parameters $\rho\in[0,1], \delta,\alpha>0$, integers $d,n,h>0$, a noisy EPR state $\psi_{AB}$ with the maximal correlation $\rho=\rho\br{\psi_{AB}}$, and degree-$d$ multilinear joint random operators
\[\br{\mathbf{P},\mathbf{Q}}=\br{\sum_{\sigma\in[4]_{\geq 0}^h}p_{\sigma}\br{\mathbf{g}}\A_{\sigma},\sum_{\sigma\in[4]_{\geq 0}^h}q_{\sigma}\br{\mathbf{h}}\B_{\sigma}}_{\br{\mathbf{g},\mathbf{h}}\sim\G_{\rho}^{\otimes n}}\in L^2\br{\H_2^{\otimes h},\gamma_n}\times L^2\br{\H_2^{\otimes h},\gamma_n},\]
where $\set{\A_i}_{i=0}^3,\set{\B_i}_{i=0}^3$ are both standard orthonormal basis in $\M_2$, then there exists an explicitly computable $n_0=n_0\br{d,h,\delta}$, and joint random operators
\[\br{\mathbf{P}_M^{\br{1}},\mathbf{Q}_M^{\br{1}}}=\br{\sum_{\sigma\in[4]_{\geq 0}^h}p^{\br{1}}_{\sigma,M}\br{\mathbf{x}}\A_{\sigma},\sum_{\sigma\in[4]_{\geq 0}^h}q^{\br{1}}_{\sigma,M}\br{\mathbf{y}}\B_{\sigma}}_{\br{\mathbf{x},\mathbf{y}}\sim\G_{\rho}^{\otimes n_0}},\]
for any $M\in\reals^{n\times D}$ such that the following holds.

If we sample $\mathbf{M}\sim \gamma_{n\times D}$, then with probability at least  $1-\delta-2\alpha$, it holds that
\begin{enumerate}
  \item $N_2\br{\mathbf{P}_{\mathbf{M}}^{(1)}}_2\leq(1+\delta)N_2\br{\mathbf{P}}$ and $N_2\br{\mathbf{Q}_{\mathbf{M}}^{(1)}}\leq(1+\delta)N_2\br{\mathbf{Q}}$;
      \item $\abs{\expec{\mathbf{P}}{\Tr~\mathbf{P}^{(1)}_{\mathbf{M}}}-\expec{\mathbf{P}}{\Tr~\mathbf{P}}}\leq\delta2^hN_2\br{\mathbf{P}}$ and $\abs{\expec{\mathbf{Q}}{\Tr~\mathbf{Q}_{\mathbf{M}}^{(1)}}-\expec{\mathbf{Q}}{\Tr~\mathbf{Q}}}\leq\delta2^hN_2\br{\mathbf{Q}}$;
  \item $\expec{\mathbf{P}}{\Tr~\zeta\br{\mathbf{P}_{\mathbf{M}}^{\br{1}}}}\leq \frac{1}{\sqrt{\alpha}}\expec{\mathbf{P}}{\Tr~\zeta\br{\mathbf{P}}}$ and $\expec{\mathbf{Q}}{\Tr~\zeta\br{\mathbf{Q}_{\mathbf{M}}^{\br{1}}}}\leq \frac{1}{\sqrt{\alpha}}\expec{\mathbf{Q}}{\Tr~\zeta\br{\mathbf{Q}}}$;
  \item $\abs{\expec{\mathbf{P},\mathbf{Q}}{\Tr\br{\mathbf{P}_{\mathbf{M}}^{\br{1}}}\otimes \mathbf{Q}_\mathbf{M}^{\br{1}}}\psi_{AB}^{\otimes h}-\expec{\mathbf{P},\mathbf{Q}}{\Tr\br{\mathbf{P}\otimes \mathbf{Q}}\psi_{AB}^{\otimes h}}}\leq\delta N_2\br{\mathbf{P}}N_2\br{\mathbf{Q}}$.
\end{enumerate}
	
	In particular, one may take $n_0=\frac{4^{3h+4}d^{O\br{d}}}{\delta^2}$.
\end{lemma}

\begin{fact}~\cite{Ghazi:2018:DRP:3235586.3235614}\label{fac:dimensionreduction}
	Given parameters $n,d\in\mathbb{Z}_{>0}$, $\rho\in[0,1]$ and $\delta>0$, there exists an explicitly computable $D=D\br{d,\delta}$ such that the following holds.
	
	For any $n$ and any degree-$d$ multilinear polynomials $f,g:\reals^n\rightarrow\reals$, and $M\in\reals^{n\times D}$, define functions $f_M, g_M:\reals^{D}\rightarrow\reals$ as
\begin{equation}\label{eqn:fmgm}
  f_M\br{x}\defeq f\br{\frac{Mx}{\twonorm{x}}}~\mbox{and}~g_M\br{x}\defeq g\br{\frac{Mx}{\twonorm{x}}}.
\end{equation}
Then
\[\Pr_{\mathbf{M}\sim \gamma_{n\times D}}\Br{\abs{\innerproduct{f_{\mathbf{M}}}{g_{\mathbf{M}}}_{\G_{\rho}^{\otimes D}}-\innerproduct{f}{g}_{\G_{\rho}^{\otimes n}}}<\delta\twonorm{f}\twonorm{g}}\geq 1-\delta.\]
In particular, one may take $D=\frac{d^{O\br{d}}}{\delta^2}$.
\end{fact}

Choosing $g\equiv 1$, we get
\begin{equation}\label{eqn:expec}
  \Pr_{\mathbf{M}\sim \gamma_{n\times D}}\Br{\abs{\widehat{f_{\mathbf{M}}}\br{\mathbf{0}}-\widehat{f}\br{\mathbf{0}}}<\delta\twonorm{f}}\geq 1-\delta.
\end{equation}

If $f$ and $g$ are identical and $\rho=1$, we have

\begin{eqnarray}
  &&\Pr_{\mathbf{M}\sim \gamma_{n\times D}}\Br{\abs{\twonorm{f_{\mathbf{M}}}^2-\twonorm{f}^2}\leq\delta\twonorm{f}^2}\geq1-\delta; \label{eqn:normf}\\
  &&\Pr_{\mathbf{M}\sim \gamma_{n\times D}}\Br{\abs{\twonorm{g_{\mathbf{M}}}^2-\twonorm{g}^2}\leq\delta\twonorm{g}^2}\geq1-\delta. \label{eqn:normg}
\end{eqnarray}

\begin{fact}~\cite{Ghazi:2018:DRP:3235586.3235614}\label{fac:reductionrounding}
  Given integers $n, k, D>0$, let $f\in L^2\br{\reals^k,\gamma_n}$, and $\Delta$ be a convex body in $\reals^k$ with rounding map $\R$ defined in Section~\ref{subsec:misc}. Let $f_M:\reals^D\rightarrow\reals^k$ be defined analogously to Eq.~\eqref{eqn:fmgm}. It holds that,
  \[\Pr_{\mathbf{M}\sim \gamma_{n\times D}}\Br{\twonorm{\R\circ f_{\mathbf{M}}-f_{\mathbf{M}}}\leq\frac{1}{\delta}\twonorm{\R\circ f-f}}\geq1-\delta^2,\]
  for $0<\delta<1$.
\end{fact}

We are now ready to prove the main lemma.
	
\begin{proof}[Proof of Lemma~\ref{lem:dimensionreduction}]
      From Lemma~\ref{lem:normofM} and Lemma~\ref{lem:mg}, we may assume that the basis $\set{\A_i}_{i=0}^3$ and $\set{\B_i}_{i=0}^3$ satisfy  $\Tr\br{\A_i\otimes\B_j}\psi_{AB}=c_i\delta_{i,j}$ without loss of generality.
	Then from Lemma~\ref{lem:convariancetensor}
	\begin{eqnarray*}
	 \expec{\br{\mathbf{g},\mathbf{h}}\sim\G_{\rho}^{\otimes n-h}}{\Tr\br{P\br{\mathbf{g}}\otimes Q\br{\mathbf{h}}}\br{\psi_{AB}}^{\otimes n}}=\sum_{\sigma\in[4]_{\geq 0}^{\otimes h}}c_{\sigma}\innerproduct{p_{\sigma}}{g_{\sigma}}_{\G_{\rho}^{\otimes n}},
	\end{eqnarray*}
where $c_{\sigma}\defeq c_{\sigma_1}\cdots c_{\sigma_h}$.

Applying Fact~\ref{fac:dimensionreduction} with $\delta\leftarrow\frac{\delta^2}{ 4^{h+2}}, n_0\leftarrow\frac{4^{3h+4}d^{O\br{d}}}{\delta^2}$ and the union bound, it holds

\begin{equation}\label{eqn:mpq}
  \Pr_{\mathbf{M}\sim \gamma_{n\times n_0}}\Br{\br{\forall \sigma\in[4]_{\geq 0}^h}~\abs{\innerproduct{p_{\sigma,\mathbf{M}}}{q_{\sigma,\mathbf{M}}}_{\G_{\rho}^{\otimes n_0}}-\innerproduct{p_{\sigma}}{q_{\sigma}}_{\G_{\rho}^{\otimes n}}}\leq \delta\twonorm{p_{\sigma}}\twonorm{q_{\sigma}}}\geq1-\delta/16,
\end{equation}
and
\begin{equation}\label{eqn:mpq1}
  \Pr_{\mathbf{M}\sim \gamma_{n\times n_0}}\Br{\br{\forall \sigma\in[4]_{\geq 0}^h}:~\abs{\widehat{p_{\sigma,\mathbf{M}}}\br{\mathbf{0}}-\widehat{p_{\sigma}}\br{\mathbf{0}}}\leq \delta\twonorm{p_{\sigma}}}\geq1-\delta/16,
\end{equation}
and
\begin{equation}\label{eqn:mpq2}
  \Pr_{\mathbf{M}\sim \gamma_{n\times n_0}}\Br{\br{\forall \sigma\in[4]_{\geq 0}^h}:~\abs{\widehat{q_{\sigma,\mathbf{M}}}\br{\mathbf{0}}-\widehat{q_{\sigma}}\br{\mathbf{0}}}\leq \delta\twonorm{q_{\sigma}}}\geq1-\delta/16,
\end{equation}
and
\begin{eqnarray}
  &&\Pr_{\mathbf{M}\sim \gamma_{n\times n_0}}\Br{\br{\forall\sigma\in[4]_{\geq 0}^h}:~\abs{\twonorm{p_{\sigma,\mathbf{M}}}^2-\twonorm{p_{\sigma}}^2}\leq\delta\twonorm{p_{\sigma}}^2}\geq1-\delta,\label{eqn:mpq5}
\end{eqnarray}
and
\begin{eqnarray}
  &&\Pr_{\mathbf{M}\sim \gamma_{n\times n_0}}\Br{\br{\forall\sigma\in[4]_{\geq 0}^h}:~\abs{\twonorm{q_{\sigma,\mathbf{M}}}^2-\twonorm{q_{\sigma}}^2}\leq\delta\twonorm{q_{\sigma}}^2}\geq1-\delta,\label{eqn:mpq6}.
\end{eqnarray}
Define
\[\br{\mathbf{P}^{\br{1}}_M,\mathbf{P}^{\br{1}}_M}\defeq\br{\sum_{\sigma\in[4]_{\geq 0}^h}p_{\sigma,M}\br{\mathbf{g}}\A_{\sigma},\sum_{\sigma\in[4]_{\geq 0}^h}q_{\sigma,M}\br{\mathbf{h}}\B_{\sigma}}_{\br{\mathbf{g},\mathbf{h}}\sim\G_{\rho}^{\otimes n_0}}.\]
For any $M$ satisfying Eq.~\eqref{eqn:mpq}, we have
\begin{eqnarray*}
  &&\abs{\expec{}{\Tr\br{\mathbf{P}_M^{\br{1}}\otimes \mathbf{Q}_M^{\br{1}}}\psi_{AB}^{\otimes h}}-\expec{}{\Tr\br{\mathbf{P}\otimes \mathbf{Q}}\psi_{AB}^{\otimes h}}} \\
  &=&\abs{\sum_{\sigma\in[4]_{\geq 0}^h}c_{\sigma}\br{\innerproduct{p^{\br{1}}_{\sigma,M}}{q^{\br{1}}_{\sigma,M}}_{\G_{\rho}^{\otimes n_0}}-\innerproduct{p_{\sigma}}{q_{\sigma}}_{\G_{\rho}^{\otimes n}}}}\quad\quad\mbox{(Lemma~\ref{lem:convariancetensor})}\\
  &\leq&\delta\sum_{\sigma\in[4]_{\geq 0}^h}\twonorm{p_{\sigma}}\twonorm{q_{\sigma}}\quad\quad\mbox{(Eq.~\eqref{eqn:mpq} and $c_{\sigma}\leq 1$ due to Lemma~\ref{lem:normofM})}\\
  &\leq&\delta\br{\sum_{\sigma\in[4]_{\geq 0}^h}\twonorm{p_{\sigma}}^2}^{1/2}\br{\sum_{\sigma\in[4]_{\geq 0}^h}\twonorm{q_{\sigma}}^2}^{1/2}\\
  &=&\delta N_2\br{\mathbf{P}}N_2\br{\mathbf{Q}}\quad\quad\mbox{(Lemma~\ref{lem:randoperator})},
\end{eqnarray*}
Thus
\begin{equation}\label{eqn:m1}
  \Pr_{\mathbf{M}\sim \gamma_{n\times n_0}}\Br{\abs{\expec{}{\Tr\br{\mathbf{P}_M^{\br{1}}\otimes \mathbf{Q}_M^{\br{1}}}\psi_{AB}^{\otimes h}}-\expec{}{\Tr\br{\mathbf{P}\otimes \mathbf{Q}}\psi_{AB}^{\otimes h}}}\leq\delta N_2\br{\mathbf{P}}N_2\br{\mathbf{Q}}}\geq1-\delta/16.
\end{equation}
For any $M$ satisfying Eq.~\eqref{eqn:mpq1},
\begin{eqnarray*}
  &&\abs{\expec{}{\Tr~\mathbf{P}_M^{(1)}}-\expec{}{\Tr~\mathbf{P}}} \\
   &=&2^h\abs{\widehat{p_{\mathbf{0},M}}\br{\mathbf{0}}-\widehat{p_{\mathbf{0}}}\br{\mathbf{0}}}\\ &\leq&\delta2^h\twonorm{p_{\mathbf{0}}}\leq\delta2^hN_2\br{\mathbf{P}}.
\end{eqnarray*}
Hence
\begin{equation}\label{eqn:expecp}
  \Pr_{\mathbf{M}\sim \gamma_{n\times n_0}}\Br{\abs{\expec{}{\Tr~\mathbf{P}_{\mathbf{M}}^{(1)}}-\expec{}{\Tr~\mathbf{P}}}\leq\delta2^hN_2\br{\mathbf{P}}}\geq1-\delta/16.
\end{equation}
Symmetrically, we have
\begin{equation}\label{eqn:expecq}
  \Pr_{\mathbf{M}\sim \gamma_{n\times n_0}}\Br{\abs{\expec{}{\Tr~\mathbf{Q}_{\mathbf{M}}^{(1)}}-\expec{}{\Tr~\mathbf{Q}}}\leq\delta2^hN_2\br{\mathbf{Q}}}\geq1-\delta/16.
\end{equation}
For any $M$ satisfying Eq.~\eqref{eqn:mpq5},
\[N_2\br{\mathbf{P}_M^{(1)}}^2=\sum_{\sigma\in[4]_{\geq 0}^h}\twonorm{p_{\sigma,M}}^2\leq\br{1+\delta}\sum_{\sigma\in[4]_{\geq 0}^h}\twonorm{p_{\sigma}}^2=\br{1+\delta}N_2\br{\mathbf{P}}^2,\]
where the both equalities are from Lemma~\ref{lem:randoperator}. Hence
\begin{equation}\label{eqn:expecnormp}
  \Pr_{\mathbf{M}\sim \gamma_{n\times n_0}}\Br{N_2\br{\mathbf{P}^{(1)}_{\mathbf{M}}}\leq\br{1+\delta}N_2\br{\mathbf{P}}}\geq1-\delta/16.
\end{equation}
Symmetrically, we have
\begin{equation}\label{eqn:expecnormq}
  \Pr_{\mathbf{M}\sim \gamma_{n\times n_0}}\Br{N_2\br{\mathbf{Q}^{(1)}_{\mathbf{M}}}\leq\br{1+\delta}N_2\br{\mathbf{Q}}}\geq1-\delta/16.
\end{equation}
Set \[\Delta\defeq\set{x\in\reals^{4^h}:0\leq\sum_{\sigma\in[4]_{\geq 0}^h}x_{\sigma}\A_{\sigma}\leq\id}.\]
It is easy to verify that $\Delta$ is a convex body. Let $\R$ be the rounding map of $\Delta$. Note that for any random operator $\mathbf{P}\in L^2\br{\H_2^{\otimes h},\gamma_n}$ with the associated vector-valued function $p$, $\twonorm{\R\circ p-p}^2=\frac{1}{2^h}\Tr~\zeta\br{\mathbf{P}}$ due to Lemma~\ref{lem:closedelta}.
%Set $p=\br{p_{\sigma}}_{\sigma\in[4]_{\geq 0}^h}:\reals^n\rightarrow\reals^{4^h}$ to be a vector-valued function. For any Hermitian matrix $M\in\H_d$, the $\twonorm{\cdot}$ -distance between $M$ and the set $\set{X\in\H_d:0\leq X\leq \id}$ is $\sqrt{\Tr~\zeta\br{M}}$ due to Lemma~\ref{lem:closedelta}. It implies that $\twonorm{\R\circ p-p}^2=\frac{1}{2^h}\expec{}{\Tr~\zeta\br{\mathbf{P}}}$ and $\twonorm{\R\circ p^{\br{1}}-p^{\br{1}}}^2=\frac{1}{2^h}\expec{}{\Tr~\zeta\br{\mathbf{P}^{\br{1}}}}$.
Hence Fact~\ref{fac:reductionrounding} implies that
\begin{equation}\label{eqn:rpq}
  \Pr_{\mathbf{M}\sim \gamma_{n\times n_0}}\Br{\Tr~\zeta\br{\mathbf{P}}\leq \frac{1}{\sqrt{\alpha}}\Tr~\zeta\br{\mathbf{P}^{\br{1}}}}\geq1-\alpha.
\end{equation}
Applying the same argument to $\mathbf{Q}$ and $\mathbf{Q}^{\br{1}}$, we have
\begin{equation}\label{eqn:rpq2}
  \Pr_{\mathbf{M}\sim \gamma_{n\times n_0}}\Br{\Tr~\zeta\br{\mathbf{Q}}\leq \frac{1}{\sqrt{\alpha}}\Tr~\zeta\br{\mathbf{Q}^{\br{1}}}}\geq1-\alpha.
\end{equation}
Again applying the union bound to Eqs.~\eqref{eqn:m1}\eqref{eqn:expecp}\eqref{eqn:expecq}\eqref{eqn:expecnormp}\eqref{eqn:expecnormq}\eqref{eqn:rpq}\eqref{eqn:rpq2}, with probability at least $1-\delta-2\alpha$ over $\mathbf{M}\sim \gamma_{n\times D}$ all the events in Eqs.~\eqref{eqn:m1}\eqref{eqn:expecp}\eqref{eqn:expecq}\eqref{eqn:expecnormp}\eqref{eqn:expecnormq}\eqref{eqn:rpq}\eqref{eqn:rpq2} occur. Setting $p^{\br{1}}_{\sigma,\mathbf{M}}=p_{\sigma,\mathbf{M}}$ and $q^{\br{1}}_{\sigma,\mathbf{M}}=q_{\sigma,\mathbf{M}}$, we conclude the lemma.

\end{proof}

\section{Smoothing random operators}\label{sec:smoothrandom}

The main result in this section is the following, which is a generalization of Lemma~\ref{lem:smoothing of strategies} to random operators.

		\begin{lemma}\label{lem:smoothgaussian}
	Given integers $n,h>0$, a noisy EPR state $\psi_{AB}$ with the maximal correlation $\rho\defeq \rho\br{\psi_{AB}}<1$, there exist an explicit $d=d\br{\rho, \delta}$ and a map $f:L^2\br{\H_2^{\otimes h},\gamma_n}\rightarrow L^2\br{\H_2^{\otimes h},\gamma_n}$ such that, for any joint random operators $\br{\mathbf{P},\mathbf{Q}}\in L^2\br{\H_2^{\otimes h},\gamma_n}\times L^2\br{\H_2^{\otimes h},\gamma_n}$  with random variables drawn from $\G_{\rho}^{\otimes n}$, $\br{\mathbf{P}^{(1)},\mathbf{Q}^{(1)}}\defeq \br{f\br{\mathbf{P}},f\br{\mathbf{Q}}}\in L^2\br{\H_2^{\otimes h},\gamma_n}\times L^2\br{\H_2^{\otimes h},\gamma_n}$ satisfy the following.
	\begin{enumerate}
		\item $\deg\br{\mathbf{P}^{(1)}}\leq d$ and $\deg\br{\mathbf{Q}^{(1)}}\leq d$.
		
		\item $\expec{}{\Tr~\mathbf{P}^{(1)}}=\expec{}{\Tr~\mathbf{P}}$ and $\expec{}{\Tr~\mathbf{Q}^{(1)}}=\expec{}{\Tr~\mathbf{Q}}$.
		\item $N_2\br{\mathbf{P}^{(1)}}\leq N_2\br{\mathbf{P}}$ and $N_2\br{\mathbf{Q}^{(1)}}\leq N_2\br{\mathbf{Q}}$.
		\item $\expec{}{\Tr~\zeta\br{\mathbf{P}^{(1)}}}\leq2\br{\expec{}{\Tr~\zeta\br{\mathbf{P}}}+\delta 2^hN_2\br{\mathbf{P}}^2}$ and~\\ $\expec{}{\Tr~\zeta\br{\mathbf{Q}^{(1)}}}\leq2\br{\expec{}{\Tr~\zeta\br{\mathbf{Q}}}+\delta 2^hN_2\br{\mathbf{Q}}^2}$
		\item $\abs{\expec{}{\Tr\br{\mathbf{P}\otimes\mathbf{Q}}\psi_{AB}^{\otimes h}}-\expec{}{\Tr\br{\mathbf{P}^{(1)}\otimes\mathbf{Q}^{(1)}}\psi_{AB}^{\otimes h}}}\leq \delta N_2\br{\mathbf{P}}N_2\br{\mathbf{Q}}$.
	\end{enumerate}
	In particular, one may take $d=O\br{\frac{\log^2\frac{1}{\delta}}{\delta\br{1-\rho}}}$.
\end{lemma}

	\begin{fact}~\cite{Ghazi:2018:DRP:3235586.3235614}\label{fac:smoothgaussian}
		Let $\rho\in[0,1),\delta>0,k,n\in\mathbb{Z}_{>0}$ be any given constant parameters, $f,g\in L^2\br{\reals^k,\gamma_n}$; $\Delta_1,\Delta_2\subseteq\reals^k$ be convex bodies. Set $\R_1$ and $\R_2$ be the rounding maps with respect to $\Delta_1$ and $\Delta_2$, respectively. There exist an explicit $d=d\br{\rho,\delta}$ and functions $f^{(1)}, g^{(1)}\in L^2\br{\reals^k,\gamma_n}$,where $f^{(1)}$ only depends on $f$ and $g^{(1)}$ only depends on $g$, such that the following hold.
		\begin{enumerate}
		\item $f^{(1)}$ and $g^{(1)}$ are degree at most $d$.
\item $\expec{}{f^{(1)}}=\expec{}{f}$ and $\expec{}{g^{(1)}}=\expec{}{g}$.
		\item For any $i\in[k]$, it holds that $\twonorm{f_i^{(1)}}\leq\twonorm{f_i}$ and $\twonorm{g_i^{(1)}}\leq\twonorm{g_i}$.
		\item $\twonorm{\R\circ f^{(1)}-f^{(1)}}\leq\twonorm{\R\circ f-f}+\delta\twonorm{f}$ and $\twonorm{\R\circ g^{(1)}-g^{(1)}}\leq\twonorm{\R\circ g-g}+\delta\twonorm{g}$.
		\item For every $i\in [k]$,
\[\abs{\innerproduct{f_i\br{\mathbf{x}}}{g_i\br{\mathbf{y}}}_{\G_{\rho}^{\otimes n}}-\innerproduct{f_i^{(1)}\br{\mathbf{x}}}{g_i^{(1)}\br{\mathbf{y}}}_{\G_{\rho}^{\otimes n}}}\leq\delta\twonorm{f_i}\twonorm{g_i}.\]
\end{enumerate}
		In particular, one may take $d=O\br{\frac{\log^2\frac{1}{\delta}}{\delta\br{1-\rho}}}$.
	\end{fact}

\begin{proof}[Proof of Lemma~\ref{lem:smoothgaussian}]
 From Lemma~\ref{lem:normofM}, there exist standard orthonormal basis $\set{\A_i}_{i=0}^3$ and $\set{\B_i}_{i=0}^3$  satisfying that $\Tr\br{\A_i\otimes\B_j}\psi_{AB}=c_i\delta_{i,j}$, where $1=c_0>c_1=\rho\geq c_2\geq c_3\geq 0$. Let $p,q:\reals^n\rightarrow\reals^{4^h}$ be the associated vector-valued functions of $\mathbf{P}$ and $\mathbf{Q}$ under the basis $\set{\A_i}_{i=0}^3$ and $\set{\B_i}_{i=0}^3$, respectively. Then
	\[\mathbf{P}=\sum_{\sigma\in[4]_{\geq 0}^h}p_{\sigma}\br{\mathbf{g}}\A_{\sigma}~\mbox{and}~\mathbf{Q}=\sum_{\sigma\in[4]_{\geq 0}^h}q_{\sigma}\br{\mathbf{h}}\B_{\sigma},\]
where $\br{\mathbf{g},\mathbf{h}}\sim\G_{\rho}^{\otimes n}$.
	From Lemma~\ref{lem:convariancetensor},
	\begin{eqnarray*}
	 \expec{}{\Tr~\br{\mathbf{P}\otimes \mathbf{Q}}\br{\psi_{AB}}^{\otimes h}}=\sum_{\sigma\in[4]_{\geq 0}^{\otimes h}}c_{\sigma}\innerproduct{p_{\sigma}}{q_{\sigma}}_{\G_{\rho}^{\otimes n}}.
\end{eqnarray*}
Applying Fact~\ref{fac:smoothgaussian} to $\br{p,q}$, we obtain
$\br{p^{(1)},q^{(1)}}$.
Define
\[\mathbf{P}^{(1)}\defeq\sum_{\sigma\in[4]_{\geq 0}^h}p^{(1)}_{\sigma}\br{\mathbf{g}}\A_{\sigma} ~\mbox{and}~\mathbf{Q}^{(1)}\defeq\sum_{\sigma\in[4]_{\geq 0}^h}q^{(1)}_{\sigma}\br{\mathbf{h}}\B_{\sigma}.\]
Item 1 follows directly.

Note that $\expec{}{\Tr~\mathbf{P}^{(1)}}=2^h\expec{}{p^{(1)}_{\mathbf{0}}}$ and $\expec{}{\Tr~\mathbf{P}}=2^h\expec{}{p_{\mathbf{0}}}$. Thus the first equality in item 2 follows from Fact~\ref{fac:smoothgaussian} item 2. The second equality follows similarly.

To prove item 3, we have that $$N_2\br{\mathbf{P}^{(1)}}=\twonorm{p^{(1)}}\leq\twonorm{p}=N_2\br{\mathbf{P}},$$
where the both equalities are from Lemma~\ref{lem:randoperator}; the inequality is from Fact~\ref{fac:smoothgaussian} item 3. The second part of item 2 in Lemma~\ref{lem:smoothgaussian} follows by the same argument.
To prove item 3, we define \[\Delta_1\defeq\set{x\in\reals^{4^h}:0\leq\sum_{\sigma\in[4]_{\geq 0}^h}x_{\sigma}\A_{\sigma}\leq\id},\]
and
\[\Delta_2\defeq\set{x\in\reals^{4^h}:0\leq\sum_{\sigma\in[4]_{\geq 0}^h}x_{\sigma}\B_{\sigma}\leq\id}.\]
It is easy to verify that both $\Delta_1$ and $\Delta_2$ are convex bodies. From Lemma~\ref{lem:closedelta},
\begin{eqnarray*}
  \twonorm{\R\circ p-p}^2&=&\frac{1}{2^h}\expec{}{\Tr~\zeta\br{\mathbf{P}}},\\
  \twonorm{\R\circ q-q}^2&=&\frac{1}{2^h}\expec{}{\Tr~\zeta\br{\mathbf{Q}}}, \\
  \twonorm{\R\circ p^{(1)}-p^{(1)}}^2&=&\frac{1}{2^h}\expec{}{\Tr~\zeta\br{\mathbf{P}^{(1)}}},\\
  \twonorm{\R\circ q^{(1)}-q^{(1)}}^2&=&\frac{1}{2^h}\expec{}{\Tr~\zeta\br{\mathbf{Q}^{(1)}}}.
\end{eqnarray*}
The  Fact~\ref{fac:smoothgaussian} item 4 implies that
\[\br{\frac{1}{2^h}\expec{}{\Tr~\zeta\br{\mathbf{P}^{(1)}}}}^{1/2}\leq\br{\frac{1}{2^h}\expec{}{\Tr~\zeta\br{\mathbf{P}}}}^{1/2}+\delta\twonorm{p}.\]
Note that $\twonorm{p}^2=N_2\br{\mathbf{P}}^2$ by Lemma~\ref{lem:randoperator}. Taking square on both sides of the inequality above, we conclude the first inequality in Lemma~\ref{lem:smoothgaussian} item 3. The second inequality follows exactly same.

To prove item 4, consider
\begin{eqnarray*}
&&\abs{\expec{}{\Tr\br{\mathbf{P}\otimes \mathbf{Q}}\psi_{AB}^{\otimes h}}-\expec{}{\Tr\br{\mathbf{P}^{\br{1}}\otimes \mathbf{Q}^{\br{1}}}\psi_{AB}^{\otimes h}}}\\
&=&\abs{\sum_{\sigma\in[4]_{\geq 0}^ h}c_{\sigma}\br{\innerproduct{p_{\sigma}}{q_{\sigma}}_{\G_{\rho}^{\otimes n}}-\innerproduct{p^{(1)}_{\sigma}}{q^{(1)}_{\sigma}}_{\G_{\rho}^{\otimes n}}}}\\
&\leq&\delta\sum_{\sigma\in[4]_{\geq 0}^{\otimes h}}\twonorm{p_{\sigma}}\twonorm{q_{\sigma}}\quad\quad\mbox{(Fact~\ref{fac:smoothgaussian} item 5)and $c_{\sigma}\leq 1$ due to Lemma~\ref{lem:normofM})}\\
&\leq&\delta\br{\sum_{\sigma\in[4]_{\geq 0}^h}\twonorm{p_{\sigma}}^2}^{1/2}\br{\sum_{\sigma\in[4]_{\geq 0}^h}\twonorm{q_{\sigma}}^2}^{1/2}\\
&=&\delta N_2\br{\mathbf{P}}N_2\br{\mathbf{Q}}\quad\quad\mbox{(Lemma~\ref{lem:randoperator})}.
\end{eqnarray*}
\end{proof}

\section{Multilinearization of random operators}\label{sec:multilinear}

	\begin{lemma}\label{lem:multiliniearization}
		Given $0\leq \rho<1$, integers $d,h,n>0$, a noisy EPR state $\psi_{AB}$ with the maximal correlation $\rho\defeq \rho\br{\psi_{AB}}$, there exists a map $f:L^2\br{\H_2^{\otimes h},\gamma_n}\rightarrow L^2\br{\H_2^{\otimes h},\gamma_{nt}}$ such that, for any degree-$d$ joint random operators
\[\br{\mathbf{P},\mathbf{Q}}=\br{\sum_{\sigma\in[4]_{\geq 0}^h}p_{\sigma}\br{\mathbf{g}}\A_{\sigma},\sum_{\sigma\in[4]_{\geq 0}^h}q_{\sigma}\br{\mathbf{h}}\B_{\sigma}}_{\br{\mathbf{g},\mathbf{h}}\sim\G_{\rho}^{\otimes n}}\in L^2\br{\H_2^{\otimes h},\gamma_n}\times L^2\br{\H_2^{\otimes h},\gamma_n},\]
where $\set{\A_i}_{i=0}^3$ and $\set{\B_i}_{i=0}^3$ are standard orthonormal basis in $\M_2$,
\begin{eqnarray*}
	&&\br{\mathbf{P}^{(1)},\mathbf{Q}^{(1)}}\defeq\br{f\br{\mathbf{P}},f\br{\mathbf{Q}}}\\
	&=&\br{\sum_{\sigma\in[4]_{\geq 0}^h}p^{(1)}_{\sigma}\br{\mathbf{x}}\A_{\sigma},\sum_{\sigma\in[4]_{\geq 0}^h}q^{(1)}_{\sigma}\br{\mathbf{y}}\B_{\sigma}}_{\br{\mathbf{x},\mathbf{y}}\sim\G_{\rho}^{\otimes n}}\in L^2\br{\H_2^{\otimes h},\gamma_{nt}}\times L^2\br{\H_2^{\otimes h},\gamma_{nt}}
\end{eqnarray*}
are multilinear joint random operators. It further holds that
\begin{enumerate}
  \item Both $\deg\br{\mathbf{P}^{(1)}}$ and $\deg\br{\mathbf{Q}^{(1)}}$ are at most $d$.
  \item For all $\br{i,j}\in[n]\times[t]$ and $\sigma\in[4]_{\geq 0}^h$,

  $\influence_{in+j}\br{p^{(1)}_{\sigma}}\leq\delta\cdot\influence_i\br{p_{\sigma}}~\mbox{and}~\influence_{in+j}\br{q^{(1)}_{\sigma}}\leq\delta\cdot\influence_i\br{q_{\sigma}};$
  \item $N_2\br{\mathbf{P}^{(1)}}\leq N_2\br{\mathbf{P}}~\mbox{and}~N_2\br{\mathbf{Q}^{(1)}}\leq N_2\br{\mathbf{Q}};$
  \item $\expec{}{\Tr~\mathbf{P}^{(1)}}=\expec{}{\Tr~\mathbf{P}}~\mbox{and}~\expec{}{\Tr~\mathbf{Q}^{(1)}}=\expec{}{\Tr~\mathbf{Q}};$
  \item $\abs{\expec{}{\Tr~\zeta\br{\mathbf{P}^{(1)}}}-\expec{}{\Tr~\zeta\br{\mathbf{P}}}}\leq \delta 2^{h+2}N_2\br{\mathbf{P}}^2$  and

      $\abs{\expec{}{\Tr~\zeta\br{\mathbf{Q}^{(1)}}}-\expec{}{\Tr~\zeta\br{\mathbf{Q}}}}\leq \delta2^{h+2}N_2\br{\mathbf{Q}}^2;$
  \item $\abs{\expec{}{\Tr~\br{\mathbf{P}^{(1)}\otimes\mathbf{Q}^{(1)}}\psi_{AB}^{\otimes h}}-\expec{}{\Tr~\br{\mathbf{P}\otimes\mathbf{Q}}\psi_{AB}^{\otimes h}}}\leq\delta N_2\br{\mathbf{P}}N_2\br{\mathbf{Q}}.$
\end{enumerate}
In particular, we may take $t=O\br{\frac{d^2}{\delta^2}}$.
	\end{lemma}

	\begin{definition}\label{def:linear}
		Suppose $f\in L^2\br{\reals^n,\gamma_n}$ is given by the Hermite expansion $f=\sum_{\mathbf{\sigma}\in\mathbb{Z}_{\geq 0}^n}\widehat{f}\br{\mathbf{\sigma}}H_{\mathbf{\sigma}}$. The {\em multilinear truncation} of $f$ is defined to be the function $f^{\text{ml}}\in L^2\br{\reals^n,\gamma_n}$ given by
		\[f^{\text{ml}}\defeq\sum_{\mathbf{\sigma}\in\set{0,1}^n}\widehat{f}\br{\mathbf{\sigma}}H_{\mathbf{\sigma}}\br{x}.\]
	\end{definition}
\begin{fact}~\cite{Ghazi:2018:DRP:3235586.3235614}\label{fac:mulilinear}
	Given parameters $\rho\in[0,1], \delta>0$ and $d\in\mathbb{Z}_{\geq 0}$, there exists $t=t\br{d,\delta}$ such that the following holds:
	
	Let $f,g\in L^2\br{\reals^n,\gamma_n}$ be degree-$d$ polynomials. There exist polynomials $\bar{f},\bar{g}\in L^2\br{\reals^{nt},\gamma_{nt}}$ over variables $\bar{x}\defeq\set{x_j^{\br{i}}:\br{i,j}\in[n]\times[t]}$ and $\bar{y}\defeq\set{y_j^{\br{i}}:\br{i,j}\in[n]\times[t]}$, respectively as
	\[\bar{f}\br{\bar{x}}\defeq f\br{x^{(1)},\ldots,x^{(n)}}~\text{and}~\bar{g}\br{\bar{y}}\defeq g\br{y^{(1)},\ldots,y^{\br{n}}}.\]
	Let $\bar{f}^{\text{ml}}$ and $\bar{g}^{\text{ml}}$ be the multilinear truncations of $\bar{f}$ and $\bar{g}$, respectively. Then the following holds.
	\begin{enumerate}
		\item $\bar{f}^{\text{ml}}$ and $\bar{g}^{\text{ml}}$ are multilinear with degree $d$.
		\item $\var{\bar{f}^{\text{ml}}}\leq\var{f}$ and $\var{\bar{g}^{\text{ml}}}\leq\var{g}$.
		\item $\twonorm{\bar{f}^{\text{ml}}}\leq\twonorm{\bar{f}}=\twonorm{f}$ and $\twonorm{\bar{g}^{\text{ml}}}\leq\twonorm{\bar{g}}=\twonorm{g}$.
\item Given two independent distributions $\mathbf{g}\sim \gamma_n$ and $\mathbf{x}\sim \gamma_{nt}$. The distributions of $f\br{\mathbf{g}}$ and $\bar{f}\br{\mathbf{x}}$ are identical. The distributions of $g\br{\mathbf{g}}$ and $\bar{g}\br{\mathbf{x}}$ are identical.

    \item $\twonorm{\bar{f}-\bar{f}^{\text{ml}}}\leq\frac{\delta}{2}\twonorm{f}$ and $\twonorm{\bar{g}-\bar{g}^{\text{ml}}}\leq\frac{\delta}{2}\twonorm{g}$
		\item For all $\br{i,j}\in[n]\times[t]$, it holds that $\influence_{x^{\br{i}}_j}\br{\bar{f}^{\text{ml}}}\leq\delta\cdot\influence_i\br{f}$ and $\influence_{y^{\br{i}}_j}\br{\bar{g}^{\text{ml}}}\leq\delta\cdot \influence_i\br{g}$.
\item $\widehat{f}\br{\mathbf{0}}=\widehat{\bar{f}}\br{\mathbf{0}}=\widehat{\bar{f}^{\text{ml}}}\br{\mathbf{0}}$
    and
    $\widehat{g}\br{\mathbf{0}}=\widehat{\bar{g}}\br{\mathbf{0}}=\widehat{\bar{g}^{\text{ml}}}\br{\mathbf{0}}$
		\item $\abs{\innerproduct{\bar{f}^{\text{ml}}}{\bar{g}^{\text{ml}}}_{\G^{\otimes nt}_{\rho}}-\innerproduct{f}{g}_{\G^{\otimes n}_{\rho}}}\leq\delta\twonorm{f}\twonorm{g}$.
	\end{enumerate}
In particular, we may take $t=O\br{\frac{d^2}{\delta^2}}.$
\end{fact}

	\begin{proof}[Proof of Lemma~\ref{lem:multiliniearization}]
	Applying Fact~\ref{fac:mulilinear} to $\set{p_{\sigma}}_{\sigma\in[4]_{\geq 0}^h}$ and $\set{q_{\sigma}}_{\sigma\in[4]_{\geq 0}^h}$ we get  $\set{\overline{p_{\sigma}}}_{\sigma\in[4]_{\geq 0}^h}$ and $\set{\overline{q_{\sigma}}}_{\sigma\in[4]_{\geq 0}^h}$.
Define joint random operators
\[\br{\overline{\mathbf{P}},\overline{\mathbf{Q}}}\defeq\br{\sum_{\sigma\in[4]_{\geq 0}^h}\overline{p_{\sigma}}\br{\mathbf{x}}\A_{\sigma},\sum_{\sigma\in[4]_{\geq 0}^h}\overline{q_{\sigma}}\br{\mathbf{y}}\B_{\sigma}}_{\br{\mathbf{x},\mathbf{y}}\sim\G_{\rho}^{\otimes n}}\in L^2\br{\H_2^{\otimes h},\gamma_{nt}}\times L^2\br{\H_2^{\otimes h},\gamma_{nt}}.\]
Let $p^{(1)}_{\sigma}\br{\cdot}\defeq\overline{p_{\sigma}}^{\text{ml}}\br{\cdot}$ and
$q^{(1)}_{\sigma}\br{\cdot}\defeq\overline{q_{\sigma}}^{\text{ml}}\br{\cdot}$.
Define
\[\br{\mathbf{P}^{(1)},\mathbf{Q}^{(1)}}\defeq\br{\sum_{\sigma\in[4]_{\geq 0}^h}p^{(1)}_{\sigma}\br{\mathbf{x}}\A_{\sigma},\sum_{\sigma\in[4]_{\geq 0}^h}q^{(1)}_{\sigma}\br{\mathbf{y}}\B_{\sigma}}_{\br{\mathbf{x},\mathbf{y}}\sim\G_{\rho}^{\otimes nt}}\]
Items 1 and item 2 are implied by Fact~\ref{fac:mulilinear} item 1 and item 5, respectively.
Item 3 is from Lemma~\ref{lem:randoperator} and the item 3 in Fact~\ref{fac:mulilinear}.
Item 4 follows from the fact that $\expec{}{\mathbf{P}}=2^h\widehat{p}_{\mathbf{0}}\br{\mathbf{0}}$ and
$\expec{}{\mathbf{Q}}=2^h\widehat{q}_{\mathbf{0}}\br{\mathbf{0}}$ and the item 7 in Fact~\ref{fac:mulilinear}.

We will prove the first inequality in item 5. The second one can be proved similarly. Define a convex body
\[\Delta\defeq\set{x\in\reals^{4^h}:0\leq\sum_{\sigma\in[4]_{\geq 0}^h}x_{\sigma}\A_{\sigma}\leq\id}\]
with the rounding map $\R$. Set $p\defeq\br{p_{\sigma}}_{\sigma\in[4]_{\geq 0}^h}, \bar{p}\defeq\br{\bar{p}_{\sigma}}_{\sigma\in[4]_{\geq 0}^h}$ and $p^{(1)}\defeq\br{p^{(1)}_{\sigma}}_{\sigma\in[4]_{\geq 0}^h}$ to be vector-valued functions. Then by Lemma~\ref{lem:closedelta},
\[\twonorm{p-\R\circ p}^2=\frac{1}{2^h}\expec{}{\Tr~\zeta\br{\mathbf{P}}}~\mbox{and}~\twonorm{p^{(1)}-\R\circ p^{(1)}}^2=\frac{1}{2^h}\expec{}{\Tr~\zeta\br{\mathbf{P}^{(1)}}}.\]
Hence
\begin{eqnarray*}
  &&\frac{1}{2^h}\abs{\expec{}{\Tr~\zeta\br{\mathbf{P}^{(1)}}}-\expec{}{\Tr~\zeta\br{\mathbf{P}}}}\\ \\
  &=&\abs{\twonorm{p^{(1)}-\R\circ p^{(1)}}^2-\twonorm{p-\R\circ p}^2}\\
  &=&\abs{\twonorm{\bar{p}^{\text{ml}}-\R\circ \bar{p}^{\text{ml}}}^2-\twonorm{\bar{p}-\R\circ \bar{p}}^2}\quad\quad\mbox{(Fact~\ref{fac:mulilinear} item 4)}\\
  &=&\abs{\br{\twonorm{\bar{p}^{\text{ml}}-\R\circ \bar{p}^{\text{ml}}}-\twonorm{\bar{p}-\R\circ \bar{p}}}\br{\twonorm{\bar{p}^{\text{ml}}-\R\circ \bar{p}^{\text{ml}}}+\twonorm{\bar{p}-\R\circ \bar{p}}}}\\
  &\leq&4\twonorm{p}\abs{\twonorm{\bar{p}^{\text{ml}}-\R\circ \bar{p}^{\text{ml}}}-\twonorm{\bar{p}-\R\circ \bar{p}}}\\
  &&\quad\quad\quad\quad\quad\quad\mbox{(Fact~\ref{fac:rounding}, Fact~\ref{fac:mulilinear} item 4 and $\R$ is a contraction map)}\\
  &\leq&4\twonorm{p}\br{\twonorm{\bar{p}-\bar{p}^{\text{ml}}}+\twonorm{\R\circ\bar{p}-\R\circ\bar{p}^{\text{ml}}}}\quad\quad\mbox{(Triangle inequality)}\\
  &\leq&8\twonorm{p}\twonorm{\bar{p}-\bar{p}^{\text{ml}}}\quad\quad\mbox{(Fact~\ref{fac:rounding} and $\R$ is a contraction map)}\\
  &\leq&4\delta\twonorm{p}^2\quad\quad\mbox{(Fact~\ref{fac:mulilinear} item 5)}\\
  &=&4\delta N_2\br{\mathbf{P}}^2\quad\quad\mbox{(Lemma~\ref{lem:randoperator})}.
 \end{eqnarray*}

To prove item 7, consider
\begin{eqnarray*}
&&\abs{\expec{}{\Tr\br{\mathbf{P}\otimes \mathbf{Q}}\psi_{AB}^{\otimes h}}-\expec{}{\Tr\br{\mathbf{P}^{\br{1}}\otimes \mathbf{Q}^{\br{1}}}\psi_{AB}^{\otimes h}}}\\
&=&\abs{\sum_{\sigma\in[4]_{\geq 0}^ h}c_{\sigma}\br{\innerproduct{p_{\sigma}}{q_{\sigma}}_{\G_{\rho}^{\otimes n}}-\innerproduct{p^{(1)}_{\sigma}}{q^{(1)}_{\sigma}}_{\G_{\rho}^{\otimes n}}}}\\
&\leq&\delta\sum_{\sigma\in[4]_{\geq 0}^{\otimes h}}\twonorm{p_{\sigma}}\twonorm{q_{\sigma}}\quad\quad\mbox{(Fact~\ref{fac:mulilinear} item 8)}\\
&\leq&\delta\br{\sum_{\sigma\in[4]_{\geq 0}^h}\twonorm{p_{\sigma}}^2}^{1/2}\br{\sum_{\sigma\in[4]_{\geq 0}^h}\twonorm{q_{\sigma}}^2}^{1/2}\\
&=&\delta N_2\br{\mathbf{P}}N_2\br{\mathbf{Q}}\quad\quad\mbox{(Lemma~\ref{lem:randoperator})}.
\end{eqnarray*}	
	\end{proof}
	\bibliographystyle{alpha}
	\bibliography{references}
	
	\appendix
	
	\section{Facts on Fr\'{e}chet derivative}\label{sec:frechet}
	In this section, we summarize some basic facts on Fr\'echet derivatives.
	\begin{fact}\label{fac:frechetderivative}
		Given $f_1,f_2,g:\M_d\rightarrow \M_d$ and $P,Q_1,\ldots, Q_k\in \M_d$, it holds that
		\begin{enumerate}
			\item $D\br{f_1+f_2}\br{P}\br{Q}=Df_1\br{P}\br{Q}+Df_2\br{P}\br{Q}$.
			
			\item $D\br{f_1\cdot f_2}\br{P}\br{Q}=Df_1\br{P}\br{Q}\cdot f_2\br{P}+f_1\br{P}\cdot Df_2\br{P}\br{Q}$.
			
			\item $D\br{g\circ f}\br{P}\br{Q}=\br{Dg\br{f\br{P}}\circ Df\br{P}}\br{Q}$.
			
			\item $D^kf\br{P}\br{Q_1,\ldots, Q_k}=D^kf\br{P}\br{Q_{\sigma\br{1}},\ldots, Q_{\sigma\br{k}}}$ for any integer $k>0$ and permutation $\sigma\in S_k$.	
		\end{enumerate}
	\end{fact}
	
	The following fact follows from elementary matrix calculations. Readers who are interested may refer to~\cite [Chapter X.4]{Bhatia}.
	\begin{fact}~\cite[Page 311, Example X.4.2]{Bhatia}
		\begin{itemize}
			\item Let $f\br{x}=x^2$. Then
			\[Df\br{P}\br{Q}=\anticommutator{P}{Q}.\]
			
			\item Let $f\br{x}=x^{-1}$. Then for any invertible $P$,
			\[Df\br{P}\br{Q}=-P^{-1}QP^{-1}.\]
		\end{itemize}
	\end{fact}
	\section{Facts on analysis}\label{sec:analysis}
	
	In this section, we list some basic results of matrix-valued functions. Most of the proofs are the direct generalization of the analysis of one-variable real functions. Here we include proofs for completeness.
	\begin{lemma}\label{lem:meanvalue}
		Suppose $f:[a,b]\rightarrow\H_d$ is a continuous mapping and is differentiable in $(a,b)$. Then there exists $x\in(a,b)$ such that
		\[\twonorm{f\br{a}-f\br{b}}\leq\br{b-a}\twonorm{f'\br{x}}\]
	\end{lemma}
	\begin{proof}
		Set $Z\defeq f\br{b}-f\br{a}$ and define
		\[\psi\br{t}\defeq\Tr~Zf\br{t},\]
		for $t\in[a,b]$. Then $\psi$ is a real-valued continuous function on $[a,b]$ which is differentiable in $(a,b)$. The mean value theorem shows that
		\[\psi\br{b}-\psi\br{a}=\br{b-a}\psi'\br{x}=\br{b-a}\Tr~Zf'\br{x},\]
		for some $x\in(a,b)$. On the other hand,
		\[\psi\br{b}-\psi\br{a}=\Tr\br{Z\cdot\br{f\br{b}-f\br{a}}}=\twonorm{Z}^2.\]
		Thus
		\[\twonorm{Z}^2=\br{b-a}\Tr~Zf'\br{x}\leq\br{b-a}\twonorm{Z}\twonorm{f'\br{x}}.\]
	\end{proof}
	\begin{lemma}\label{lem:limitderivative}
		Let $\set{f_n}_{n\in\mathbb{N}}$ be a sequence of functions mapping $\reals$ to $\H_d$, differentiable in $[a,b]$ and $\lim_{n\rightarrow\infty}f_n\br{x_0}=f\br{x_0}$ for some point $x_0\in[a,b]$. Suppose $\set{f_n'\br{x}}_{n\in\mathbb{N}}$ converges uniformly on $[a,b]$. Namely for any $\epsilon>0$, there exists $n_0=n\br{\epsilon}$ such that for any $m\geq n>n_0$ and $x\in[a,b]$, $\twonorm{f_n'\br{x}-f_m'\br{x}}\leq\epsilon.$ Then $\set{f_n}_{n\in\mathbb{N}}$ converges uniformly on $[a,b]$ to a function $f$ and
		\[\lim_{n\rightarrow\infty}f_n'\br{x}=f'\br{x}\quad\br{a<x<b}.\]
	\end{lemma}
	\begin{remark}
		Note that all $p$-norms of matrices are topologically equivalent. Thus the norm $\twonorm{\cdot}$ used in Lemma~\ref{lem:limitderivative} can be replaced by $\norm{\cdot}_p$ for any $1\leq p\leq+\infty$.
	\end{remark}
	\begin{proof}
		Let $\epsilon>0$ be given. Choose $N$ such that for any $m, n\geq N$,
		\begin{equation}\label{eqn:limdiffx0}
		\twonorm{f_n\br{x_0}-f_m\br{x_0}}<\epsilon/2,
		\end{equation}
		and for any $t\in[a,b]$
		\[\twonorm{f'_n\br{t}-f'_m\br{t}}<\frac{\epsilon}{2\br{b-a}}.\]
		Applying Lemma~\ref{lem:meanvalue} to $f_n-f_m$, we have
		\begin{equation}\label{eqn:limconvmn}
		\twonorm{f_n\br{x}-f_m\br{x}-f_n\br{t}+f_m\br{t}}\leq\frac{\abs{x-t}\epsilon}{2\br{b-a}}\leq\epsilon/2,
		\end{equation}
		for any $x,t$ on $[a,b]$ if $m, n\geq N$. Then from Eqs.~\eqref{eqn:limdiffx0}~\eqref{eqn:limconvmn}
		\begin{eqnarray*}
			&&\twonorm{f_n\br{x}-f_m\br{x}}\leq\twonorm{f_n\br{x}-f_m\br{x}-f_n\br{x_0}+f_m\br{x_0}}+\twonorm{f_n\br{x_0}-f_m\br{x_0}}\leq\epsilon,
		\end{eqnarray*}
		for any $m,n\geq N$ and $x\in[a,b]$.
		So $\set{f_n}_{n\in\mathbb{N}}$ converges uniformly on $[a,b]$. Let $f\br{x}\defeq\lim_{n\rightarrow\infty}f_n\br{x}$ for $x\in[a,b]$.
		
		Fix $x\in[a,b]$ and define
		\[\psi_n\br{t}\defeq\frac{f_n\br{t}-f_n\br{x}}{t-x},\quad \psi\br{t}\defeq\frac{f\br{t}-f\br{x}}{t-x},\]
		for $t\in[a,b], t\neq x$. Then Eq.~\eqref{eqn:limconvmn} implies that
		\[\twonorm{\psi_n\br{t}-\psi_m\br{t}}\leq\frac{\epsilon}{2\br{b-a}}.\]
		Thus $\set{\psi_n}_{n\in\mathbb{N}}$ uniformly converges for $t\neq x$. Note that $\lim_{x\rightarrow t}\psi_n\br{x}=f'_n\br{t}$. Thus
		\[\lim_{n\rightarrow\infty}\psi_n\br{t}=\psi\br{t},\]
		uniformly for $a\leq t\leq b, t\neq x$.
		Thus From Theorem 7.17 in~\cite{Rudin76}  , we conclude
		\[f'(x)=\lim_{t\rightarrow x}\psi\br{t}=\lim_{n\rightarrow\infty}f'_n\br{x}.\]
	\end{proof}

% \begin{fact}~\cite{Rudin76}\label{fac:Taylor}
%	Suppose $f$ is a real function on $[a,b]$, $n$ is a positive integer, $f^{\br{n-1}}$ is continuous on $[a,b]$ and $f^{\br{n}}\br{t}$ exists for every $t\in(a,b)$.Let $\alpha,\beta$ be distinct points of $[a,b]$, and define
%	\[P\br{t}\defeq\sum_{k=0}^{n-1}\frac{f^{\br{k}}\br{\alpha}}{k!}\br{t-\alpha}^k.\]
%	Then there exists a point $x$ between $\alpha$ and $\beta$ such that
%	\[f\br{\beta}=P\br{\beta}+\frac{f^{\br{n}}\br{x}}{n!}\br{\beta-\alpha}^n.\]
%\end{fact}

\begin{lemma}\label{lem:taylor}
	Let $f$ be a real function on $[a,b]$, $f^{(n-1)}$ is continuous on $[a,b]$ and $f^{\br{n}}\br{t}$ exists for all $t\in\br{a,b}$ except finite points $\set{t_1,\ldots, t_m}\subseteq(a,b)$. Moreover, $\abs{f^{\br{n}}\br{t}}\leq M$ for all $t\in\br{a,b}$ and $t\notin\set{t_1,\ldots, t_m}$. Then for any distinct points $\alpha,\beta$ in $[a,b]$, we have
	\[\abs{f\br{\beta}-P\br{\beta}}\leq\frac{M}{n!}\abs{\beta-\alpha}^n,\]
	where
\[P\br{t}\defeq\sum_{k=0}^{n-1}\frac{f^{\br{k}}\br{\alpha}}{k!}\br{t-\alpha}^k.\]
\end{lemma}

\begin{proof}
	Let $L$ be the number satisfying that
	\[f\br{\beta}=P\br{\beta}+\frac{L}{n!}\br{\beta-\alpha}^n.\]
	It suffices to show that $\abs{L}\leq M$. Set
	\[g\br{t}\defeq f\br{t}-P\br{t}-\frac{L}{n!}\br{t-\alpha}^n.\]
	We assume that $t_1<t_2<\ldots<t_m$, without loss of generality.
	Then
	\[g\br{\alpha}=g'\br{\alpha}=\ldots=g^{\br{n-1}}\br{\alpha}=0.\]
	Note that $g\br{\beta}=0$. By the mean value theorem, $g'\br{\beta_1}=0$ for some $\beta_1\in(\alpha,\beta)$. Repeat this for $n-1$ steps, we get $\beta_{n-1}\in\br{\alpha,\beta}$ such that $g^{\br{n-1}}\br{\beta_{n-1}}=0$.
	Note that
	\[g^{\br{n-1}}\br{t}=f^{\br{n-1}}\br{t}-f^{\br{n-1}}\br{\alpha}-L\br{t-\alpha}.\]
	Set $t_0=\alpha$.
	Let $i_0$ be the largest integer such that $t_{i_0}<\beta_{n-1}$. Then
	\[g^{\br{n-1}}\br{\beta_{n-1}}=\br{f^{\br{n-1}}\br{\beta_{n-1}}-f^{\br{n-1}}\br{t_{i_0}}}+\sum_{i=0}^{i_0-1}\br{f^{\br{n-1}}\br{t_{i+1}}-f^{\br{n-1}}\br{t_i}}-L\br{t-\alpha}.\]
	Applying the mean value theorem, we have
	\[g^{\br{n-1}}\br{\beta_{n-1}}=f^{\br{n}}\br{\xi_{i_0}}\br{\beta_{n-1}-t_{i_0}}+\sum_{i=0}^{i_0-1}f^{\br{n}}\br{\xi_i}\br{t_{i+1}-t_i}-L\br{\beta-\alpha},\]
	where $\xi_{i_0}\in[t_{i_0},\beta]$ and $\xi_i\in[t_i,t_{i+1}]$.
	As $g^{\br{n-1}}\br{\beta_{n-1}}=0$ and $\abs{f^{\br{n}}\br{t}}\leq M$ for any $t$ where $f^{\br{n}}\br{t}$ is defined, we have
	\[\abs{L}\br{\beta-\alpha}\leq \abs{M\br{\beta-\alpha}}.\]
	Thus $\abs{L}\leq M$.
\end{proof}

\section{Proofs in Section~\ref{sec:derivative}}\label{sec:zetataylor}

Before proving Lemma~\ref{lem:zetataylor} and Lemma~\ref{lem:zetaadditivity}, we first introduce Lyapunov equation, a well studied equation in control theory~\cite{doi:10.1080/00207179208934253}.

\begin{definition}\label{def:sylvester}
	Let $P,Q$ be two Hermitian matrices in $\H_d$. We define Lyapunov equation.
	\begin{equation}\label{eqn:lyapunov}
	PX+XP=Q.
	\end{equation}
	The solution to Eq.~\eqref{eqn:lyapunov} is denoted by $L\br{P,Q}$.
\end{definition}

\begin{lemma}\label{lem:lyapunovsol}
	Given Hermitian matrices $P, Q\in\H_d$, the Lyapunov equation ~\eqref{eqn:lyapunov} has an unique solution if and only if $P$ and $-P$ has no common eigenvalues. Namely, $I_d\otimes P+P\otimes I_d$ is invertible.
	
	Moreover, let $P=UDU^{\dagger}$ be a spectral decomposition of $P$, where $D=\textsf{Diag}\br{d_1,\ldots,d_n}$ satisfies that $d_i+d_j\neq 0$ for any $0\leq i, j\leq n$. Then Eq.~\eqref{eqn:lyapunov} has a unique solution $X_0$ and it satisfies that
	\[\br{U^{\dagger}X_0U}_{i,j}=\frac{\br{U^{\dagger}QU}_{i,j}}{d_i+d_j}.\]
\end{lemma}
\begin{proof}
	Let  $X'\defeq U^{\dagger}XU$ and $Q'\defeq U^{\dagger}QU$. Then we have
	\[DX'+X'D=Q',\]
	which is equivalent to
	\[\br{d_i+d_j}X'_{ij}=Q'_{ij},\]
	for $1\leq i,j\leq n$.
	Hence is has a unique solution if and only if  $d_i+d_j\neq 0$ for all $i, j$.
\end{proof}
% The following lemma is easy to verify.
%\begin{lemma}\label{lem:lyapunovtensor}
%	Given Hermitian matrices $P ,Q\in\H_d$, it holds that $L\br{\id\otimes P,\id\otimes Q}=\id\otimes L\br{P,Q}$.
%\end{lemma}
%\begin{fact}\label{fac:lyapunovsol}~\cite[Page 207, Theorem VII.2.5]{Bhatia}
%	Given Hermitian matrices $P ,Q\in\H_d$, if $\id_d\otimes P+P\otimes\id_d$ is invertible, then
%	\[L\br{P,Q}=\int_{-\infty}^{\infty}e^{-\mathrm{i}tP}Qe^{-\mathrm{i}tP}f\br{t}dt,\]
%	where $f$ is any function in $L^1\br{\reals}$ such that $\int_{-\infty}^{\infty}e^{-its}f(t)dt=\frac{1}{s}$ whenever $s=\lambda_i\br{P}+\lambda_j\br{P}$ for $1\leq i,j\leq d$.
%\end{fact}

\begin{fact}\label{fac:lysol2}~\cite[Page 205, Theorem VII.2.3]{Bhatia}
Let $P$ be a positive definite matrix. Then
\[L\br{P,Q}=\int_{0}^{\infty}e^{-tP}Qe^{-tP}dt.\]
\end{fact}

\begin{fact}\cite{SENDOV2007240}\label{fac:sendov}
Let $g$ be a $k$-times differentiable real-valued function defined on a set $I\subseteq\reals$ which is a union of a constant number of open intervals. Let $X\in\H_d$ have eigenvalues in $I$. Then the $k$-th order Fr\'echet derivative $\Tr~D^kf\br{X}\br{Y,\ldots, Y}$ exists for any $Y\in \H_d$.
\footnote{The theorem in~\cite{SENDOV2007240} is stated for real symmetric matrices and $I$ is an interval. But the proofs can be directly generalized to Hermitian matrices and a union of constant number of open intervals.}
\end{fact}

 \begin{definition}\label{def:hfunction}
	For any Hermitian matrices $P, Q$ that $P$ is invertible, we define
	\[\ell_Q\br{P}\defeq L\br{\abs{P},PQ+QP}.\]
\end{definition}
\noindent It is easy to verify that $\ell_Q\br{P}=Q$ if $P>0$.

\begin{definition}\label{def:kappa}
	For any Hermitian matrices $P$ and $Q$ and $P$ is invertible,
	\[\kappa_Q\br{P}\defeq\anticommutator{P}{\ell_Q\br{P}}=
	\anticommutator{P}{L\br{\abs{P},PQ+QP}}.\]
\end{definition}
%The follow lemma is implied by Lemma~\ref{lem:lyapunovtensor}.
%\begin{lemma}\label{lem:kappatensor}
%	It holds that
%	\[\kappa\br{\id\otimes P,\id\otimes Q}=\id\otimes\kappa_Q\br{P}.\]
%\end{lemma}

\begin{lemma}\label{lem:derivative}
	Let $P, Q$ be Hermitian matrices, where $P$ is invertible. The following holds.
	\begin{enumerate}
		\item Let $f\br{x}\defeq \sqrt{x}$ for $x\geq 0$. Then
		$Df\br{P}\br{Q}=L\br{\sqrt{P}, Q}$ if $P$ is positive definite.		
		\item Let $f\br{x}\defeq \abs{x}$. Then
		$Df\br{P}\br{Q}=\ell_Q\br{P}$.
		
		\item Let $f\br{x}=x\abs{x}$. Then    $Df\br{P}\br{Q}=\frac{1}{2}\br{\anticommutator{\abs{P}}{Q}+\kappa_Q\br{P}}.$
		
		\item Let $p\br{x}=\begin{cases}
		x^2~\mbox{if $x\geq 0$}\\
		0~\mbox{otherwise}.
		\end{cases}$
		Then 		\[Dp\br{P}\br{Q}=\frac{1}{2}\anticommutator{P}{Q}+\frac{1}{4}\anticommutator{\abs{P}}{Q}+\frac{1}{4}\kappa_Q\br{P}.\]
	\end{enumerate}
\end{lemma}

\begin{proof}
	\begin{enumerate}
		\item Let $g\br{x}\defeq x^2$ and $Df\br{P}\br{Q}=X$. Applying the composition rule in Fact~\ref{fac:frechetderivative}, we have
		\[Q=Dg\circ f\br{P}\br{Q}=\br{Dg\br{f\br{P}}\circ Df\br{P}}\br{Q}=Dg\br{\sqrt{P}}\br{X}=\anticommutator{\sqrt{P}}{X}.\]
		Hence $X=L\br{\sqrt{P}, Q}$.
		
		\item 	Let $g\br{x}=x^2$ and $h\br{x}=\sqrt{x}$. Then $f=h\circ g$.
		\[Df\br{P}\br{Q}=Dh\br{g\br{P}}\circ Dg\br{P}\br{Q}=Dh\br{P^2}\br{PQ+QP}=\ell_Q\br{P}.\]

		\item Let $g\br{x}=x$ and $h\br{x}=\abs{x}$. From Fact~\ref{fac:frechetderivative} item 2,
		\begin{align*}
		&Df\br{P}\br{Q}=Dg\br{P}\br{Q}h\br{P}+g\br{P}Dh\br{P}\br{Q}\\
		&=Q\abs{P}+P\ell_Q\br{P}.
		\end{align*}
		Using $f=h\cdot g$. we have
		\begin{align*}
		&Df\br{P}\br{Q}=Dg\br{P}\br{Q}h\br{P}+g\br{P}Dh\br{P}\br{Q}\\
		&=\abs{P}Q+\ell_Q\br{P}P
		\end{align*}
		Then
		\begin{align*}
		&Df\br{P}\br{Q}=\frac{1}{2}\br{\br{Q\abs{P}+P\ell_Q\br{P}}+\br{\abs{P}Q+\ell_Q\br{P}P}}\\
		&=\frac{1}{2}\br{\anticommutator{\abs{P}}{Q}+\kappa_Q\br{P}}.
		\end{align*}

		\item It follows from that $f\br{x}=\frac{1}{2}x^2+\frac{1}{2}x\abs{x}.$
	\end{enumerate}
\end{proof}
\begin{lemma}\label{lem:derivativeh}
	Let $P, Q$ be Hermitian matrices where $P$ is invertible. It holds that
	\begin{equation}\label{eqn:dh}
D\ell_Q\br{P}\br{Q}=L\br{\abs{P},2Q^2-2\ell_Q\br{P}^2}.
\end{equation}
	Moreover, if $P=\textsf{Diag}\br{a_1,\ldots,a_d}$ diagonal, then
	\begin{equation}\label{eqn:hb}
	\br{\ell_Q\br{P}}_{i,j}=\frac{Q_{ij}\br{a_i+a_j}}{\abs{a_i}+\abs{a_j}}.
	\end{equation}
	\begin{equation}\label{eqn:Ds}
	\br{D\ell_Q\br{P}\br{Q}}_{i,j}=2\frac{\sum_kQ_{ik}Q_{kj}\br{1-\frac{\br{a_i+a_k}\br{a_k+a_j}}{\br{\abs{a_i}+\abs{a_k}}\br{\abs{a_k}+\abs{a_j}}}}}{\abs{a_i}+\abs{a_j}}.
	\end{equation}

\end{lemma}
\begin{proof}
	From the definition of $\ell_Q\br{\cdot}$ in Definition~\ref{def:hfunction}, we have
	\[\abs{P}\ell_Q\br{P}+\ell_Q\br{P}\abs{P}=PQ+QP.\]
	Taking Fr\'echet derivative on both sides with respect to $Q$, we have
	\[\abs{P}D\ell_Q\br{P}\br{Q}+D\ell_Q\br{P}\br{Q}\abs{P}=2Q^2-2\ell_Q\br{P}^2.\]
We conclude Eq.~\eqref{eqn:dh}.

	%
%	Taking Fr\'echet derivative with respect to $Q$ again, we have
%	\[\abs{P}D^2\ell_Q\br{P}\br{Q}+D^2\ell_Q\br{P}\br{Q}\abs{P}=-3\br{D\ell_Q\br{P}\br{Q}\ell_Q\br{P}+\ell_Q\br{P}D\ell_Q\br{P}\br{Q}}.\]
%	Substituting $P=\textsf{Diag}\br{a_1,\ldots, a_n}$ and solving the resulting systems of linear equations, we get the desired results.
\end{proof}

%\begin{fact}\label{fac:ab2}~\cite{Yangfeng:2002}
%  For any $A,B\in\M_d$, it holds that
%  \[\abs{\Tr~(AB)^4}\leq\Tr(A^{\dagger}A)^2\br{B^{\dagger}B}^2.\]
%\end{fact}

\begin{lemma}\label{lem:aiaj}
	Given nonzero reals $a_1,\ldots, a_d$, let $M$ be a $d\times d$ Hermitian matrix defined to be $M_{ij}\defeq\frac{a_i+a_j}{\abs{a_i}+\abs{a_j}}$. For any $d\times d$ Hermitian matrix $A$, it holds that
	\[\twonorm{M\circ A}\leq \twonorm{A},\]
	and
	\[\norm{M\circ A}_4\leq c\norm{A}_4,\]
	where the $c\geq 1$ is an absolute constant.
\end{lemma}
\begin{proof}
	Note that $\twonorm{A}^2=\sum_{ij}\abs{A\br{i,j}}^2$. The first inequality follows from the fact that $\abs{M\br{i,j}}\leq 1$ for all $i,j$.
	To prove the second inequality, we may assume that $a_1,\ldots, a_s\geq 0$ and $a_{s+1},\ldots, a_d<0$ without loss of generality.
	Let $A=\begin{pmatrix}
	A_1 & A_2\\
	A_2^{\dagger} & A_3
	\end{pmatrix},$
	where $A_1,A_2,A_3$ are of size $s\times s$, $s\times (d-s)$ and $(d-s)\times (d-s)$, respectively. Let $M=\begin{pmatrix}
	M_1 & M_2\\
	M_2^{\dagger} & M_3
	\end{pmatrix}$ be the decomposition of the same size.
	Let $P$ be a $d\times d$ matrix defined to be
	\[P\br{i,j}\defeq\frac{\abs{a_i}-\abs{a_j}}{\abs{a_i}+\abs{a_j}}.\]
	Then
	\[\begin{pmatrix}
	0 & A_2\\
	A_2^{\dagger} & 0
	\end{pmatrix}\circ P=\begin{pmatrix}
	0 & A_2\circ M_2\\
	-A_2^{\dagger}\circ M_2^{\dagger} & 0
	\end{pmatrix}.\]
	Fact~\ref{fac:davis} implies that
	\[\norm{\begin{pmatrix}
		0 & A_2\circ M_2\\
		-A_2^{\dagger}\circ M_2^{\dagger} & 0
		\end{pmatrix}}_4\leq c \norm{\begin{pmatrix}
		0 & A_2\\
		A_2^{\dagger} & 0
		\end{pmatrix}}_4,\]
	for some absolute constant $c$, which implies that
	\[\norm{\begin{pmatrix}
		0 & A_2\circ M_2\\
		A_2^{\dagger}\circ M_2^{\dagger} & 0
		\end{pmatrix}}_4\leq c \norm{\begin{pmatrix}
		0 & A_2\\
		A_2^{\dagger} & 0
		\end{pmatrix}}_4.\]
	Then
	\begin{eqnarray*}
		&&\norm{A\circ M}_4\leq\norm{\begin{pmatrix}
				A_1 & 0\\
				0 & A_3
		\end{pmatrix}}_4+\norm{\begin{pmatrix}
		0 & A_2\circ M_2\\
		A_2^{\dagger}\circ M_2^{\dagger} & 0
	\end{pmatrix}}_4\leq \norm{\begin{pmatrix}
     		A_1 & 0\\
     		0 & A_3
     \end{pmatrix}}_4+ c \norm{\begin{pmatrix}
     0 & A_2\\
     A_2^{\dagger} & 0
 \end{pmatrix}}_4\\
&\leq& (c+1)\norm{A}_4,
	\end{eqnarray*}
where the last inequality is from the fact that
\[\norm{\begin{pmatrix}
	A_1 & 0\\
	0 & A_3
	\end{pmatrix}}_4\leq\norm{A}_4~\mbox{and}~\norm{\begin{pmatrix}
	0 & A_2\\
	A_2^{\dagger} & 0
	\end{pmatrix}}_4\leq\norm{A}_4.\]
\end{proof}
\begin{fact}~\cite{Davies:1988}[Corollary 5]\label{fac:davis}
	Given $a_1,\ldots, a_d,b_1,\ldots b_d>0$, let $M$ be a $d\times d$ matrix defined to be $M\br{i,j}\defeq\frac{a_i-b_j}{a_i+b_j}$. For any $d\times d$ matrix $A$, it holds that
	\[\norm{A\circ M}_4\leq c\norm{A}_4,\]
	for some absolute constant $c$.
\end{fact}
\begin{lemma}\label{lem:dlq3}
Let $P$ and $Q$ be Hermitian matrices where $P$ is invertible. It holds that
\[\twonorm{\ell_Q\br{P}}\leq\twonorm{Q},\]
and
\[\norm{\ell_Q\br{P}}_4\leq c\norm{Q}_4,\]
for some absolute constant $c\geq 1$.
\end{lemma}

\begin{proof}
We assume that $P=\mathsf{Diag}\br{a_1,\ldots, a_d}$ is a diagonal matrix without loss of generality. Define $M$ be a $d\times d$ matrix defined to be $M_{ij}\defeq\frac{a_i+a_j}{\abs{a_i}+\abs{a_j}}$ Then by Eq.~\eqref{eqn:hb} and Lemma~\ref{lem:aiaj},
\begin{eqnarray*}
  &&\twonorm{\ell_Q\br{P}} =\twonorm{Q\circ M}\leq\twonorm{Q};\\
  &&\norm{\ell_Q\br{P}}_4=\norm{Q\circ M}_4\leq c\norm{Q}_4
\end{eqnarray*}
\end{proof}

\begin{lemma}\label{lem:trlkappa}
  For any Hermitian matrices $P$ and $Q$, it holds that
  \begin{enumerate}
    \item $\Tr~\kappa_Q\br{P}=2\Tr~\abs{P}Q$.
    \item $\Tr~P\kappa_Q\br{P}=2\Tr~\abs{P}PQ$.
  \end{enumerate}
\end{lemma}
\begin{proof}
    Without loss of generality, we assume that $P=\mathsf{Diag}\br{a_1,\ldots, a_d}$ be a diagonal matrix. Then  $\kappa_Q\br{P}_{i,j}=\frac{Q_{ij}\br{a_i+a_j}}{\abs{a_i}+\abs{a_j}}$ from the definition of $\kappa_Q$. Thus have
  \[\Tr~\kappa_Q\br{P}=2\Tr~P\ell_Q\br{P}=2\sum_i\abs{a_i}Q_{ii}=2\Tr~\abs{P}Q.\]
  For the second equality, consider
  \[\Tr~P\kappa_Q\br{P}=2\Tr~P^2\ell_Q\br{P}=2\sum_i\abs{a_i}a_iQ_{ii}=2\Tr~\abs{P}PQ.\]
\end{proof}

Before proving Lemma~\ref{lem:zetataylor}, we need to compute the the first three orders of Fr\'echet derivatives of the function
\begin{equation}\label{eqn:qfunction}
  q\br{x}=\begin{cases}
	x^3~&\mbox{if $x\geq 0$}\\
	0~&\mbox{otherwise}.
	\end{cases}
\end{equation}

\begin{lemma}\label{lem:dq}
	Given an integer $d>0$ and $P,Q\in\H_d$, let $f(t)=\Tr~q\br{P+tQ}$. Then $f', f''$ exist on $\reals$ and $f'''$ exists except for finite points.
	
	Moreover, it holds that
	\begin{eqnarray}
		&&f'\br{0}=\Tr~\br{Qp\br{P}+P^2Q
		+P\abs{P}Q}; \label{eqn:dzetaprime}\\
&&f''\br{0}=\Tr~\br{4PQ^2+\frac{3}{2}\abs{P}Q^2+\frac{3}{4}Q\kappa_Q\br{P}};\label{eqn:dzetaprimetwo}\\
\end{eqnarray}
If $P$ is invertible, then
\begin{eqnarray}
		&&f'''\br{0}=\Tr\br{4Q^3+3Q^2\ell_Q\br{P}+\frac{3}{4}Q\anticommutator{P}{D\ell_Q\br{P}\br{Q}}}.\label{eqn:dzetaprimze3}
	\end{eqnarray}
\end{lemma}

\begin{proof}
	Note that $q',q''$ exist in $\reals$ and $q'''$ exists on $\reals\setminus\set{0}$. Thus $f'$ and $f''$ exist followed from Fact~\ref{fac:sendov}. If $P$ is invertible, $P+tQ$ has eigenvalue $0$ only for finite choices of $t$. Again applying Fact~\ref{fac:sendov}, $f'''\br{t}$ exists except for finite $t$'s.
	Note that $q\br{x}=xp\br{x}$, where $p\br{\cdot}$ is defined in Lemma~\ref{lem:derivative} item 4.
	\begin{eqnarray}
	&&\Tr~Dq\br{P}\br{Q}=\Tr~Qp\br{P}+\Tr~P^2Q
	+\frac{1}{2}P\abs{P}Q+\frac{1}{4}\Tr~P\kappa_Q\br{P}\nonumber\\
	&=&\Tr~Qp\br{P}+\Tr~P^2Q
	+\Tr~P\abs{P}Q\nonumber\\
	&\defeq&g_{1,Q}\br{P}+g_{2,Q}\br{P}+g_{3,Q}\br{P},\label{eqn:dq}
	\end{eqnarray}
	where the second equality is from Lemma~\ref{lem:trlkappa}.

	Further taking derivates of $g_{1,Q}$, $g_{2,Q}$, and  $g_{3,Q}$ we have	
	 \begin{eqnarray}
	&&Dg_{1,Q}\br{P}\br{Q}=\Tr\br{\frac{1}{2}Q\anticommutator{P}{Q}+\frac{1}{4}Q\anticommutator{\abs{P}}{Q}+\frac{1}{4}Q\kappa_Q\br{P}}\nonumber\\
	&=&\Tr~PQ^2+\frac{1}{2}\Tr~\abs{P}Q^2+\frac{1}{4}\Tr~Q\kappa_Q\br{P}.\label{eqn:dg1b}
	\end{eqnarray}

	\begin{eqnarray}
	&&Dg_{2,Q}\br{P}\br{Q}=2\Tr~PQ	^2;\label{eqn:dg2b}\\
	&&Dg_{3,Q}\br{P}\br{Q}=\Tr~\br{\abs{P}Q^2+QP\ell_Q\br{P}}.\nonumber
	\end{eqnarray}
	By symmetry,
	\[Dg_{3,Q}\br{P}\br{Q}=\Tr~\br{\abs{P}Q^2+Q\ell_Q\br{P}P}.\]
	Thus
	\begin{equation}\label{eqn:dg3b}
		Dg_{3,Q}\br{P}\br{Q}=\Tr~\br{\abs{P}Q^2+\frac{1}{2}Q\kappa_Q\br{P}}.
	\end{equation}
	
		Combining Eqs.~\eqref{eqn:dg1b}\eqref{eqn:dg2b}\eqref{eqn:dg3b} we conclude
		\begin{eqnarray}
			&&f''\br{t}=\Tr~\br{4PQ^2+\frac{3}{2}\abs{P}Q^2+\frac{3}{4}Q\kappa_Q\br{P}}\nonumber\\
			&\defeq&g_{4,Q}\br{P}+g_{5,Q}\br{P}+g_{6,Q}\br{P}. \label{eqn:g789}
		\end{eqnarray}
		Further taking the derivative on Eq.~\eqref{eqn:dg1b},
			\begin{equation}\label{eqn:a}
		Dg_{4,Q}\br{P}\br{Q}=4\Tr~Q^3.
		\end{equation}
		
		Applying Lemma~\ref{lem:derivative} item 2,
		
		\begin{equation}\label{eqn:b}
		Dg_{5,Q}\br{P}\br{Q}=\frac{3}{2}\Tr~\ell_Q\br{P}Q^2.
		\end{equation}
		
	From Definition~\ref{def:kappa}, 	
\begin{equation*}
  D\kappa_Q\br{P}\br{Q}=\anticommutator{Q}{\ell_Q\br{P}}+\anticommutator{P}{D\ell_Q\br{P}\br{Q}}
\end{equation*}
Thus
	\begin{equation} Dg_{6,Q}\br{P}\br{Q}=\Tr~\br{\frac{3}{2}Q^2\ell_Q\br{P}+\frac{3}{4}Q\anticommutator{P}{D\ell_Q\br{P}\br{Q}}}\label{eqn:dg4b}.
	\end{equation}

Combining Eqs.~\eqref{eqn:a}\eqref{eqn:b}\eqref{eqn:dg4b}, we conclude Eq.~\eqref{eqn:dzetaprimze3}.
	
\end{proof}

\begin{lemma}\label{lem:zetataylorbound}
Given an integer $d>0$ and $P,Q\in\H_d$ and $P$ is invertible, let $f(t)=\Tr~q\br{P+tQ}$. It holds that
\[\abs{f'''\br{0}}= c\twonorm{Q}\norm{Q}_4^2,\]
for some absolute constant $c$.
\end{lemma}
\begin{proof}
  We upper bound each term in Eq.~\eqref{eqn:dzetaprimze3}.
  For the first term, consider
  \begin{equation}\label{eqn:q3}
    \abs{4\Tr~Q^3}\leq4\twonorm{Q}\norm{Q^2}_2=4\twonorm{Q}\norm{Q}_4^2.
  \end{equation}
  For the second term,
  \begin{eqnarray}
   &&\abs{3\Tr~Q^2\ell_Q\br{P}}\leq3\norm{Q}_4^2\twonorm{\ell_Q\br{P}}\leq3\norm{Q}_4^2\twonorm{Q},\label{eqn:qlq}
  \end{eqnarray}
  where the second inequality is from Lemma~\ref{lem:dlq3}.
  For the final term, assuming that $P=\mathsf{Diag}\br{a_1,\ldots, a_d}$ is a diagonal matrix and applying Lemma~\ref{lem:derivativeh}, we have
  \begin{eqnarray}
  &&\abs{\frac{3}{4}\Tr Q\anticommutator{P}{D\ell_Q\br{P}\br{Q}}}\nonumber\\
  &=&\frac{3}{4}\abs{\sum_{ijk}Q_{ij}Q_{jk}Q_{ki}\br{\frac{a_i+a_j}{\abs{a_i}+\abs{a_j}}-\frac{\br{a_i+a_j}\br{a_j+a_k}\br{a_k+a_i}}{\br{\abs{a_i}+\abs{a_j}}\br{\abs{a_j}+\abs{a_k}}\br{\abs{a_k}+\abs{a_i}}}}}\nonumber\\
  &\leq&\frac{3}{4}\abs{\sum_{ijk}Q_{ij}Q_{jk}Q_{ki}\frac{a_i+a_j}{\abs{a_i}+\abs{a_j}}}+\frac{3}{4}\abs{\sum_{ijk}Q_{ij}Q_{jk}Q_{ki}\frac{\br{a_i+a_j}\br{a_j+a_k}\br{a_k+a_i}}{\br{\abs{a_i}+\abs{a_j}}\br{\abs{a_j}+\abs{a_k}}\br{\abs{a_k}+\abs{a_i}}}}\nonumber\\
  &=&\frac{3}{4}\abs{\Tr~\br{\ell_Q\br{P}Q^2}}+\frac{3}{4}\abs{\Tr~\ell_Q\br{P}^3}\nonumber\quad\quad\mbox{(Eq.~\eqref{eqn:hb})}\\
  &\leq&\frac{3}{4}\twonorm{\ell_Q\br{P}}\norm{Q}_4^2+\frac{3}{4}\twonorm{\ell_Q\br{P}}\norm{\ell_Q\br{P}}_4^2\nonumber\\
  &\leq&c\twonorm{Q}\norm{Q}_4^2\quad\quad\mbox{(Lemma~\ref{lem:dlq3})}.\label{eqn:qdpl}
  \end{eqnarray}
  Combining Eqs.~\eqref{eqn:q3}\eqref{eqn:qlq}\eqref{eqn:qdpl}, the result follows.

\end{proof}

It is now ready to prove Lemma~\ref{lem:zetataylor}.
\begin{proof}[Proof of Lemma~\ref{lem:zetataylor}]
 We assume $P$ is invertible. The general case follows by the continuity. Then $P+tQ$ is invertible except for the finite $t$'s.

From the definition of $\zeta_{\lambda}$, we have
 \begin{equation}\label{eqn:zetap}
 \zeta_{\lambda}\br{x}=x^2+\frac{\lambda^2}{3}-\frac{q\br{\lambda+x}}{6\lambda}+\frac{q\br{x-\lambda}}{6\lambda}+\frac{q\br{x-1+\lambda}}{6\lambda}-\frac{q\br{x-1-\lambda}}{6\lambda}.
 \end{equation}
Note that $q\br{\cdot}$ is the 1st order and 2nd order differentiable. $q'''\br{\cdot}$ exists except for finite points. Thus from Lemma~\ref{lem:taylor} and Fact~\ref{fac:sendov}, it suffices to upper bound $\Tr~D^3\zeta_{\lambda}\br{P}\br{Q}$, which is directly implied by Lemma~\ref{lem:zetataylorbound}.
\end{proof}

\begin{proof}[Proof of Lemma~\ref{lem:zetaadditivity}]
  Note that $\zeta\br{x}=p\br{x-1}+p\br{-x}$. Then from Item 4 of Lemma~\ref{lem:derivative}
  \begin{eqnarray*}
    &&\Tr~D\zeta\br{P}\br{Q}=\Tr\br{2P-I}Q+\frac{1}{2}\Tr~\br{\abs{P-I}-\abs{P}}Q+\frac{1}{4}\br{\kappa_Q\br
  {P-I}-\kappa_Q\br{-P}} \\
   &=&\Tr\br{2P-I}Q+\Tr~\br{\abs{P-I}-\abs{P}}Q,
  \end{eqnarray*}
  where the second equality is from Lemma~\ref{lem:trlkappa}.

  Assuming that $P=\mathrm{Diag}\br{a_1,\ldots,a_d}$ is a diagonal matrix, we have
  \[\abs{\Tr~D\zeta\br{P}\br{Q}}=\abs{\sum_i\br{2a_i-1+\abs{a_i-1}-\abs{a_i}}Q_{ii}}\leq4\sum_i\abs{a_iQ_{ii}}.\]
  Thus for any $P,Q\in\H_d$, there exists a unitary $U$ such that
  \[\abs{\Tr~D\zeta\br{P}\br{Q}}\leq4\Tr\abs{UPU^{\dagger}}\abs{UQU^{\dagger}}.\]
  Then by the mean value theorem,
  \[\abs{\Tr\br{\zeta\br{P+Q}-\zeta\br{P}}}=\br{Tr~D\zeta\br{P+\theta Q}\br{Q}}\leq4\Tr~\abs{U\br{P+\theta Q}U^{\dagger}}\abs{UQU^{\dagger}},\]
  for some $\theta\in[0,1]$ and unitary $U$.
  Moreover,
  \begin{eqnarray*}
  &&\Tr~\abs{U\br{P+\theta Q}U^{\dagger}}\abs{UQU^{\dagger}}\\
  &\leq&\twonorm{U\br{P+\theta Q}U^{\dagger}}\cdot\twonorm{UQU^{\dagger}}\\
  &=&\twonorm{P+\theta Q}\twonorm{Q}\\
  &\leq&\twonorm{P}\twonorm{Q}+\twonorm{Q}^2.
  \end{eqnarray*}
\end{proof}

\section{Proofs in Section~\ref{sec:invariance}}\label{sec:appinvariance}

\begin{proof}[Proof of Claim~\ref{claim:bc}]
		A crucial observation is that
	\begin{align}
	&\mathbf{A}=\id_{2^i}\otimes\id_2\otimes\widetilde{\mathbf{A}}\label{eqn:A}\\
	&\mathbf{B}=\id_{2^i}\otimes\widetilde{\mathbf{B}}\label{eqn:B}\\
	&\mathbf{C}=\id_{2^i}\otimes\widetilde{\mathbf{C}}.\label{eqn:C}
	\end{align}
	The both equalities can be proved by expanding the both sides. For the first equality,
	\begin{eqnarray*}
		&&\expec{}{\Tr~\mathbf{B}f\br{\mathbf{A}}}\\
		&=&2^i\expec{}{\Tr~\widetilde{\mathbf{B}}\br{\id_2\otimes f\br{\widetilde{\mathbf{A}}}}}\\
		&=&2^i\sum_{\sigma\in[4]^n}\widehat{M}\br{\sigma}\expec{}{\Tr~\br{\prod_{j=1}^{i}\mathbf{g}_{j,\sigma_j}}\P_{\sigma_{> i}}\br{\id_2\otimes f\br{\widetilde{\mathbf{A}}}}}\\
		&=&2^{i+1}\sum_{\sigma\in[4]^n:\sigma_{i+1}=0}\widehat{M}\br{\sigma}\expec{}{\Tr~\br{\prod_{j=1}^{i}\mathbf{g}_{j,\sigma_j}}\P_{\sigma_{> i+1}}f\br{\widetilde{\mathbf{A}}}}.
	\end{eqnarray*}
	
	And
	\begin{eqnarray*}
		&&\expec{}{\Tr~\mathbf{C}f\br{\mathbf{A}}}\\
		&=&2^i\expec{}{\Tr~\widetilde{\mathbf{C}}\br{\id_2\otimes f\br{\widetilde{\mathbf{A}}}}}\\
		&=&2^i\sum_{\sigma\in[4]^n}\widehat{M}\br{\sigma}\expec{}{\Tr~\br{\id_2\otimes\br{\prod_{j=1}^{i+1}\mathbf{g}_{j,\sigma_j}}\P_{\sigma_{> i+1}}}\br{\id_2\otimes f\br{\widetilde{\mathbf{A}}}}}\\
		&=&2^{i+1}\sum_{\sigma\in[4]^n:\sigma_{i+1}=0}\widehat{M}\br{\sigma}\expec{}{\Tr~\br{\prod_{j=1}^{i}\mathbf{g}_{j,\sigma_j}}\P_{\sigma_{> i+1}}f\br{\widetilde{\mathbf{A}}}}.
	\end{eqnarray*}
	For the second equality,
	\begin{eqnarray*}
		&&\expec{}{\Tr~\mathbf{B}f\br{\mathbf{A}}\mathbf{B}g\br{\mathbf{A}}}\\
		&=&2^i\expec{}{\Tr~\widetilde{\mathbf{B}}\br{\id_2\otimes f\br{\widetilde{\mathbf{A}}}}\widetilde{\mathbf{B}}\br{\id_2\otimes g\br{\widetilde{\mathbf{A}}}}}\\
		&=&2^i\sum_{\sigma,\tau\in[4]^n}\widehat{M}\br{\sigma}\widehat{M}\br{\tau}\expec{}{\Tr~\br{\prod_{j=1}^{i}\mathbf{g}_{j,\sigma_j}\mathbf{g}_{j,\tau_j}}\P_{\sigma_{> i}}\br{\id_2\otimes f\br{\widetilde{\mathbf{A}}}}P_{\tau_{> i}}\br{\id_2\otimes g\br{\widetilde{\mathbf{A}}}}}\\
		&=&2^{i+1}\sum_{\sigma,\tau\in[4]^n:\sigma_{i+1}=\tau_{i+1}}\widehat{M}\br{\sigma}\widehat{M}\br{\tau}\expec{}{\Tr~\br{\prod_{j=1}^{i}\mathbf{g}_{j,\sigma_j}\mathbf{g}_{j,\tau_j}}\P_{\sigma_{> i+1}} f\br{\widetilde{\mathbf{A}}}P_{\tau_{> i+1}} g\br{\widetilde{\mathbf{A}}}}.
	\end{eqnarray*}
	And
	\begin{eqnarray*}
		&&\expec{}{\Tr~\mathbf{C}f\br{\mathbf{A}}\mathbf{C}g\br{\mathbf{A}}}\\
		&=&2^i\expec{}{\Tr~\widetilde{\mathbf{C}}\br{\id_2\otimes f\br{\widetilde{\mathbf{A}}}}\widetilde{\mathbf{C}}\br{\id_2\otimes g\br{\widetilde{\mathbf{A}}}}}\\
		&=&2^i\sum_{\sigma,\tau\in[4]^n}\widehat{M}\br{\sigma}\widehat{M}\br{\tau}\expec{}{\Tr~\br{\br{\prod_{j=1}^{i+1}\mathbf{g}_{j,\sigma_j}\mathbf{g}_{j,\tau_j}}\br{\id_2\otimes \P_{\sigma_{> i+1}}}\br{\id_2\otimes f\br{\widetilde{\mathbf{A}}}}\atop\br{\id_2\otimes \P_{\tau_{> i+1}}}\br{\id_2\otimes g\br{\widetilde{\mathbf{A}}}}}}\\
		&=&2^{i+1}\sum_{\sigma,\tau\in[4]^n:\sigma_{i+1}=\tau_{i+1}}\widehat{M}\br{\sigma}\widehat{M}\br{\tau}\expec{}{\Tr~\br{\prod_{j=1}^{i}\mathbf{g}_{j,\sigma_j}\mathbf{g}_{j,\tau_j}}\P_{\sigma_{> i+1}} f\br{\widetilde{\mathbf{A}}}\P_{\tau_{> i+1}} g\br{\widetilde{\mathbf{A}}}}.
	\end{eqnarray*}
\end{proof}
\begin{proof}[Proof of Claim~\ref{claim:1}]
	
	To prove the first equality, from Eqs.~\eqref{eqn:dq}\eqref{eqn:zetap}, it suffices to show that
	\begin{equation}\label{eqn:tcase1}
	\expec{}{\Tr~t\br{\mathbf{A},\mathbf{B}}}=\expec{}{\Tr~t\br{\mathbf{A},\mathbf{C}}},
	\end{equation}
	for
	\begin{align*}
	t\br{A,B}\in\set{p\br{A}B,A^2B,A\abs{A}B},
	\end{align*}
	which directly follows by Eq.~\eqref{eqn:bc} in Claim~\ref{claim:bc}.
	
	To prove the second equality, from Eqs.~\eqref{eqn:zetap}\eqref{eqn:g789}, we first prove that
	
	\[\expec{}{\Tr~t\br{\mathbf{A},\mathbf{B}}}=\expec{}{\Tr~t\br{\mathbf{A},\mathbf{C}}}\]
	for
	\begin{align*}
	t\br{A,B}\in\set{AB^2,\abs{A}B^2,B\kappa_B\br{A}}.
	\end{align*}
	when $A$ is invertible. Then the second equality in Claim~\ref{claim:1} follows by the continuity of $D^2\zeta_{\lambda}\br{\cdot}$ due to Lemma~\ref{lem:dq} and Fact~\ref{fac:sendov}.
	
	The first two cases directly follow from Eq.~\eqref{eqn:bc2} in Claim~\ref{claim:bc}. To prove the final case, we use Fact~\ref{fac:lysol2},
	\begin{eqnarray*}
	&&\Tr~B\kappa_B\br{A}\\
	&=&\Tr~\br{AB+BA}\int_0^{\infty}e^{-t\abs{A}}\br{AB+BA}e^{-t\abs{A}}dt\\
	&=&2\int_0^{\infty}\Tr~\br{Ae^{-t\abs{A}}BAe^{-t\abs{A}}B+A^2e^{-t\abs{A}}Be^{-t\abs{A}}B}~dt
	\end{eqnarray*}
	 when $A$ is invertible. Then the result follows from Eq.~\eqref{eqn:bc2} in Claim~\ref{claim:bc}.
\end{proof}
\end{document}